\documentclass{article}
\usepackage{lmodern}
\usepackage[a4paper, top=1.5cm, bottom=2cm, left=2.5cm, right=1.5cm]{geometry}
\usepackage[english]{babel}
\usepackage{graphicx}
\usepackage{onlyamsmath}
\usepackage{amsfonts}
\usepackage{amssymb}
\usepackage{ntheorem}
\usepackage{dutchcal}
\usepackage{authblk}

\usepackage{soul}
\setstcolor{red}
\usepackage
	{todonotes}

\usepackage{color}

\newcommand{\e}{\varepsilon}
\newcommand{\se}{\sqrt{\varepsilon}}
\newcommand{\E}{\mathbb{E}}
\newcommand{\Pro}{\mathbb{P}}
\newcommand{\R}{\mathbb{R}}

\newcommand{\CalA}{\mathcal{A}}
\newcommand{\CalI}{\mathcal{I}}
\newcommand{\CalD}{\mathcal{D}}
\newcommand{\CalC}{\mathcal{C}}
\newcommand{\CalK}{\mathcal{K}}
\newcommand{\CalT}{\mathcal{T}}
\newcommand{\CalR}{\mathcal{R}}
\newcommand{\CalV}{\mathcal{V}}
\newcommand{\CalS}{\mathcal{S}}
\newcommand{\CalO}{\mathcal{O}}
\newcommand{\CalJ}{\mathcal{J}}
\newcommand{\CalL}{\mathcal{L}}
\newcommand{\CalF}{\mathcal{F}}

\newcommand{\CalU}{\mathcal{U}}

\newcommand{\CalX}{X}

\newcommand{\br}{\mathbf{r}}

\newcommand{\bx}{\mathbf{x}}
\newcommand{\by}{\mathbf{y}}
\newcommand{\bk}{\mathbf{k}}

\newcommand{\bq}{\mathbf{q}}
\newcommand{\bp}{\mathbf{p}}

\newcommand{\bC}{\mathbf{C}}

\newcommand{\bcals}{\mathbcal{s}}

\newcommand{\hpi}{\hat{u}^{\e,inc}_\omega}

\newcommand{\hpr}{\hat{u}^{\e,ref}_\omega}

\newcommand{\hpt}{\hat{u}^{\e,tr}_\omega}

\newcommand{\hatr}{\hat{a}^{\e,tr}_\omega}

\newcommand{\hbr}{\hat{b}^{\e,ref}_\omega}

\newcommand{\hp}{\hat{u}^\e_\omega}

\newcommand{\hau}{\hat{a}^\e_{0,\omega}}
\newcommand{\hbu}{\hat{b}^\e_{0,\omega}}

\newcommand{\lwu}{\bcals^\e_{0}}

\newcommand{\lwd}{\bcals^\e_{1}}
\newcommand{\ku}{k_{0,\omega}}

\newcommand{\1}{\mathbf{1}}

\newcommand{\ba}{\begin{eqnarray}}
\newcommand{\ea}{\end{eqnarray}}
\newcommand{\ban}{\begin{eqnarray*}}
\newcommand{\ean}{\end{eqnarray*}}

\newtheorem{theorem}{Theorem}[section]

\newtheorem{proposition}{Proposition}[section]
\newtheorem{corollary}{Corollary}[section]
\newtheorem{lemma}{Lemma}[section]

\theoremstyle{nonumberplain}
\theorembodyfont{\upshape}
\newtheorem{proof}{Proof}

\title{Generalized Snell's laws for rough interfaces}

\author[1]{Christophe Gomez}
\author[2]{Knut S{\o}lna}

\affil[1]{Aix Marseille Univ, CNRS, I2M, Marseille, France}
\affil[2]{UC Irvine, Department of Mathematics, Irvine, CA, USA}

\date{} 

\begin{document}

\maketitle

\begin{abstract} 
In this paper, we consider the reflection and transmission problem of waves by a rapidly oscillating rough interface that exhibits general mixing properties. Using an asymptotic analysis based on a separation of scales, corresponding to a paraxial (parabolic) scaling regime, we precisely characterize the specular and speckle (diffusive) components of the reflected and transmitted fields. A \emph{critically scaled} interface is considered, in the sense that the amplitudes of the interface fluctuations and the central wavelength are of the same order. When the correlation length of the interface fluctuations is of the same order as the beam width, random specular components arise in both the reflected and transmitted waves, while no speckle component is observed. Equivalently, the reflected and transmitted fields are essentially confined to the cones formed by the specular components (specular cones) with directions given by the classical Snell's law of reflection and refraction. When the correlation length is smaller than the beam width, a specular homogenization regime emerges. In this case, the rough interface can be approximated by an effective flat interface, yielding deterministic specular reflected and transmitted cones. However, broader cones containing the specular cones appear, within which the wavefields form speckle patterns (speckle cones) whose total energy is of leading order. We provide the two-point correlation functions of these speckle patterns and establish a central-limit-theorem-type result, showing that they can be modeled as Gaussian random fields. These results enable the identification of generalized Snell's laws of reflection and transmission, which depend on an effective scattering operator at the interface.
\end{abstract}

\begin{flushleft}
\end{flushleft}


\tableofcontents 

\section*{Introduction}

Wave scattering by rough surfaces lies at the heart of several branches of physics and engineering. Understanding these phenomena is crucial to a wide array of disciplines, including optics, solid-state physics, remote sensing, radar technology, environmental monitoring, communications, and non-destructive testing \cite{bass79, Darmon, ishimaru78, book, ogilvy91}. When waves impinge upon a rough surface, they undergo scattering due to irregularities in the surface profile, leading to changes in their direction, amplitude, and phase. The interactions between waves and rough surfaces give rise to a complex interplay of reflection, transmission, and diffraction. Consequently, the study of wave scattering by rough surfaces aims at unraveling the intricate mechanisms governing these interactions in order to predict and analyze wave behavior in practical scenarios. For radar systems, for instance, the theory of electromagnetic wave scattering from rough surfaces is essential to describe the effects of land or sea surface roughness. In optics, theories describing scattering from materials exhibiting angle-dependent reflectance play an important role. For acoustic waves, significant applications arise in ocean acoustic tomography or thermometry, where seabed and sea-surface roughness are critical, as well as in non-destructive testing using ultrasonic waves.

A number of works in the physical literature consider random interface problems in a perturbative approach, in which small interface variations give rise to relatively small corrections to the transmitted and reflected wavefields. In the present work, we model the interface fluctuations as a random field, which allows us to describe the transmitted and reflected fields even when they are modified at leading order by the presence of the random interface. Beam propagation involving reflection from, or transmission through, a rough interface requires specialized techniques that differ somewhat from those associated with bulk propagation \cite{fouque}.

\begin{figure}[tbh]
\centerline{\includegraphics[scale=0.3]{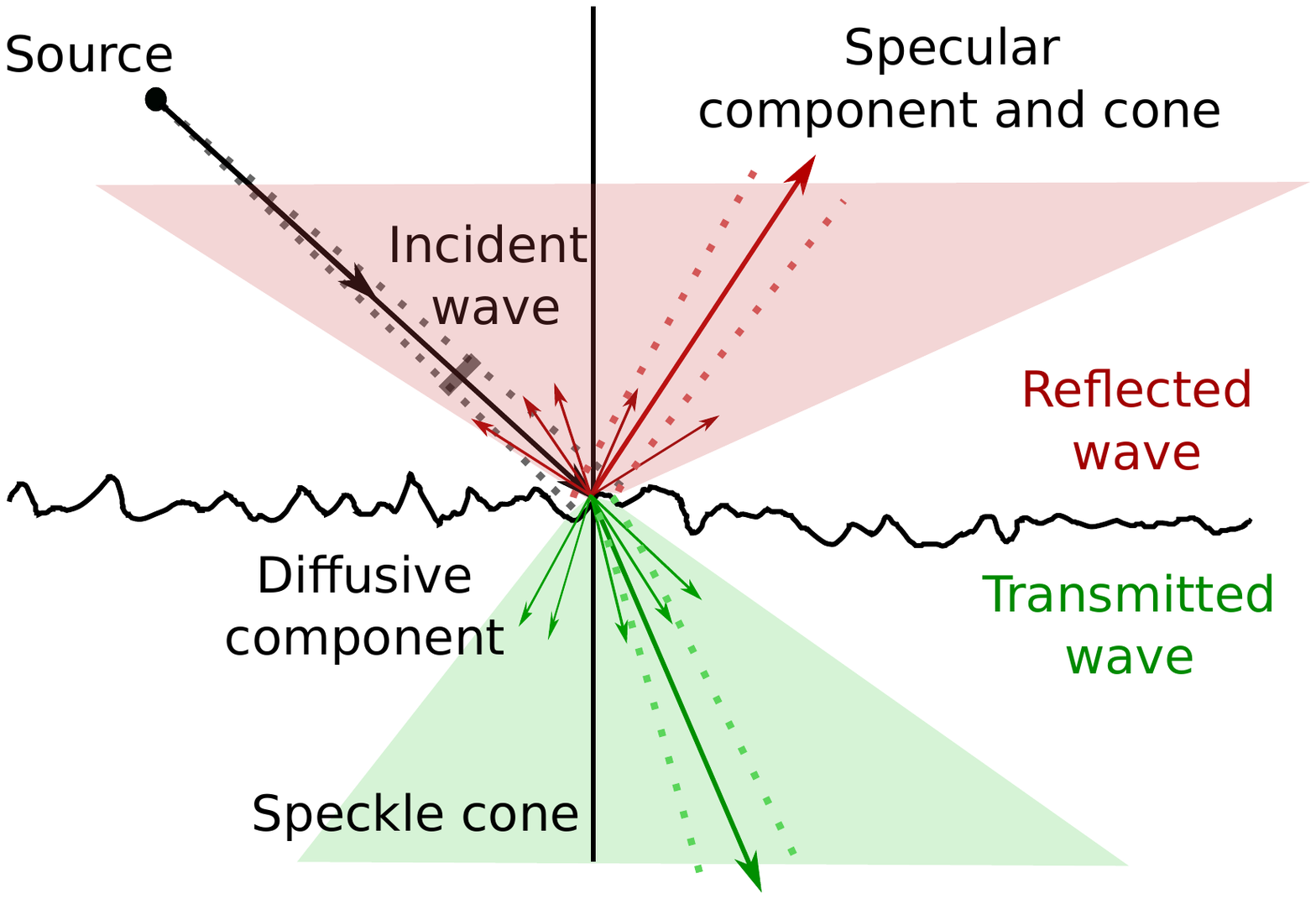} }
\caption{Illustration of the basic physical setup. A source illuminates (gray) a rough surface, producing a reflected wave (red) and a transmitted wave (green). The incident pulse is illustrated by a gray rectangle, and the gray dots represent the associated probing cone. Both the reflected and transmitted waves generally exhibit specular and diffusive components. The directions of the specular components, given by the classical Snell's law of reflection and refraction, correspond to the long arrows, and the associated specular cones are drawn with dotted lines. The diffusive components (speckles), whose directions are represented by small arrows, result from significant scattering of the incident wave by the rough interface. The reflected and transmitted speckle cones are illustrated as light red and green regions, respectively.}
\label{fig_0}
\end{figure}
The basic configuration we consider is illustrated in Figure~\ref{fig_0}. A beam-like wavefield, under the paraxial (parabolic) scaling, illuminates a rough interface separating two homogeneous media with different parameters above and below the interface. The incident wave gives rise to a reflected and a transmitted wave, each of which can be further decomposed into two main components: a specular component and a speckle component. The specular components correspond to the contributions associated with a flat interface and follow the classical Snell's laws. The speckle component, also referred to as the incoherent or diffuse component, is generated by scattering. The central questions we address are the following: (\emph{i}) how is the {\bf specular cone}  and wavefield therein modified by the presence of interface fluctuations? and (\emph{ii}) how can the \textbf{speckle cone} and the wavefield component therein be characterized? In particular, what determines the spatial support of the cone and the associated wavefield statistics? A key parameter distinguishing various canonical scattering regimes is the correlation length, or characteristic scale of variation, of the random interface relative to the beam width. In the critical scaling scenario considered here, the amplitude of the interface fluctuations is of the same order as the wavelength. We show in particular that when the interface spatial fluctuations occur on the scale of the beam width (or probing cone), the reflected and transmitted wave pulses exhibit random arrival-time properties. However, their main wave energy is confined to the cones formed by the specular components (specular cone, with directions given by the classical Snell's law of reflection and refractions), since the interface fluctuations do not induce strong coupling between modes with different lateral wavenumbers propagating in oblique directions. By contrast, when the interface spatial fluctuations occur on scales smaller than the width of the probing cone, such mode coupling becomes significant. This results in homogenized specular reflected and transmitted waves with frequency-dependent attenuation determined by the interface elevation statistics. The energy missing from these effective waves is transferred to broad speckle components, which exhibit Gaussian statistics and carry a total energy of relative order one. Although the interface roughness in our setting is not strong enough to generate enhanced backscattering effects \cite{ishimaru91}, the scattering operator we derive is similar to that obtained in the physical literature under the Born (single-scattering) approximation \cite{santenac09}, while no such approximation is made here. The detailed characterization of the speckle cones further enables the derivation of generalized Snell's laws of refraction and transmission, analogous to those reported in \cite{aieta, yu}. Overall, this work bridges two different approaches to wave scattering by rough interfaces within a unified framework, by rigorously deriving both the scattering operator at the interface and the associated generalized Snell's laws.

To summarize, our objective here is to characterize the specular and speckle cones and how the wave fields behave within them.  

There is a substantial literature addressing the important problem of wave scattering by rough surfaces or interfaces. Most of this literature deals with physically motivated expansions, such as perturbative approaches or Kirchhoff-type approximations, which impose strong conditions on the applicable scaling regimes \cite{ingve2,Darmon,TME,Sanchez,Li,book,ingve}. Integral equation formulations have also been extensively studied \cite{DeSanto,aa,shi}, in particular for the description of near-field scattering \cite{Sanchez2}. In addition, sophisticated numerical methods have been developed to analyze wave scattering from complex geometries \cite{rokhlin,bruno}. In this paper, we focus on situations in which such geometries are described only in a statistical sense. While homogenization techniques have been used to derive effective interface conditions for random interfaces \cite{gallas,nevard}, our emphasis here is also on the statistical characterization of the diffuse (speckle) components of the scattered fields. In our analysis, the total energy carried by the diffusive component may be comparable to that carried by the specular component.
 
Despite the importance of this problem and its wide range of applications, there exists only a limited number of rigorous results characterizing wavefields transmitted through or reflected from rough interfaces \cite{alonso12, bal99, gomez15, li, nevard}. In the high-frequency regime we consider a family of characteristic scaling regimes depending on the magnitude of the beam width relative to the correlation length of the interface fluctuations. A first technical challenge stems from limitations of the standard theory for stochastic partial differential equations (SPDEs) \cite{prato}. This difficulty arises because we assume stationarity of the interface height fluctuations, so that the noise term cannot be represented by a Hilbert--Schmidt covariance operator. From this perspective, the situation is similar to that considered in \cite{dawson1984random}, where well-posedness of the white-noise paraxial (It\^o--Schr\"odinger) equation is studied. However, in our setting the randomness is carried by the interface fluctuations rather than by bulk fluctuations, and we consider a scaling-limit regime rather than starting from a white-noise model. A second technical challenge arises from the need to describe not only the transmitted wavefield but also the reflected wave. This is achieved through an embedding-type formulation that parameterizes the wavefield in terms of a family of up- and down-propagating wave components. A third technical challenge concerns the nature of the random interface, which depends only on the lateral variable but not on the propagation (longitudinal) variable. As a consequence, the classical framework of diffusion approximation does not apply \cite{fouque}.

Our main objective is to develop a novel framework for the precise characterization of diffusely scattered waves (speckles) from a rough interface in terms of Gaussian random fields, together with the associated cones produced by these components, both in transmission and reflection (e.g.\ radar backscattering). Such analytical frameworks are promising for remote-sensing imaging problems, including (\emph{i}) the imaging of parameters of a rough (random) interface, and (\emph{ii}) the imaging of an object hidden behind the interface, for instance by computing the empirical spectrum of the reflected speckle field \cite{102,108,S2,S1}. This analysis can be further generalized to Synthetic Aperture Radar (SAR) imaging scenarios, in which the source cone moves and is typically mounted on an aircraft. In more general imaging contexts, one may aim to exploit the so-called memory effect, whereby the speckle pattern illuminating a hidden object is not completely altered but is instead shifted by a specific amount when the incident angle of the source is changed. We further note that, in the context of polarimetric imaging schemes, it is crucial to capture the coupling between different polarization modes at the interface and to understand how this coupling depends on the interface statistics \cite{bruno}. Generalizations to such imaging configurations, to the case of general hyperbolic systems \cite{85}, as well as to strongly fluctuating regimes in which the interface fluctuation amplitudes are large relative to the wavelength, will be addressed in future work.

The outline of the paper is as follows. First, in Section \ref{snell}, we summarize the main results of the paper, in particular the generalized Snell's laws. In Section \ref{sect1P2}, we describe the wave propagation scenario and the scaling regime under consideration. In this section, we also discuss mixing properties in the statistical modeling of the interface fluctuations. Section \ref{sec:flat_interface} presents the basic decomposition of the wavefield into lateral Fourier modes (lateral wavenumbers) and into up- and down-propagating components in the context of a flat interface. This decomposition allows us to derive the reflection and transmission conditions at the flat interface, which couple up- and down-propagating wave modes for each lateral wavenumber. In Section \ref{sec:randint}, this description is extended to the case of a random interface. This leads to the identification of a random scattering operator at the interface, capturing the coupling between up- and down-propagating wave modes at different lateral wavenumbers induced by the random interface fluctuations. 
 In Section \ref{sec:specint}, we discuss the regime in which the interface fluctuations decorrelate on a scale comparable to the width of the probing cone, leading to specular cones that carry essentially all of the energy. In Sections \ref{sec:randint1}--\ref{sec:randint3}, we consider the case where the interface fluctuations decorrelate on a scale that is small relative to the width of the probing cone, but large compared to the wavelength. In this regime, the interface generates both specular cones and broader speckle cones, leading to the identification of generalized Snell's laws. Finally, Sections \ref{sec:proof1} and \ref{sec:proof2}, together with Appendices \ref{sec:proof_continuity}, \ref{proof:lem_mix}, and \ref{sec:proof_lem_tech}, contain detailed proofs of the technical results.

\section{Summary of results and Generalized Snell's laws} \label{snell}  

The proposed analytical framework allows us to derive, from first principles of physics, generalized Snell's laws of reflection and refraction for rough interfaces, as illustrated in Figure~\ref{fig:ref_intro}. We consider a three-dimensional propagation medium, in which the reflection and transmission angles are defined with respect to the axis normal to the rough interface and are therefore nonnegative. In this setting, the incident wave may be reflected back toward the source location, 
and also transmitted through the interface, and with the scattering depending on the roughness of the interface.

Consider first the specular reflection and transmission cones illustrated in Figure \ref{fig_0}. These cones obey the classical Snell's law, with a reflection angle equal to the incident angle $\theta_{inc}$ and a transmission angle $\theta^0_{tr}$ defined as in the flat-interface case (see Figure \ref{fig:ref_intro}). If the interface fluctuations occur on a spatial scale that is small compared to the beam width, a homogenization result applies and the effect of the random fluctuations can be represented by a frequency-dependent damping factor. This damping factor is defined in terms of the marginal distribution of the interface height fluctuations. For a pulsed source, this results in a modification through convolution with a pulse-shaping function. When the interface fluctuations occur on the scale of the beam width, the leading-order field consists only of specular components, which become random. The randomness arises through a phase term reflecting the perturbation of the interface in the integral representation of the transmitted and reflected wavefields. These results are discussed in more detail in Sections \ref{sec:specint} and \ref{sec:randint1}

We next introduce the speckle component illustrated in Figure \ref{fig_0}. 
To motivate our main result regarding the generalized Snell's law, we first consider the case of a standard diffraction grating.  Consider an incoming plane wave with wavelength $\lambda$ propagating in a horizontal direction normal to the diffraction grooves and incident angle $\theta_{inc}$ relative to the grating normal. For $d$ the uniform spatial period of the grating, the reflected diffracted field exhibits maxima at diffraction angles $\theta_{ref}(m)$ given by the grating equation
\begin{equation}\label{grating}
\sin(\theta_{ref}(m)) = \sin(\theta_{inc}) + \frac{m\lambda}{d},
\qquad m \in \mathbb{Z}.
\end{equation}
We now ask how this picture is modified when the interface is not a perfect periodic grating but rather a rough random interface with isotropic random height fluctuations. The interface separates two half-spaces with propagation speeds $c_0$ and $c_1$, respectively (see Figure~\ref{fig:ref_intro}). In this setting the spatial scale $d$ is \emph{replaced} by $\ell_c$, the correlation length characterizing the scale of variation of the interface fluctuations. Moreover, instead of a discrete set of characteristic reflection angles, we obtain a continuum of reflected angles giving rise to a speckle cone. In this framework the index $m$ is \emph{replaced} by a relative scattering slowness vector $\bp \in \R^2$, whose distribution is described by the following effective frequency-dependent measure
\begin{equation}\label{def:A_intro}
\CalA(v,\omega,\bp) :=
\int \E\Big[ e^{i\omega v ( V(\by) - V(0)) } \Big] 
e^{-i\omega \bp \cdot\by} \, d\by,
\end{equation}
with $v = 2 \cos(\theta_{inc})/c_0$. According to the Bochner--Khinchin theorem, $\CalA$ defines a finite positive measure on $\R^2$ with respect to $\bp$, which we refer to as the \emph{scattering distribution} in this paper. For notational simplicity, even though $\CalA$ is a measure in $\bp$, we still use the notation $\CalA(v,\omega,\bp)$ throughout the paper. This measure involves the stationary random interface elevation $V$ through the characteristic function of the relative elevation $V(\by)-V(0)$. Note that in our formulation the scattering operator depends explicitly on the angular frequency $\omega$. This dependence captures the interaction between the frequency content of the pulse and the relative elevation of the interface.

A distribution of this form was previously derived in the physical literature \cite{santenac09} under the Born (single-scattering) approximation to describe the mean intensity of the diffusely reflected wavefield. Here we recover the same scattering operator without invoking this approximation. For a flat interface the distribution of $\bp$ is supported at the origin. In general, $\bp$ is determined by the observation point relative to the specular cone, and the intensity of the reflected field is proportional to $\CalA$. The position at which the speckle is evaluated relative to the specular component is given by \eqref{eq:off} below.

 \begin{figure}
\centerline{\includegraphics[scale=0.27]{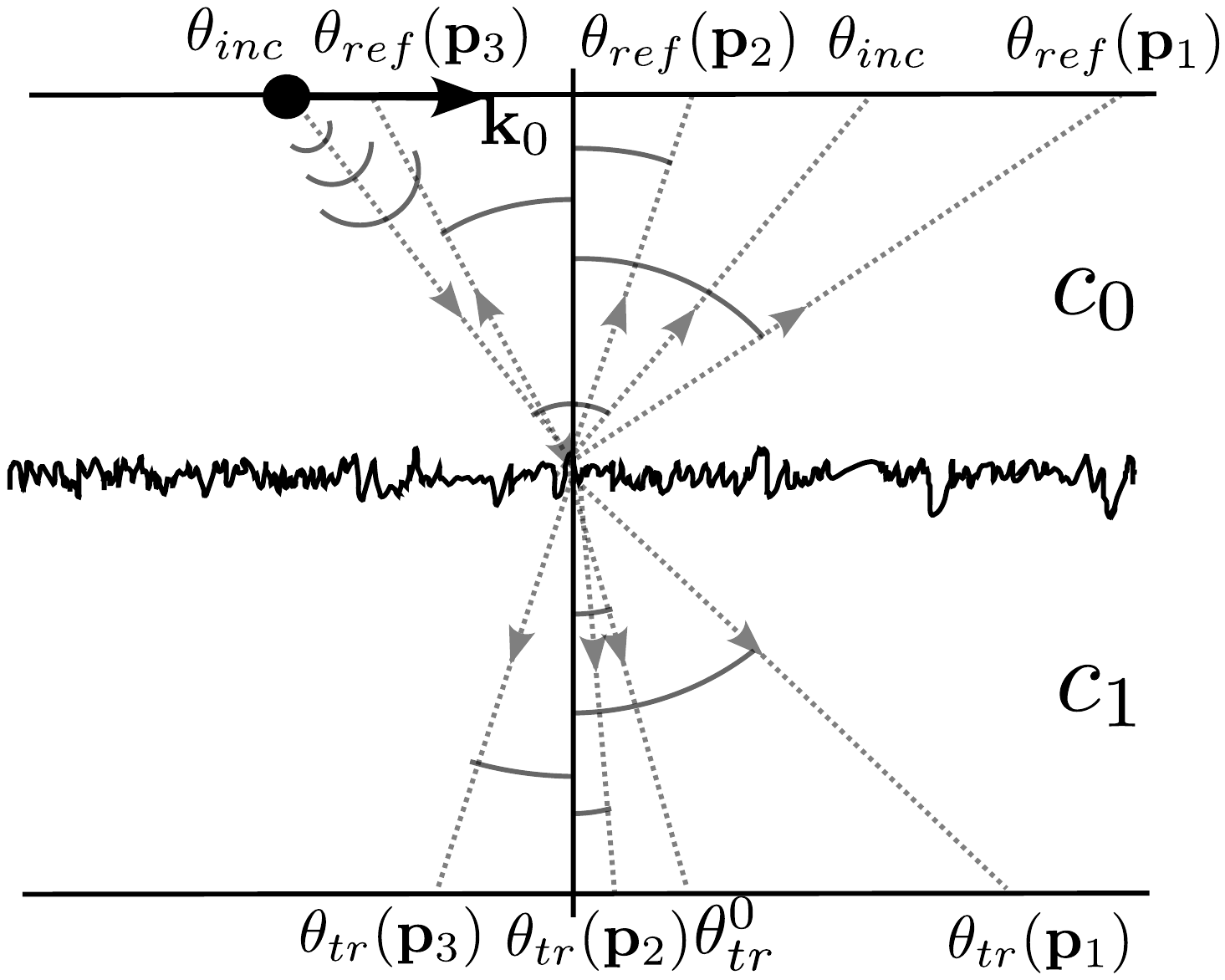} }
\caption{Illustration of the generalized Snell's laws of reflection and refraction. The reflection and transmission angles are parameterized by  the scattering slowness vector $\bp$.   }
\label{fig:ref_intro}
\end{figure}

We next introduce the generalized Snell’s law that associates a value of $\bp$ with a reflection angle. As remarked above, the reflected intensity in this direction depends on the scattering distribution $\CalA$.
Consider  an incident wave with angular frequency $\omega$ and  
 two-dimensional slowness vector $\bk_0$ so that 
 the  incident angle $\theta_{inc} > 0$ is given by
\[
  \sin( \theta_{inc} )   = c_0 |\bk_0| .
 \]
 see (\ref{eq:snell0}) below.
 The reflection angle $\theta_{ref}$ at scattering slowness vector  $\bp$  is then given by
\begin{equation}\label{eq:ref_formula_intro}
\frac{\sin(\theta_{ref}(\bp))}{c_0} = \frac{\sin(\theta_{inc})}{c_0} +  |\bk_0| \CalC(\bp;\bk_0,\theta_{inc},\theta_{inc}) ,  
\end{equation} 
with the speckle  modulation factor being 
\[
\CalC(\bp;\bk,\theta,\phi) =  \sqrt{\frac{\Xi(\bp;\bk,\theta,\phi)}{1+\sin^2(\theta)\Big( \Xi(\bp;\bk,\theta,\phi)-1\Big)}}-1 
\]
where
\[
\Xi(\bp;\bk,\theta,\phi) = \Big(1+\xi(\phi)\frac{\bp\cdot\bk}{\cos^2(\theta)}\Big)^2 + \xi^2(\phi)(\bp\cdot \bk^\perp)^2, \qquad 
\text{with}\qquad \xi(\phi)  =\frac{r^2_0}{\ell_c L}\frac{c_0^2}{\sin^2(\phi)},
\]
and with $\bk_0^\perp = (-\bk_{0,2}, \bk_{0,1})^T$.  
Here, $r_0$ denotes the beam width generated by the source term, $\ell_c$ the characteristic correlation length of the interface fluctuations, and $L$ is the typical propagation distance from the source plane to the interface. Note that for a flat interface with scattering distribution $\CalA$ supported in $\bp$ only at the origin, we obtain the classic Snell's law in refraction since $\CalC({\bf 0};\bk,\theta) =  0 $.

We next consider transmission through the interface and  the refraction angle. The generalized Snell's law for transmission reads
\begin{equation}\label{eq:tr_formula_intro}
\frac{\sin(\theta_{tr}(\bp))}{c_1} = \frac{\sin(\theta_{inc})}{c_0}+|\bk_0| \CalC(\bp; \bk_0, \theta^0_{tr},\theta_{inc}) . 
\end{equation}
Here, $\theta^0_{tr}$ denotes the specular refraction angle for a flat interface and is given by the standard Snell's law for refraction:
\[
\frac{\sin(\theta^0_{tr})}{c_1}=\frac{\sin(\theta_{inc})}{c_0}.
\]

The scattering distribution   $\CalA(v,\omega,\bp)$
in \eqref{def:A_intro} is here evaluated at  $v=\cos(\theta_{inc})/c_0 - \cos(\theta^0_{tr})/c_1$. In both relations \eqref{eq:ref_formula_intro} and \eqref{eq:tr_formula_intro}, the key parameter characterizing the size of the speckle (diffusive) cones, as well as the magnitude of the corrections relative to the classical Snell's laws, is the ratio $r_0^2/(\ell_c L)$. This parameter is analogous to the inverse of the Fresnel number with $\ell_c$ plays the role of the propagation distance in the usual definition of this number. In the paraxial regime considered here, 
we assume for  $\lambda$ is the central wavelength of the source term: 
\[
\frac{r^2_0}{\lambda} \sim L,
\]
corresponding to a Rayleigh length of the order of the typical propagation distance $L$. The Rayleigh length corresponds to the distance from the beam waist to the location where the cross-sectional area of the beam is doubled. In other words, in this scaling, the beam does not diverge significantly and remains mainly collimated over the typical propagation distance. In our situation, this scaling implies
\[
\frac{r^2_0}{\ell_c L} \sim \frac{\lambda }{ \ell_c }.
\] 
The  latter ratio characterizes the relative roughness of the interface. 
The interface is considered rough when $\lambda \sim \ell_c$, whereas the interface appears smooth to the wave when $\lambda \ll \ell_c$. In what follows the parameter $\lambda_c/\ell_c$ is referred to as the \emph{relative roughness parameter}, which is an analogous of the inverse Fresnel-type number $r^2_0/(\ell_c L)$.
 

The corrections in \eqref{eq:ref_formula_intro} and \eqref{eq:tr_formula_intro} are similar to those obtained in the physical literature \cite{aieta, yu} for relatively {\it smooth}  interfaces:
\begin{equation}\label{eq:correction}
|\bk_0| \CalC_{ref}(\bp)  \underset{\lambda\ll \ell_c}{\simeq} \frac{1}{n_0k_0} \frac{2\pi \bp \cdot \widehat \bk_0}{ c_0 \ell_c \cos^2(\theta_{inc})}
\qquad\text{and}
\qquad
|\bk_0| \CalC_{tr}(\bp) \underset{\lambda\ll \ell_c}{\simeq} \frac{1}{k_0} \frac{2 \pi \bp \cdot \widehat \bk_0}{\ell_c \cos^2(\theta^0_{tr})},
\end{equation}
with wavenumber $k_0 = 2\pi/\lambda$, refractive index $n_0 = 1/c_0$, and initial lateral propagation direction $\widehat{\bk}_0 = \bk_0/|\bk_0|$.
In contrast to \cite{aieta, yu}, where these corrections are expressed in terms of space-dependent phase shifts, they are here formulated in a probabilistic setting in terms of the scattering direction $\bp$, distributed according to the effective scattering distribution $\CalA$ defined by \eqref{def:A_intro} as it can be recast as a probability measure after as appropriate normalization.
From the generalized Snell's relation \eqref{eq:ref_formula_intro}, it follows that for  a relatively smooth interface ($\lambda\ll \ell_c$) the reflection angle can be approximated by
\begin{equation}\label{eq:correction_angle_ref}
\theta_{ref}(\bp) \underset{\lambda\ll \ell_c}{\simeq}  \theta_{inc} + \Big(\frac{\lambda}{ \ell_c} \Big) \frac{c_0 \bp\cdot \widehat \bk_0}{\cos^3(\theta_{inc})} 
   \end{equation}
yielding small deviations from the specular reflection   angle  of order $\lambda/\ell_c$. 
By comparison with diffraction grating \eqref{grating}, which in the relatively smooth case gives a relation of the form
\[ 
  \sin( \theta_{ref}(m) )  \simeq  \sin(\theta_{inc}) + \frac{m\lambda}{ \ell_c} \qquad\text{with}\qquad m=\frac{c_0 \bp\cdot \widehat \bk_0}{\cos^2(\theta_{inc})},
\]
we see that the parameter $d$ (characterizing the periodicity) in the correction is replaced by $\ell_c$ (the correlation length) describing the scale of variation of the interface fluctuations.

Similarly we have in view of  \eqref{eq:tr_formula_intro}  that  the   refraction angle can be approximated by 
\begin{equation}\label{eq:correction2}
 \theta_{tr}(\bp) \underset{\lambda\ll \ell_c}{\simeq}  \theta^0_{tr} +\Big(\frac{\lambda}{\ell_c}\Big)\frac{c_1 \bp\cdot \widehat \bk_0}{\cos^3(\theta^0_{tr})} ,
\end{equation}

Note also that  for normal incidence $\theta_{inc}=0$, these relations reduce to 
\[
\theta_{ref}(\bp) = \arctan\Big(\frac{\lambda c_0 |\bp|}{\ell_c }\Big)\qquad\text{and}\qquad \theta_{tr}(\bp) = \arctan\Big(\frac{\lambda c_1|\bp|}{ \ell_c }\Big);
\]
see (\ref{eq:tan1})  and (\ref{eq:tant}).

Finally, in the case of a source emitting a short pulse, the transmitted and reflected speckle (diffusive) components are shown to be {\it space--time Gaussian random fields}, whose covariance functions are characterized explicitly through the scattering distribution $\CalA$. These covariance functions show  that, when observed along horizontal planes, the random fields form ellipses characterized by appropriate dispersion matrices, which evolve in time (see Figure~\ref{fig:speckle_intro} for an illustration).

To summarize,  the important parameter governing  the speckle reflection and refraction angles is the inverse Fresnel-type number $r_0^2/\ell_c L$, of order the relative roughness parameter $\lambda/\ell_c$, while the important quantity governing  the statistical structure of the transmitted and reflected speckle fields is the  scattering distribution $\CalA$. 

All the analysis presented in this paper is carried out for a time-dependent source emitting a short pulse, so that the source term has a broad bandwidth. The time-harmonic case, where the source contains a single frequency, will be investigated in future work. We note, however, that the generalized Snell's law derived here remains valid in the time-harmonic regime.

\begin{figure}
\begin{center}
\includegraphics*[scale=0.25]{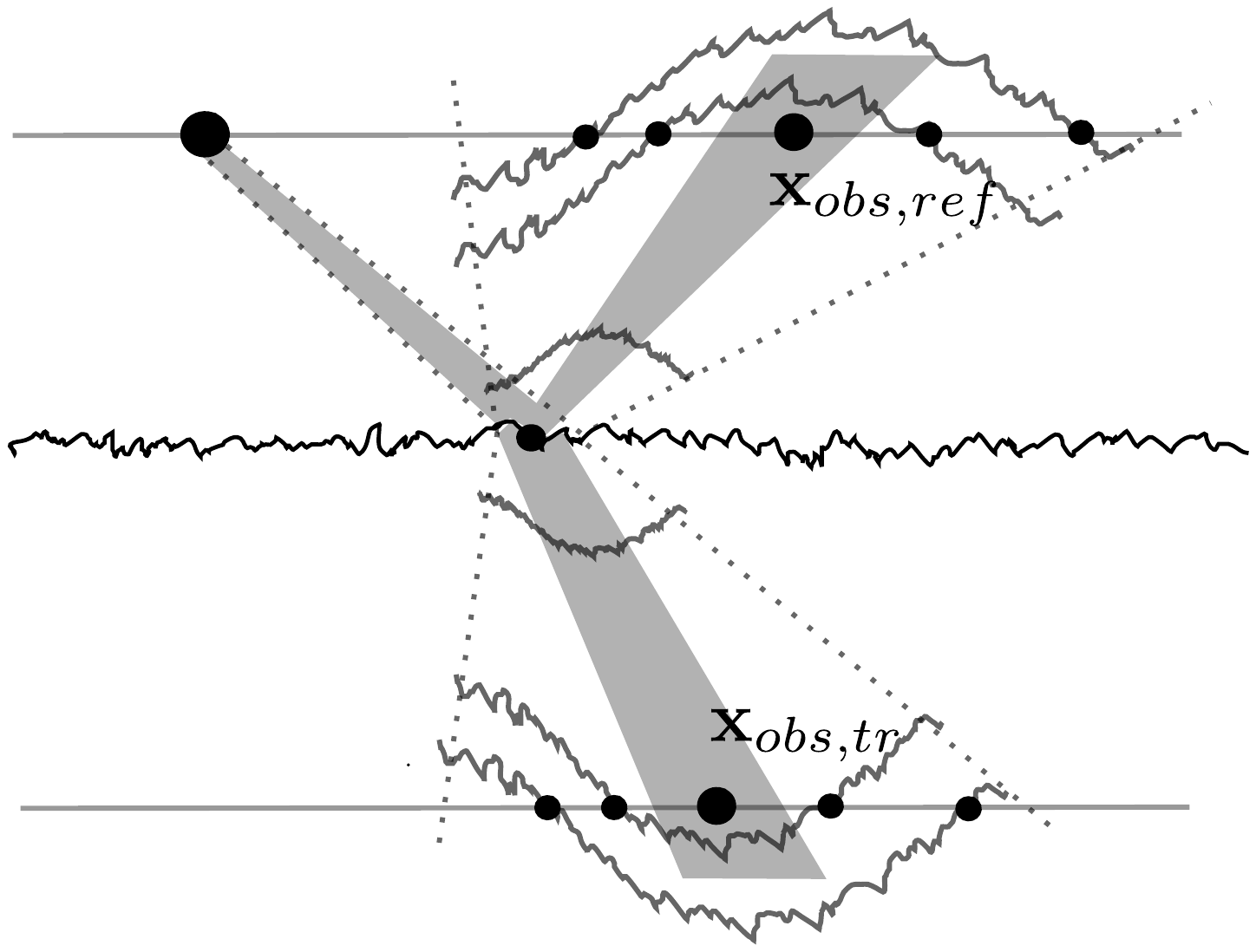} 
\includegraphics*[scale=0.35]{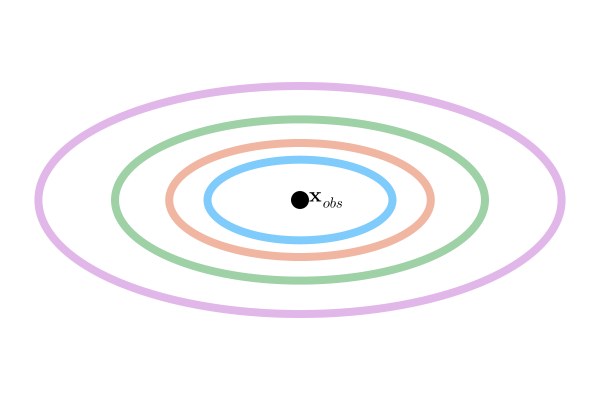} 
\end{center}
\caption{\label{fig:speckle_intro} The left panel illustrates the reflection and refraction of a pulse by a rough interface. The gray region corresponds to the cone formed by the incident wave, as well as the cones associated with the reflected and transmitted specular components observed at $\bx_{obs,ref}$ and $\bx_{obs,tr}$, respectively. When the speckle components intersect a horizontal plane, they form ellipses whose shape evolves in time, as illustrated in the right panel using different colors and centered at $\bx_{obs}$, which denotes either $\bx_{obs,ref}$ or $\bx_{obs,tr}$. The width of these ellipses is determined by the initial pulse duration.
}
\end{figure} 

\section{The physical model}\label{sect1P2}

\paragraph{The wave equation.}

In this paper, we consider three-dimensional linear wave propagation modeled by the scalar wave equation
\begin{equation}\label{eq:wave_equation}
\Delta u - \frac{1}{c^2(\bx,z)}\partial^2_{tt} u = \nabla \cdot {F}(t,\bx,z) \qquad (t,\bx,z)\in \R\times \R ^2 \times \R,
\end{equation}
equipped with zero initial conditions
\[
u(t=0, \bx, z) = \partial_t u (t=0, \bx, z)=0 \qquad (\bx,z)\in \R ^2 \times \R.
\]
The wave model \eqref{eq:wave_equation} implies continuity conditions at the interface, which will be specified below. The coordinate $z$ represents the main propagation axis, while $\bx$ denotes the transverse directions, and $F$ is the source term. Here, the Laplacian operator $\Delta=\Delta_\perp + \partial_{zz}^2$ acts on all spatial variables $\bx$ and $z$. The propagation medium consists of two homogeneous subdomains separated by a random  interface located near the plane $z=z_{int}>0$:
\begin{equation}\label{def:D0}
\CalD^0:=\{(\bx,z)\in \R^2\times \R \quad\text{s.t.}\quad z<z_{int}+\sigma V(\bx/\ell_c)\}
\end{equation}
and
\begin{equation}\label{def:D1}
\CalD^1:=\{(\bx,z)\in \R^2\times \R \quad\text{s.t.}\quad z>z_{int}+\sigma V(\bx/\ell_c)\}.
\end{equation}
We refer to Figure \ref{fig:setup} for an illustration of the physical setup.
\begin{figure}
\begin{center}
\includegraphics*[scale=0.25]{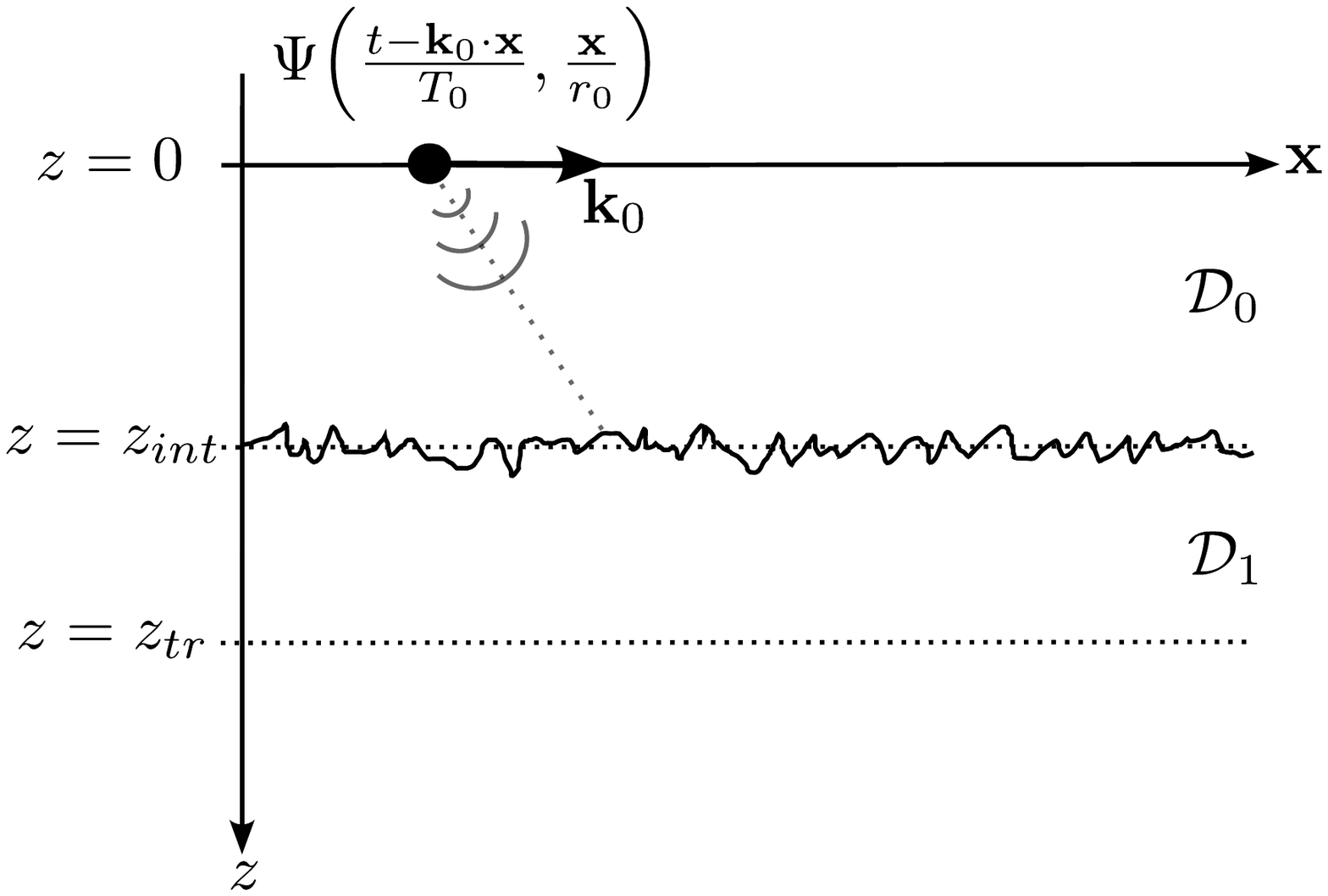} 
\end{center}
\caption{\label{fig:setup} Illustration of the physical setup. The plane $z=0$ contains the source location, while $z=z_{int}$ is the plane around which the rough interface separating $\CalD^0$ and $\CalD^1$ is located. The reflected wave is observed at $z=0$, whereas the transmitted wave is observed at $z=z_{tr}$.
}
\end{figure} 
The function $V$ denotes a mean-zero random field with second-order derivatives, modeling the interface fluctuations and characterizing its roughness. Away from the interface, the wave speed is given by
\begin{equation}\label{eq:wave_speed}
c(\bx,z):= \left\{ \begin{array}{ccl} 
c_0 & \text{ if }  & (\bx,z)\in \CalD^0, \\
c_1 & \text{ if }  &(\bx,z)\in \CalD^1.\\
\end{array} \right.
\end{equation}
The forcing term 
\begin{equation}\label{def:sourceterm}  {F} (t,\bx,z):=
\Psi\Big(\frac{t-\bk_0\cdot \bx}{T_0},\frac{\bx}{r_0}\Big)\delta(z)\mathbf{e}_z,
\end{equation}
where $\mathbf{e}_z$ denotes the unit vector in the $z$-direction, models a source located in the plane $z=0$ that emits a quasi-plane wave $\Psi$ in the direction $(\bk_0,\sqrt{c_0^{-2}-|\bk_0|^2})$ toward the random interface (we assume $|\bk_0|<c_0^{-1}$). The divergence form of the source term in \eqref{eq:wave_equation} is standard in linear acoustics, where $u$ represents the pressure field; see, for instance, \cite{fouque}. This form is adopted here for mathematical convenience, although other types of source terms could be treated similarly. The parameters $T_0$ and $r_0$ denote the temporal pulse width and the spatial beam width of the source, respectively.

The wave equation \eqref{eq:wave_equation} yields the following two continuity relations across the randomly perturbed interface:
\begin{equation}\label{eq:continuity_relation}
u(z=z_{int}(\bx)^+) = u(z=z_{int}(\bx)^-)\qquad\text{and}\qquad \partial_z u(z=z_{int}(\bx)^+) = \partial_z u(z=z_{int}(\bx)^-),
\end{equation}
where
\[
z_{int}(\bx):=z_{int}+\sigma V(\bx/\ell_c).
\]
The derivation of these relations is given in Appendix~\ref{sec:proof_continuity}. Finally, we assume that no waves enter the system from above the source plane or from below the rough interface; all propagating waves are generated by the source term.

\paragraph{The parameter scaling.} The approach considered in this paper is based on a separation of scales technique, in which the relevant length scales are as follows: the typical (reference) propagation distance, denoted by $L$; the central wavelength $\lambda$, related to the pulse duration by
\[
\lambda = c_0 T_0;
\]
the spatial beam radius $r_0$; the correlation length of the interface fluctuations $\ell_c$; and finally the amplitude $\sigma$ of the interface fluctuations $V$. We denote by $z_{int}$ the distance from the source to the interface, and by $z_{tr}$ the distance from the source to the plane where the transmitted wave is recorded (see Figure~\ref{fig:setup}). Both distances are of the order of the typical propagation distance $L$, and we introduce the dimensionless parameter
\[
\e := \frac{\lambda}{L} = \frac{c_0 T_0}{L} \ll 1.
\]  
We consider a high-frequency regime with $\e \ll 1$, together with a paraxial (or parabolic) scaling obtained by enforcing
\[\frac{r^2_0}{\lambda} \sim L,\]
which corresponds to a Rayleigh length of the order of the typical propagation distance. The Rayleigh length $L_R = r_0^2/\lambda$ is the distance from the beam waist to the location at which the cross-sectional area of the beam is doubled. The interface fluctuations are \emph{critically scaled}, meaning that the magnitude of the interface fluctuations and the central wavelength are of the same order,
\ba\label{eq:sl}
\sigma \sim \lambda.
\ea
Finally, let $\ell_c$ denote the correlation length of the (assumed isotropic) interface fluctuations. We consider situations in which this length scale ranges from the central wavelength to the beam width, namely
\ba\label{eq:lesss} 
 \lambda  \lesssim \ell_c  \lesssim  r_0 .
 \ea
The ratio $\lambda/\ell_c$ characterizes how rough the interface appears to the incident wave and determines the angular width of the reflected and transmitted speckle cones. An interface is considered rough when $\lambda \sim \ell_c$, and relatively smooth when $\lambda \ll \ell_c$. We also note that if $r_0 \ll \ell_c$, the situation simplifies and essentially corresponds to a planar interface with random travel-time fluctuations for the reflected and transmitted pulses. The case $\ell_c \ll \lambda$ corresponds to a degenerate regime in which the speckle (diffusive) components form wide cones, whose opening angle is $\pi/2$, and is not considered in the present work.

In the remainder of the paper, for simplicity, we adopt a dimensionless formulation and let  
\begin{equation}\label{eq:scaling_param}
L \sim   1, \qquad r_0/L\to r_0 =\se, \qquad \sigma \to \sigma /L = \e,\qquad \text{and} \qquad \ell_c/L\to \ell_c=\e^\gamma, 
\end{equation}
so that \eqref{eq:sl} is satisfied and in view of \eqref{eq:lesss} we have 
\[
\gamma \in [1/2,1] .
\]
Note that  then $T = L/c_0 = \lambda/(\e c_0) = T_0/\e $ is the characteristic propagation time 
 so that the pulse duration is short compared to the characteristic propagation time, consistent with a high-frequency regime. 
Moreover, the   Rayleigh length $L_R =  r_0^2/\lambda \sim 1$.

In the forthcoming analysis, we distinguish between the case $\gamma = 1/2$ (corresponding to $\ell_c \sim r_0$) and the case $1/2 < \gamma \leq 1$ (corresponding to $\lambda \lesssim \ell_c \ll r_0$). These two regimes lead to qualitatively different behaviors of the reflected and transmitted signals. In the former case, only random specular components are observed, whereas in the latter case, due to the rapid oscillations of the interface relative to the beam radius ($\ell_c \ll r_0$), a homogenization effect occurs for the specularly reflected and transmitted components, together with the emergence of random speckle fields over cones broader than those associated with the specular components.

\paragraph{Random fluctuations.} In \eqref{def:D0} and \eqref{def:D1}, the random fluctuations of the interface separating the two subdomains are described by a continuous, mean-zero stationary random field $V$. For simplicity, we assume that $V(0)$ admits a probability density and that the interface fluctuations satisfy \emph{mixing} properties, describing the loss of statistical dependence of $V$ along the interface. In our context, the basic idea of \emph{mixing} is as follows. For a given set of locations $\bx_1,\dots,\bx_n \in \R^2$, the corresponding values of the field $V(\bx_1),\dots,V(\bx_n)$ become independent as these locations become sufficiently far apart:
\begin{equation}\label{eq:mixing}
\Pro(V(\bx_1)\in A_1,\dots , V(\bx_n)\in A_n) \to \Pro(V(\bx_1)\in A_1)\cdots \Pro( V(\bx_n)\in A_n), 
\end{equation}
for any Borel set $A_1,\dots,A_n \subset \R$, as 
\[
\min_{j,l\in\{1,\dots,n\}} \vert \bx_j - \bx_l \vert \to \infty.
\]
This property is readily satisfied if $V$ is a stationary Gaussian random field with correlation function
\[
R(\bx) = \E[V(\bx + \by)V(\by)],
\] 
which decays to zero as $|\bx|\to \infty$. For more general random fields (not necessarily Gaussian), the property \eqref{eq:mixing} can be formalized through the notion of $\alpha$-mixing as follows. Define
\[
\alpha(r):= \sup_{\substack{S, S' \subset \R^2 \\ d(S,S')> r}} \sup_{\substack{A\in \sigma(V(\bx),\, \bx\in S) \\  B\in \sigma(V(\bx),\,\bx\in S') }} \vert \Pro(A\cap B) - \Pro(A) \Pro(B)\vert,
\]
where
\[d(S,S') = \inf_{\substack{ s\in S\\ s'\in S' }}\vert s- s'\vert\]
denotes the distance between two nonempty subsets $S$ and $S'$, and $\sigma(V(\bx),\, \bx \in S)$ is the $\sigma$-field generated by the family $V_{\vert S} := (V(\bx))_{\bx\in S}$. Informally, the quantity $\alpha(r)$ measures the degree of statistical dependence of the random field $V$ over pairs of regions separated by a distance at least $r$. The $\alpha$-mixing condition corresponds to assuming

\begin{equation}\label{def:alpha_mix}
\alpha(r) \to 0 \qquad \text{as} \qquad r \to \infty,
\end{equation}
implying vanishing statistical dependence between $V_{\vert S}$ and $V_{\vert S'}$ as the distance between $S$ and $S'$ tends to infinity. This notion of mixing plays a central role in this paper, as it allows the homogenization of the specular components, the self-averaging property of empirical two-point correlation functions of the speckle components, and the Gaussianity of the speckle patterns.
 
To quantify the effect of mixing on physical quantities and their covariation, we also introduce a related measure, \emph{$\rho$-mixing}, which captures coherence (or rather the lack thereof) between well-separated random variables. It is defined by 
\begin{equation}\label{def:rho}
\rho(r):=\sup_{\substack{S, S' \subset \R^2 \\ d(S,S')> r}} \sup_{\substack{V\in \CalL^2 ( \sigma(V(\bx),\, \bx\in S) ) \\  W \in \CalL^2(\sigma(V(\bx),\,\bx\in S')) }} 
\vert Corr(V,W) \vert ,
\end{equation}
where $\CalL^2(\CalF)$ denotes the set of $\CalF$-measurable random variables with finite second moments, and $Corr$ denotes the correlation coefficient,
\[
Corr(V,W) = \frac{Cov(V,W)}{\sqrt{Var(V)Var(W)}}.
\]
According to \cite[Theorem~1]{bradley}\footnote{The functions $\alpha$ and $\rho$ used here correspond to $\alpha^\ast$ and $\rho^\ast$ in \cite{bradley}}, one has
\begin{equation}\label{prop:equiv_mix}
\alpha(r) \leq \rho(r) \leq 2\pi\, \alpha(r), \qquad r>0,
\end{equation}
so that the notions of $\alpha$- and $\rho$-mixing are equivalent and indeed we have
$\rho(r) \to 0  ~  \text{as} ~  r \to \infty$. 
 Using the $\rho$-mixing property, the following lemma provides a precise formulation of \eqref{eq:mixing} that will be instrumental in the subsequent asymptotic analysis.
\begin{lemma}\label{lem:mixing}
Let $n\geq 1$ and $V$ be $\rho$-mixing and stationary,  we then have:
\begin{enumerate}
\item For $n$ bounded functions $f_1,\dots,f_n$: $\R \rightarrow \mathbb{C}$, and distinct $\bx_1,\dots,\bx_n\in \R^2$
 \begin{equation*}
\lim_{\eta \to 0}\E\Big[\prod_{j=1}^n f_j\Big(V\Big(\frac{\bx_j}{\eta}\Big)\Big) \Big] = \lim_{\eta \to 0}\prod_{j=1}^n \E\Big[ f_j\Big(V\Big(\frac{\bx_j}{\eta}\Big)\Big) \Big]
 =  \prod_{j=1}^n \E\Big[ f_j\big(V(0)\big) \Big];
\end{equation*}
\item For $n$ bounded functions $g_1,\dots,g_n$:  $\R^2 \rightarrow \mathbb{C}$, $\by_1,\dots,\by_n\in \R^2$, and distinct $\bx_1,\dots,\bx_n\in \R^2$,
\begin{align*}
\lim_{\eta \to 0}\E\Big[\prod_{j=1}^n g_j\Big(V\Big(\frac{\bx_j}{\eta}+\frac{\by_j}{2}\Big),V\Big(\frac{\bx_j}{\eta}-\frac{\by_j}{2}\Big)\Big) \Big] & = \lim_{\eta \to 0}\prod_{j=1}^n \E\Big[ g_j\Big(V\Big(\frac{\bx_j}{\eta}+\frac{\by_j}{2}\Big),V\Big(\frac{\bx_j}{\eta}-\frac{\by_j}{2}\Big)\Big) \Big]\\
&=\prod_{j=1}^n \E\Big[ g_j\Big(V\Big(\frac{\by_j}{2}\Big),V\Big(-\frac{\by_j}{2}\Big)\Big) \Big].
\end{align*}
\end{enumerate}
\end{lemma}
The detailed proof of this lemma is given in Appendix \ref{proof:lem_mix}.

\section{Reflection and transmission for an unperturbed interface}
\label{sec:flat_interface}

Before turning to the analysis of wave scattering by the random interface, we first describe in this section the reflection and transmission mechanisms for an 
\emph{unperturbed or flat interface}. This discussion serves in particular to introduce the main wave quantities and associated terminology in a simple setting, before examining how these notions generalize to the random case. This section is divided into four parts, which address respectively the incident wave at the interface, the transmission and reflection conditions at the unperturbed interface, and finally the description of the reflected and transmitted wave cones.

To study the reflection and transmission of the incident wave, we introduce the following scaled Fourier transform:
\begin{equation*}
\hat{f}^\e(\omega,\bk) :=  \iint f(t,\bx)e^{i\omega(t/\e-\bk \cdot \bx/\se)} dt d\bx, \qquad\text{and}\qquad f(t,\bx)= \frac{1}{(2\pi)^3 \e^2}\iint \hat{f}_\e(\omega,\bk)e^{-i\omega(t/\e-\bk \cdot \bx/\se)}  \omega^2 d\omega d\bk,
\end{equation*}
which is scaled according to the source profile (see \eqref{def:sourceterm} and \eqref{eq:scaling_param}). In the Fourier domain, the wave equation \eqref{eq:wave_equation}, restricted to $z \in (-\infty, z_{int})$, reduces to the Helmholtz equation
\begin{equation}\label{eq:Helmholtz}
\partial^2_z \hp(\bk,z) + \frac{\omega^2}{\e^{2}c_0^2}\Big(1-\e c_0^2 |\bk|^2\Big)\hp(\bk,z)= \e^2 \hat{\Psi}\Big(\omega,\bk-\frac{\bk_0}{\se}\Big)\delta'(z),
\end{equation}
where
\begin{equation}\label{def:fourier_unscale}
\hat{\Psi}(\omega,\bk) :=  \iint \Psi(t,\bx)e^{i\omega(t-\bk \cdot \bx)} dt d\bx
\end{equation}
denotes the unscaled Fourier transform of the source profile $\Psi$.

In what follows, $\bk$-modes satisfying $\se\, c_0 |\bk| < 1$ are referred to as \emph{propagating modes}, while those satisfying $\se\, c_0 |\bk| > 1$ are referred to as \emph{evanescent modes}. To avoid unnecessary treatment of evanescent modes during transmission at the interface, we assume for convenience that
\[c_1 < c_0.\] 
This assumption is used only to simplify the presentation. After minor adaptations, all the results presented in this paper remain valid for the case $c_0 < c_1$, provided that $\bk_0$ satisfies $|\bk_0| < c_1^{-1}$ in order to avoid the critical transmission angle.

Under suitable assumptions on the source profile $\hat{\Psi}$ (compact support in both variables and boundedness away from zero with respect to $\omega$) the source term generates only propagating modes in the medium for sufficiently small $\e$. Indeed, introducing the change of variables
\begin{equation*}
\bk = \bq + \frac{\bk_0}{\se},
\end{equation*}
so that $\bq$ lies in the support of $\hat{\Psi}$, we obtain
\begin{equation}\label{eq:no_eva}
\se\,c_0 |\bk| \leq \se \, c_0 |\bq| + c_0 |\bk_0|< 1,
\end{equation}
for $\e$ sufficiently small, depending only on the support of $\hat{\Psi}$.

For $z \in (z_{int}, z_{tr})$, the wave equation \eqref{eq:wave_equation} in the Fourier domain reduces to the following Helmholtz equation:
\begin{equation}\label{eq:Helmholtz1}
\partial^2_{zz} \hp(\bk,z) + \frac{\omega^2}{\e^{2}c_1^2}\Big(1-\e c_1^2 |\bk|^2\Big)\hp(\bk,z)= 0.
\end{equation}
Throughout this paper, we denote by   
\begin{equation}\label{def_lambda}
\bcals_j^\e(\bk):=\frac{1}{c_j}\sqrt{ 1-\e c_j^2 |\bk|^2},\qquad\text{for}\qquad \se \, c_j \vert \bk \vert < 1,\quad j=0,1,
\end{equation}
the vertical slowness associated with \eqref{eq:Helmholtz} and \eqref{eq:Helmholtz1}. Note that the condition on $\bk$ in \eqref{def_lambda} is satisfied thanks to \eqref{eq:no_eva} and the assumption $c_1 < c_0$.

In the forthcoming analysis, the following expansion  will be used:
\begin{equation}\label{eq:exp_lambda}
\frac{1}{\e}\bcals_j^\e\Big(\bq + \frac{\bk_0}{\se}\Big) = \frac{\bcals_j}{\e} - \frac{\bk_0\cdot \bq}{\se\,\bcals_j} - \frac{c_j}{2} \bq^T A_j\bq + \CalO(\se)\qquad j=0,1,
\end{equation}
where $\bq^T $ stands for the transpose of $\bq$, and
\begin{equation}\label{def:Aj}
A_j:= \frac{1}{(1-c_j^2|\bk_0|^2)^{3/2}}\begin{pmatrix} 
1-c_j^2 \bk^2_{0,2} &  c_j^2 \bk_{0,1}\bk_{0,2}\\
c_j^2 \bk_{0,1}\bk_{0,2} &  1-c_j^2 \bk^2_{0,1}
 \end{pmatrix}= \frac{1}{c_j^3 \bcals_j^3}\big(\mathbf{I}_2 - c_j^2 \,  \bk_0^\perp  \otimes \bk_0^\perp\big).
\end{equation}
Here, the remainder term $\CalO(\se)$ holds uniformly with respect to $\bq$ in the support of $\hat{\Psi}$, and
\begin{equation}\label{def:tauj}
\bcals_j:=\frac{\sqrt{1-c_j^2|\bk_0|^2}}{c_j}\qquad j=0,1,
\end{equation}
denotes the vertical slowness above ($j=0$) and below ($j=1$) the interface. Moreover, $\mathbf{I}_2$ denotes the $2\times 2$ identity matrix, and $\bk_0^\perp := (-\bk_{0,2}, \bk_{0,1})^T$.

\subsection{Reflection and transmission at the interface}

Introducing an up-going and down-going modal decomposition with respect to the $z$-direction, and recalling that no waves are assumed to come from above the source plane or from below the interface, the wavefield reads
\begin{align*}
\hp(\bk,z)&=  \frac{\hbu(\bk)}{\sqrt{\omega \lwu(\bk)}}e^{- i \omega \lwu(\bk)z/\e}   \,\1_{(-\infty,0)}(z)\\
&+\Big(\frac{\hau(\bk)}{ \sqrt{\omega \lwu(\bk)}}e^{ i \omega \lwu(\bk)(z-z_{int})/\e}+\frac{\hbr (\bk)}{\sqrt{\omega\lwu(\bk)}}e^{-i\omega \lwu(\bk)(z-z_{int})/\e}  \Big) \1_{(0,z_{int})}(z)\\
&+ \frac{\hatr (\bk)}{\sqrt{\omega\lwd(\bk)}}e^{i\omega \lwd(\bk)(z-z_{int})/\e}\1_{(z_{int},\infty)}(z).
\end{align*} 
Here, $\hbu(\bk)$ denotes the amplitude of the up-going modes for $z<0$, $\hau(\bk)$ the amplitude of the down-going incident modes impinging on the interface, $\hbr(\bk)$ the amplitude of the reflected modes, and $\hatr(\bk)$ the amplitude of the transmitted modes (see Figure \ref{fig:modes}).
\begin{figure}
\begin{center}
\includegraphics*[scale=0.25]{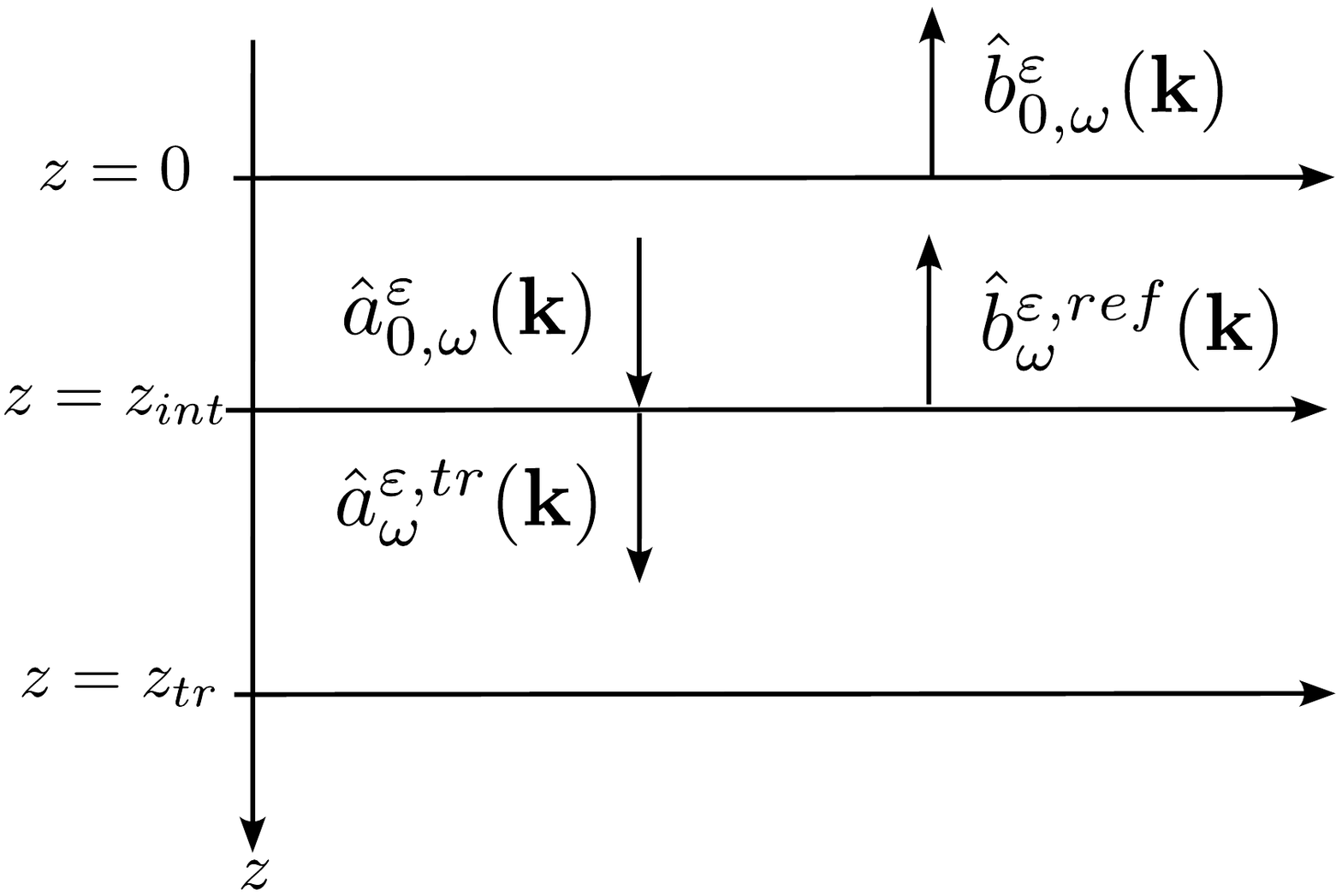} 
\end{center}
\caption{\label{fig:modes} Illustration of the up- and down-going mode amplitudes.
}
\end{figure}

The amplitudes of the incident modes are determined by the following jump conditions across the plane $z=0$. These conditions are induced by the source term in \eqref{eq:Helmholtz} and follow from \eqref{eq:jump_cond} in Appendix \ref{sec:proof_continuity}: 
\begin{align*}
\hp (\bk,z=0^+)- \hp (\bk,z=0^-)&= \e^2 \hat{\Psi}\Big(\omega,\bk-\frac{\bk_0}{\se}\Big), \\
\partial_z \hp (\bk,z=0^+)-\partial_z \hp (\bk,z=0^-) &= 0,
\end{align*}
which yields
\begin{eqnarray}\label{def:a0}
\hau (\bk) &=& \frac{\e^2\sqrt{\omega \lwu(\bk)}}{2}\hat{\Psi}\Big(\omega,\bk-\frac{\bk_0}{\se}\Big)e^{i \omega \lwu(\bk)z_{int}/\e}, \\
 \hbu(\bk)  &=&  - \frac{\e^2\sqrt{\omega \lwu(\bk)}}{2}\hat{\Psi}\Big(\omega,\bk-\frac{\bk_0}{\se}\Big) + \hbr (\bk) e^{ i \omega \lwu(\bk)z_{int}/\e}. \nonumber
\end{eqnarray} 
Therefore, the up-going modes $\hbu(\bk)$ consist of a direct contribution from the source together with the reflected component from the interface.

Denote now
\[
\hpi(\bk, z) := \frac{\hau(\bk)}{ \sqrt{\omega \lwu(\bk)}}  e^{i\omega \lwu(\bk)(z-z_{int})/\e}
= \frac{\e^2}{2}\hat{\Psi}\Big(\omega,\bk-\frac{\bk_0}{\se}\Big)e^{ i \omega \lwu(\bk)z/\e},
\]
the incident wavefield,
 \[
\hpr(\bk,z) := \frac{\hbr (\bk)}{\sqrt{\omega\lwu(\bk)}}e^{-i\omega \lwu(\bk)(z-z_{int})/\e} \qquad 0<z<z_{int},
\]
the reflected wavefield, and 
\[\hpt(\bx,z) := \frac{\hatr (\bk)}{\sqrt{\omega\lwd(\bk)}}e^{i\omega \lwd(\bk)(z-z_{int})/\e} \qquad z>z_{int},\]
the transmitted wavefield. For a flat interface ($V\equiv 0$), the continuity conditions \eqref{eq:continuity_relation} at $z=z_{int}$ yield
\begin{align*}
\hpi( \bx,z_{int}) + \hpr(\bx,z_{int})&=\hpt(\bx,z_{int}),\\
\partial_z \hpi(\bx,z_{int}) + \partial_z \hpr(\bx,z_{int}) &= \partial_z \hpt(\bx,z_{int}).
\end{align*}
These conditions lead to the system
\begin{align*}
 \frac{\hau(\bk)}{\sqrt{\lwu(\bk)}}+\frac{\hbr (\bk)}{\sqrt{\lwu(\bk)}} & =\frac{\hatr (\bk)}{\sqrt{\lwd(\bk)}},\\
\sqrt{\lwu(\bk)} \hau (\bk)-\sqrt{\lwu(\bk)}\hbr (\bk)&=\sqrt{\lwd(\bk)}\hatr (\bk),
\end{align*}
whose solution is
\[
\hatr (\bk)= \frac{1}{\tau^\e_+(\bk)}\hau(\bk)\qquad\text{and}\qquad\hbr (\bk)=\frac{\tau^\e_-(\bk)}{\tau^\e_+(\bk)} \hau(\bk),
\]
where 
\begin{equation*}
\tau^\e_\pm(\bk)=\frac{1}{2}\left(    \sqrt{\frac{\lwu(\bk)}{\lwd(\bk)}} \pm \sqrt{\frac{\lwd(\bk)}{\lwu(\bk)}}\right),
\end{equation*}
yielding the energy conservation relation
\begin{equation*}
 |\hau(\bk)|^2     =   |\hatr (\bk)|^2 + |\hbr (\bk)|^2 .
\end{equation*}

\subsection{The reflected wave}\label{sec:rec_reflect}

From the above analysis, the reflected wave observed in the plane $z=0$ can be expressed explicitly as
\begin{align*}
u^\e(t,\bx,z=0^-)&=\frac{1}{(2\pi)^3 \e^2}\iint \hpr(\bk,z=0) e^{-i\omega( t /\e - \bx\cdot \bk/\se)} \omega^2 d\omega d\bk \\
& = \frac{1}{(2\pi)^3 \e^2} \iint e^{i\omega \lwu(\bk) z_{int}/\e} \frac{\tau^\e_-(\bk)\hau(\bk)}{\tau^\e_+(\bk)\sqrt{\omega\lwu(\bk)}} e^{-i\omega( t /\e - \bx\cdot \bk/\se)} \omega^2 d\omega d\bk\\
&= \frac{1}{2(2\pi)^3} \iint e^{2i\omega \lwu(\bq+\bk_0/\se) z_{int}/\e}\frac{\tau^\e_-(\bq+ \bk_0/\se)}{\tau^\e_+(\bq+ \bk_0/\se)} e^{-i\omega( t /\e - \bx\cdot \bk_0/\e - \bx\cdot \bq/\se)} \hat \Psi(\omega,\bq) \omega^2 d\omega d\bq,
\end{align*}
where the last line follows from the change of variables $\bk = \bq + \bk_0/\se$, consistent with the argument of $\hat \Psi$ in \eqref{def:a0}.
 
Using the expansion \eqref{eq:exp_lambda}, we introduce the following observation position and time:
\begin{equation}\label{def:pos_time}
\bx = \bx_{obs,ref} +\se\, \by
\qquad\text{ and }\qquad
t=t^\e_{obs,ref}(\by):= t_{obs,ref}+\se\, \bk_0\cdot \by,
\end{equation}
where
\begin{equation}\label{def:pos_time_ref}
\bx_{obs,ref} := \frac{2\bk_0 z_{int}}{\bcals_0}\qquad\text{and}\qquad t_{obs,ref}:=\frac{2 z_{int}}{c^2_0\bcals_0},
\end{equation}
with $\bcals_0$ defined in \eqref{def:tauj}. With this choice, we obtain a nontrivial asymptotic expression for the specularly reflected wave:
\begin{align}\label{def:Uref}
U^{ref}(s,\by) & := \lim_{\e\to 0}u^\e\big(t^\e_{obs,ref}(\by)+\e s,\bx_{obs,ref} +\se\, \by,z=0\big) \nonumber\\
& =\frac{\CalR}{2(2\pi)^3} \int e^{-i\omega (s-\by \cdot \bq)}  \hat \CalU_0(\omega, \bq, 2z_{int})\hat \Psi(\omega,\bq) \omega^2  d\omega d\bq.
\end{align}
Here, $U^{ref}$ represents the asymptotic wave front observed in the reference frame of the source at position $\bx_{obs,ref}$ and time $t^\e_{obs,ref}(\by)$. The dependence of the travel time on the offset $\by$ reflects the fact that the beam width $\se$ is larger than the pulse width $\e$, combined with the oblique propagation of the wave front relative to the vertical direction (see Figure~\ref{fig:scaling1}). 
\begin{figure}
\begin{center}
\includegraphics*[scale=0.25]{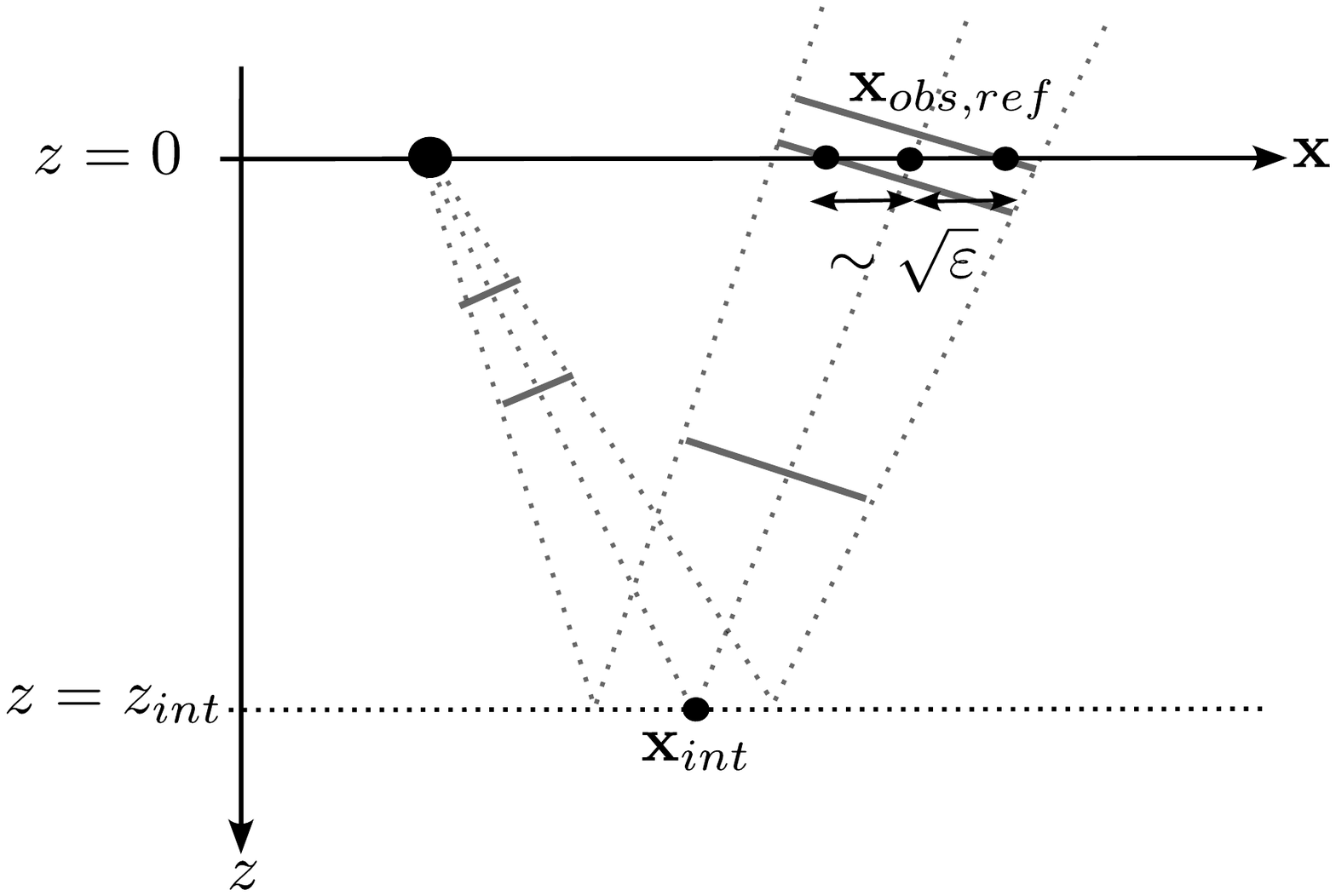} 
\end{center}
\caption{\label{fig:scaling1} Illustration of the correction at the scale of the beam width $\se$ in \eqref{def:pos_time}, which is larger than the pulse width $\e$, represented by the width of the straight gray lines within the probing and specular reflected cones. This $\se$ correction needs to be taken into account in the observation time to accurately describe the specularly reflected wave across the beam width.
}
\end{figure} 
Note that $\bx_{obs,ref}$ is twice the lateral position 
\begin{equation}\label{def:x_int}
\bx_{int} := \frac{\bk_0 z_{int}}{\bcals_0},
\end{equation}
at which the incident pulse impinges on the interface. Consequently, the standard reflection relation reads
\begin{equation}\label{def:angle_inc_ref}
\theta_{inc} = \theta^0_{ref} := \arctan\Big(\frac{ \vert \bk_0 \vert}{\bcals_0}\Big),
\end{equation}
where the incident and reflected angle are equal. We refer to Figure~\ref{fig:ref} for an illustration of the geometric properties of reflection.

\begin{figure}
\begin{center}
\includegraphics*[scale=0.25]{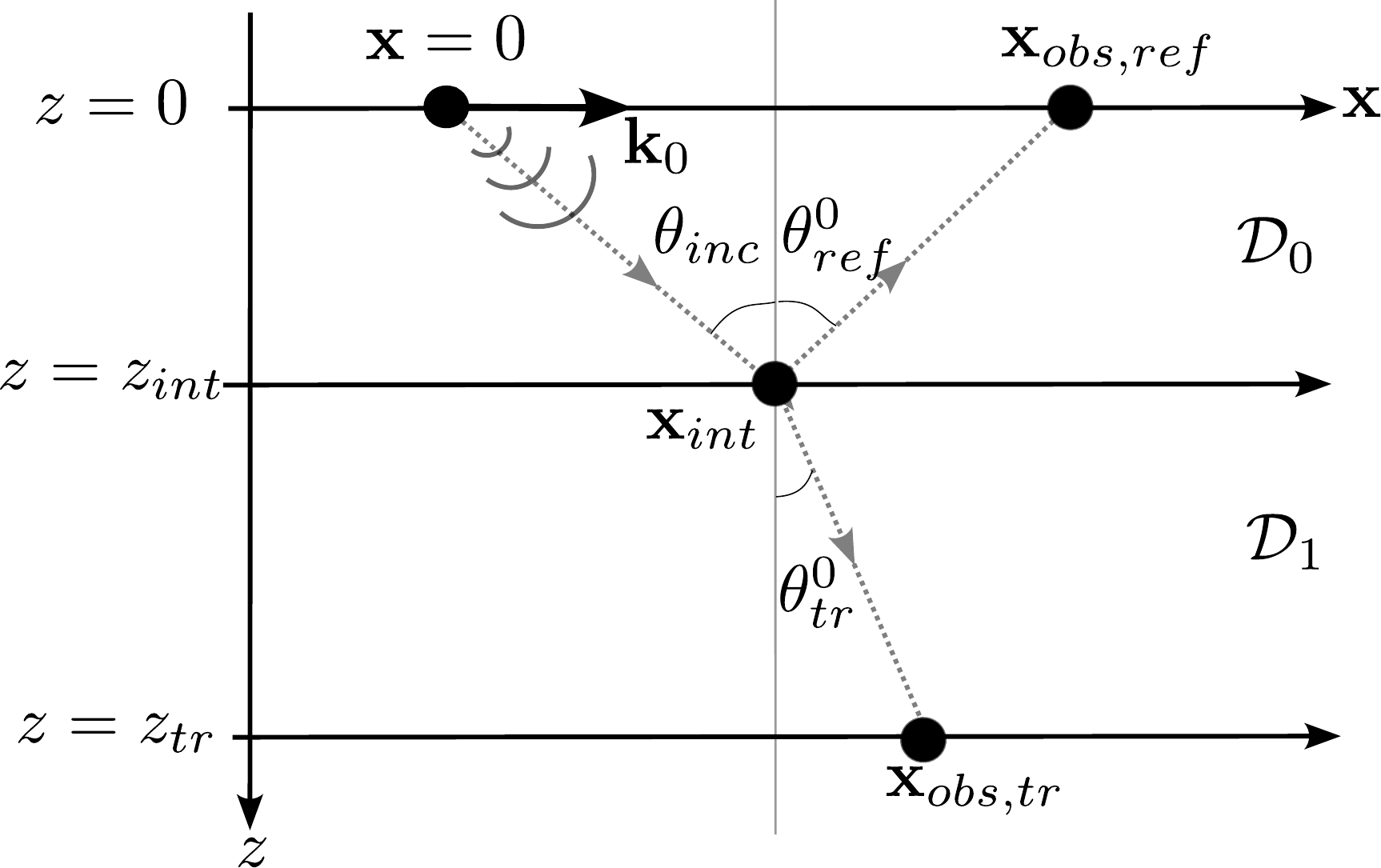} 
\end{center}
\caption{\label{fig:ref} Illustration of reflection and transmission for a flat interface located at $z=z_{int}$. The source is located at $\bx=0$ in the plane $z=0$, and the emitted wave impinges on the interface at $\bx=\bx_{int}$ with incident angle $\theta_{inc}$. The reflected wave, with angle $\theta^0_{ref}$, is observed in the plane $z=0$ at $\bx=\bx_{obs,ref}$, while the transmitted wave, with angle $\theta^0_{tr}$, is observed in the plane $z=z_{tr}$ at $\bx=\bx_{obs,tr}$.
}
\end{figure} 

Regarding the pulse profile of the reflected wave, the reflection coefficient is given by
\begin{equation}\label{def:coef_ref}
\CalR:=\frac{\bcals_0-\bcals_1}{\bcals_0+\bcals_1},
\end{equation}
and
\begin{equation}\label{def:hCalU0}
\hat \CalU_0(\omega,\bq,z) := e^{-i\omega z c_0 \bq^T A_0\bq /2 },
\end{equation}
where $A_0$ is defined in \eqref{def:Aj}. The term $\hat \CalU_0$ generates the following homogeneous semigroup:
\begin{equation*}
\check \CalU_0(\omega,\by,z) := \frac{\omega^2}{(2\pi)^2} \int e^{i\omega \by \cdot \bq} e^{-i\omega z c_0 \bq^T A_0\bq /2 }d\bq,
\end{equation*}
satisfying the Schr\"odinger type equation
\begin{equation}\label{eq:schro0}
 i \partial_z \check \CalU_0 (\omega, \by, z) + \frac{1}{2\ku} \nabla_{\by} \cdot \Big(A_0 \nabla_{\by} \check \CalU_0 \Big)(\omega,\by,z)=0\qquad z> 0,
\end{equation}
where
\[\check \CalU_0(\omega, \by, z=0)=\delta(\by)\qquad\text{and}\qquad \ku = \omega/c_0, \] 
which is characteristic of the paraxial approximation. Owing to the nonzero initial lateral direction $\bk_0$, the equation does not reduce to the standard Schr\"odinger equation with a Laplacian operator. However, the standard Laplacian structure is recovered in the limit $\bk_0 \to 0$.

In the time domain, the function 
\begin{equation}\label{def:CalU0}
\CalU_0(s, \by, z):=\frac{1}{2(2\pi)^3} \int e^{-i\omega (s-\by\cdot\bq)} \hat \CalU_0(\omega,\bq,z)\hat \Psi(\omega,\bq) \omega^2 d\omega d\bq,
\end{equation}
which corresponds to the pulse profile in \eqref{def:Uref}, satisfies the paraxial wave equation
\begin{equation}\label{eq:paraxial}
\partial^2_{s z} \CalU_0 - \frac{c_0}{2}\nabla_{\by} \cdot \Big(A_0 \nabla_{\by} \CalU_0 \Big) =0,\qquad\text{with}\qquad \CalU_0(s, z=0, \by)=\frac{1}{2} \Psi(s,\by).
\end{equation}
As a result, the wave reflected by a flat (unperturbed) interface can be written as
\begin{equation*}
U^{ref}(s,\by) = \CalR \, \CalU_0(s, \by, 2z_{int}). 
\end{equation*} 
Within the paraxial wave model, this expression corresponds to the propagation of the emitted pulse from the source location to the interface and back to the source plane, corresponding to a total propagation distance $2z_{int}$. The factor $\CalR$ accounts for reflection at the interface.

\subsection{The transmitted wave}\label{sec:rec_trans}

Following the same strategy as for the reflected wave, the transmitted wave observed at $z=z_{tr}$ can be written as
\begin{align*}
u^\e(t,\bx,z=z_{tr})&=\frac{1}{(2\pi)^3 \e^2}\iint \hpt(\bk,z=z_{tr}) e^{-i\omega( t /\e - \bx\cdot \bk/\se)} \omega^2 d\omega d\bk \\
& = \frac{1}{(2\pi)^3 \e^2}  \iint \frac{\hau(\bk)}{\tau^\e_+(\bk) \sqrt{\omega\lwd(\bk)}}e^{ i \omega \lwd(\bk)(z_{tr}-z_{int})/\e}  e^{-i\omega( t /\e - \bx\cdot \bk/\se)} \omega^2 d\omega d\bk\\
&= \frac{1}{2(2\pi)^3} \iint e^{ i \omega (\lwu(\bq+ \bk_0/\se)z_{int}+\lwd(\bq+ \bk_0/\se)(z_{tr}-z_{int}))/\e} \frac{\sqrt{\lwu(\bq+ \bk_0/\se)}}{\tau^\e_+(\bq+ \bk_0/\se)\sqrt{\lwd(\bq+ \bk_0/\se)}} \\
& \hspace{7cm} \times e^{-i\omega( t /\e - \bx\cdot \bk_0/\e - \bx\cdot \bq/\se)} \hat \Psi(\omega,\bq) \omega^2 d\omega d\bq.
\end{align*}
Using the expansion \eqref{eq:exp_lambda}, we obtain the following asymptotic expression for the specularly transmitted wave:
\begin{align}\label{def:Utr}
U^{tr}(s,\by) & := \lim_{\e\to 0}u^\e\big(t^\e_{obs,tr}(\by)+\e s,\bx_{obs,tr} + \se \by,z=z_{tr}\big) \\
& = \frac{\CalT }{2(2\pi)^3 } \sqrt{\frac{\bcals_0}{\bcals_1}} \int e^{-i\omega (s-\by\cdot \bq)}  \hat \CalU_1(\omega,\bq,z_{tr}-z_{int}) \hat \CalU_0(\omega,\bq,z_{int})  \hat\Psi(\omega,\bq) \omega^2 d\omega d\bq,\nonumber
\end{align}
with
\begin{equation}\label{def:pos_tr}
\bx_{obs,tr} := \bx_{int} + \bx_{tr} :=\bx_{int} + \frac{\bk_0(z_{tr}- z_{int})}{\bcals_1},
\end{equation}
where $\bx_{int}$ is given by \eqref{def:x_int}, and the vertical slownesses $\bcals_j$ by \eqref{def:tauj}, and The corresponding observation time by
\begin{equation}\label{def:time_tr}
t^\e_{obs,tr}(\by):=t_{obs,tr} + \se \, \bk_0 \cdot \by \qquad\text{with}\qquad t_{obs,tr}:=\frac{z_{int}}{c^2_0\bcals_0}+\frac{z_{tr}-z_{int}}{c^2_1\bcals_1}.
\end{equation}
Here, the quantity $U^{tr}$ represents the wave front observed in the reference frame of the source, at position $\bx_{obs,tr}$ and time $t^\e_{obs,tr}(\by)$. The arrival time $t_{obs,tr}$ is the sum of the travel time from the source to the interface and from the interface to the plane $z=z_{tr}$. As in the reflected case, the $\se$ correction in $t^\e_{obs,tr}(\by)$ accounts for the lateral offset of the observation point, which is large compared to the pulse width, combined with the oblique propagation of the wave front relative to the vertical direction.

The lateral position $\bx_{obs,tr}$ corresponds to the sum of two contributions: the lateral position where the pulse impinges on the interface and the additional lateral displacement accumulated as the pulse propagates through the second medium to the plane $z=z_{tr}$. We refer to Figure~\ref{fig:ref} for an illustration of the geometric properties of transmission. From this formulation, we recover the Snell’s law:
\ba\label{eq:snell0} 
\frac{\sin(\theta_{inc})}{c_0} = \frac{\sin(\theta^0_{tr})}{c_1} = |\bk_0|,
\ea
with $\theta_{inc}$ defined by \eqref{def:angle_inc_ref}, and
\begin{equation}\label{def:angle_tr}
\theta^0_{tr}:=\arctan\Big(\frac{\vert \bk_0 \vert}{\bcals_1}\Big).
\end{equation}

Regarding the pulse profile, the transmission coefficient is defined as
\begin{equation}\label{def:coef_tr}
\CalT := \frac{2 \sqrt{\bcals_0 \bcals_1}}{\bcals_0+\bcals_1},
\end{equation}
so that the conservation relation $\CalR^2 + \CalT^2 = 1$ holds. The function $\hat\CalU_0$ is defined in \eqref{def:hCalU0}, while $\hat\CalU_1$ is defined analogously by
\begin{equation}\label{def:hCalU1}
\hat \CalU_1(\omega,\bq,z) := e^{-i\omega z c_1 \bq^T A_1\bq / 2 },
\end{equation}
where $A_1$ is given by \eqref{def:Aj}. Denoting by $\CalU_1$ the analogue of \eqref{def:CalU0} with $c_1$ and $A_1$ in place of $c_0$ and $A_0$, which satisfies the corresponding paraxial wave equation \eqref{eq:paraxial}, the expression \eqref{def:Utr} describes the transmission of the emitted pulse from the source to the interface and then from the interface to the plane $z=z_{tr}$ within the paraxial approximation. The factor $\CalT$ accounts for transmission at the interface.

\section{The reflected and transmitted waves for a random interface}\label{sec:randint}

In this section, we generalize the results of the previous section to the case of a randomly fluctuating interface. We rely on the modal decomposition introduced earlier to derive expressions for the reflected and transmitted waves.

Recall that the interface is defined by   
\[
z = z_{int}^\e (\bx) := z_{int} + \e \, V\Big(\frac{\bx}{\e^\gamma}\Big) .
\]
As above we express the incident, reflected, and transmitted waves as 
\begin{align*}
u^{\e,inc}(t,\bx, z) & = \frac{1}{(2\pi)^3 \e^{2} }\iint e^{-i\omega( t /\e - \bx\cdot \bk/\se)} \frac{\hau(\bk)}{\sqrt{\omega \lwu(\bk)}} e^{i\omega\lwu(\bk)(z-z_{int})/\e}\omega^2 d\omega d\bk\qquad 0<z<z_{int}^\e (\bx) , \\
u^{\e,ref}(t,\bx, z) & = \frac{1}{(2\pi)^3 \e^{2} }\iint e^{-i\omega( t /\e - \bx\cdot \bk/\se)} \frac{\hbr(\bk)}{\sqrt{\omega \lwu(\bk)}} e^{-i\omega\lwu(\bk)(z-z_{int})/\e}\omega^2 d\omega d\bk\qquad 0<z<z_{int}^\e (\bx), \\
u^{\e,tr}(t,\bx, z) & = \frac{1}{(2\pi)^3 \e^{2} }\iint e^{-i\omega( t /\e - \bx\cdot \bk/\se)} \frac{\hatr(\bk)}{\sqrt{\omega \lwd(\bk)}} e^{i\omega\lwd(\bk) (z-z_{int})/\e} \omega^2 d\omega d\bk\qquad z>z_{int}^\e (\bx).
\end{align*}
These expressions yield, by \eqref{eq:continuity_relation}, the two continuity relations along the random interface $z=z_{int}^\e(\bx)$ that encode the transmission and reflection mechanisms:
\begin{align}\label{eq:continuity}
u^{\e,inc}(t, \bx,{z_{int}^\e (\bx)}^-) + u^{\e,ref}(t, \bx, {z_{int}^\e (\bx)}^-) &= u^{\e,tr}(t, \bx, {z_{int}^\e (\bx)}^+),\\
\partial_z u^{\e,inc}(t, \bx, {z_{int}^\e (\bx)}^-) + \partial_z u^{\e,ref}(t, \bx, {z_{int}^\e (\bx)}^-) &= \partial_z u^{\e,tr}(t, \bx, {z_{int}^\e (\bx)}^+).\nonumber
\end{align}
According to these relations, the observed reflected and transmitted waves can be written in a form that makes explicit the role of the interface fluctuations and which is convenient for the subsequent analysis.
\begin{lemma}\label{lem:tech}
At leading order in $\e$, the reflected and transmitted waves can be written as
\begin{align} \label{eq:ue_ref}
u^{\e,ref}(t,\bx, z=0) & \simeq \frac{\CalR}{2(2\pi)^3  }\iiint e^{-i\omega(t -2\bcals_0 z_{int}-\bx\cdot \bk_0 ) /\e} 
     e^{i \omega \bq \cdot (\bx  - \bx_{obs,ref})/\se}  \nonumber   \\
& \hspace{.4cm} \times
K^\e(2\bcals_0, \omega,  \bq, \bq')  
  \hat \CalU_0(\omega,\bq,z_{int}) \hat \CalU_0(\omega,\bq',z_{int}) \hat \Psi(\omega, \bq') \omega^2 d\omega  d\bq' d\bq, 
\end{align}
and
\begin{align}\label{eq:ue_tr}
u^{\e,tr}(t,\bx, z=z_{tr}) & \simeq \frac{\CalT}{2(2\pi)^3  } \sqrt{\frac{\bcals_0}{\bcals_1}} \iiint e^{-i\omega(t -\bcals_0 z_{int}-\bcals_1(z_{tr}-z_{int})-\bx\cdot \bk_0 ) /\e}    e^{i \omega \bq \cdot (\bx - \bx_{obs,tr})/\se} \nonumber \\
& \hspace{.4cm} \times   K^\e(\bcals_0-\bcals_1, \omega,  \bq, \bq')  
\hat \CalU_1(\omega,\bq,z_{tr}-z_{int}) \hat \CalU_0(\omega,\bq',z_{int})\hat \Psi(\omega, \bq') \omega^2 d\omega  d\bq' d\bq .   
\end{align}
The scattering operator $K^\e$ associated with the random interface admits the following Fourier representation:
\begin{equation}\label{def:scat_op}
K^\e(\tau,\omega,\bq,\bq') = \frac{\omega^2}{ \e (2\pi)^{2}} \int e^{i\omega (\bq' - \bq)\cdot(\bx'-\bx_{int})/\se}e^{i\omega \tau V(\bx'/\e^\gamma)} d\bx'.
\end{equation}
The vertical slownesses $\bcals_j$ are defined by \eqref{def:tauj}, the functions $\hat \CalU_j$ by \eqref{def:hCalU0} and \eqref{def:hCalU1}, and $\bx_{int}$ (resp. $\bx_{obs,ref}$ and $\bx_{obs,tr}$), representing the lateral position where the incident pulse impinges on the interface (resp. the lateral positions where the reflected and transmitted waves are observed), are defined by \eqref{def:x_int}, (resp. \eqref{def:pos_time_ref}  and \eqref{def:pos_tr}).
\end{lemma}

The scattering operator $K^\e$ is a central quantity  that encapsulates the effects of interface randomness. It is parameterized by the angular frequency $\omega$, is centered at $\bx_{int}$, and has an effective support in $\bq-\bq'$ comparable to the beam width $\se$. This operator describes how an incident plane wave with lateral direction $\bq'$, supported by the source term $\hat \Psi$, is scattered into a direction $\bq$. 

For both the reflected \eqref{eq:ue_ref} and transmitted \eqref{eq:ue_tr} waves, the incident plane wave with direction $\bq'$ propagates according to $\hat \CalU_0(\bq')$ under the paraxial approximation until it reaches the rough interface. At the interface, the operator $K^\e$ accounts for scattering and mode coupling. After interaction with the interface, the reflected component propagates according to $\hat \CalU_0(\bq)$ in the opposite $z$-direction, while the transmitted component propagates according to $\hat \CalU_1(\bq)$.

Note that for $V\equiv 0$ the scattering operator reduces to the identity, and we recover 
\begin{align*}
\lim_{\e\to 0} u^{\e,ref}(t^\e_{obs,ref}(\by)+\e s,\bx_{obs,ref} +\se\, \by, z=0) &  =   U^{ref}(s,\by) , \\\
\lim_{\e\to 0} u^{\e,tr}(t^\e_{obs,tr}(\by)+\e s,\bx_{obs,tr} + \se \by,z=z_{tr}) &   =  U^{tr}(s,\by)  , 
\end{align*}
with $U^{ref}$ and $U^{tr}$ defined in \eqref{def:Uref} and \eqref{def:Utr}, respectively.
 
In the remainder of the paper, we further analyze the effective statistical behavior of the operator $K^\e$ in the high-frequency limit $\e \to 0$. The two cases
\[
\gamma=1/2\qquad\text{and}\qquad \gamma >1/2,
\] 
are analyzed separately as they lead to qualitatively different statistical behaviors for the reflected and transmitted waves.

\section{Specular cone  in the case $\gamma=1/2$}\label{sec:specint}

The case $\gamma = 1/2$ corresponds to a situation where the correlation length of the interface fluctuations and the beam width are of the same order. Since the wavelength is also assumed to be of the same order as the amplitude of the fluctuations, it is natural to expect random arrival times for the reflected and transmitted waves, but no homogenization effect (see Figure \ref{fig:ref_rand} for an illustration).
\begin{figure}
\begin{center}
\includegraphics*[scale=0.25]{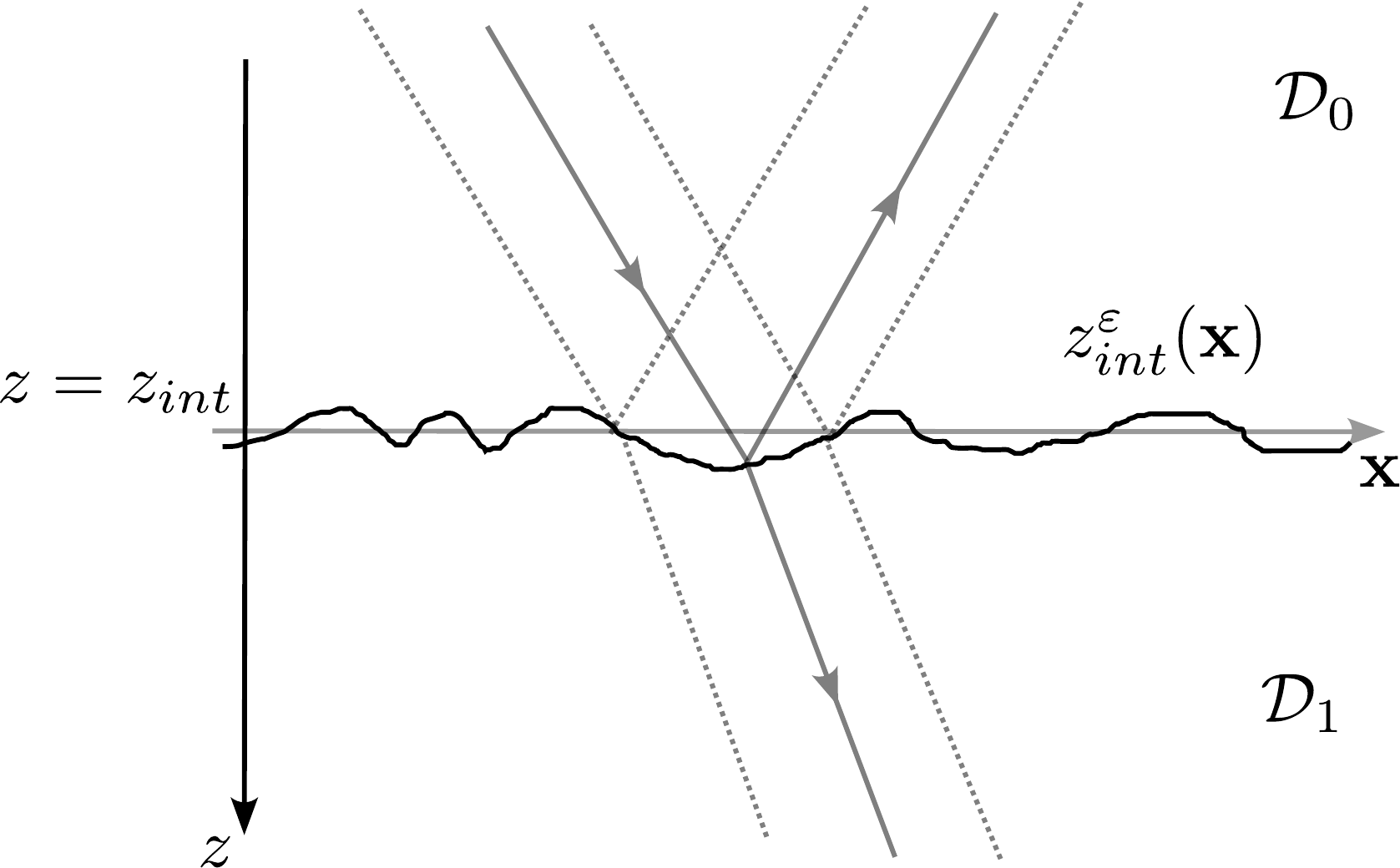} 
\end{center}
\caption{\label{fig:ref_rand} Illustration of reflection and transmission for a random interface with $\gamma = 1/2$
so that the width of the beam is on the spatial scale of variations in the random interface.  The travel times of the reflected and transmitted waves are affected by the fluctuations $V$.
}
\end{figure} 
To describe this phenomenon mathematically for the reflected wave, we perform the change of variables
\[\bx' \to \bx_{int} + \se \, \by'\]
in \eqref{eq:ue_ref}, so that, at leading order, the observed reflected wavefield reads
\begin{align*}
u^{\e,ref}(t,\bx, z=0) & \simeq \frac{\CalR}{2(2\pi)^5}\iiiint e^{-i\omega(t -2\bcals_0 z_{int}-\bx\cdot \bk_0 ) /\e} e^{i \omega (\bq'- \bq) \cdot \by'} e^{2i\omega \bcals_0 V(\bx_{int}/\se + \by')} \\
&\hspace{2.5cm} \times e^{i \omega \bq \cdot (\bx - \bx_{obs,ref})/\se} \hat \CalU_0(\omega,\bq,z_{int}) \hat \CalU_0(\omega, \bq',z_{int}) \\
& \hspace{2.5cm} \times \hat \Psi(\omega, \bq') \omega^4 d\omega d\by' d\bq' d\bq.
\end{align*}
Here, the lateral position $\bx_{obs,ref}$ is the same as in the flat-interface case (see \eqref{def:pos_time_ref}), $\bcals_0$ is defined by \eqref{def:tauj}, and $\hat \CalU_0$ by \eqref{def:hCalU0}. The presence of fast oscillatory components suggests introducing
\[
U^{\e,ref}(s, \by) := u^{\e,ref}\big(t^\e_{obs,ref}(\by) + \e s,\bx_{obs,ref} +\se \by, z=0\big)\qquad (s, \by)\in \R \times \R^2,
\]
where $t^\e_{obs,ref}(\by)$ is defined by \eqref{def:pos_time} and corresponds to the expected travel time for observing the reflected wave in the plane $z=0$. The wave front $U^{\e,ref}$ is then observed at the same lateral position and in the same moving time frame as the specularly reflected wave in the flat-interface case. 
We remark that hroughout the remaining of the paper, for notational simplicity, we denote
\[\CalX = \R \times \R^2.\]
For the reflected wave front we then have the following result. 
\begin{proposition}\label{prop:random_travel_time}
The family $(U^{\e,ref})_\e$ converges in law in $L^2(\CalX)$ to 
\begin{align*} 
U^{ref}(s, \by) & := \frac{\CalR}{2(2\pi)^5}\iiiint e^{-i\omega (s- 2\bcals_0 \CalV(\by') - \bq \cdot \by)} e^{i \omega (\bq'- \bq) \cdot \by'} \\
& \hspace{3cm} \times \hat \CalU_0(\omega,\bq,z_{int})\hat \CalU_0(\omega,\bq',z_{int}) \hat \Psi(\omega, \bq') \omega^4 d\omega d\by' d\bq' d\bq,
\end{align*}
where $\CalV$ is a random field with the same law as $V$.
\end{proposition} 

The reflected wave front $U^{ref}$ consists in a superposition of contribution observed at point $\by$. Each of these contributions result in waves emitted with initial lateral direction $\bq'$, propagating according to $\hat \CalU_0(\bq')$, and scattered into direction $\bq$ at point $\by'$ before propagating back according to $\hat \CalU_0(\bq)$ (see Figure \ref{fig:ref_wave_rand} for an illustration).
\begin{figure}
\begin{center}
\includegraphics*[scale=0.25]{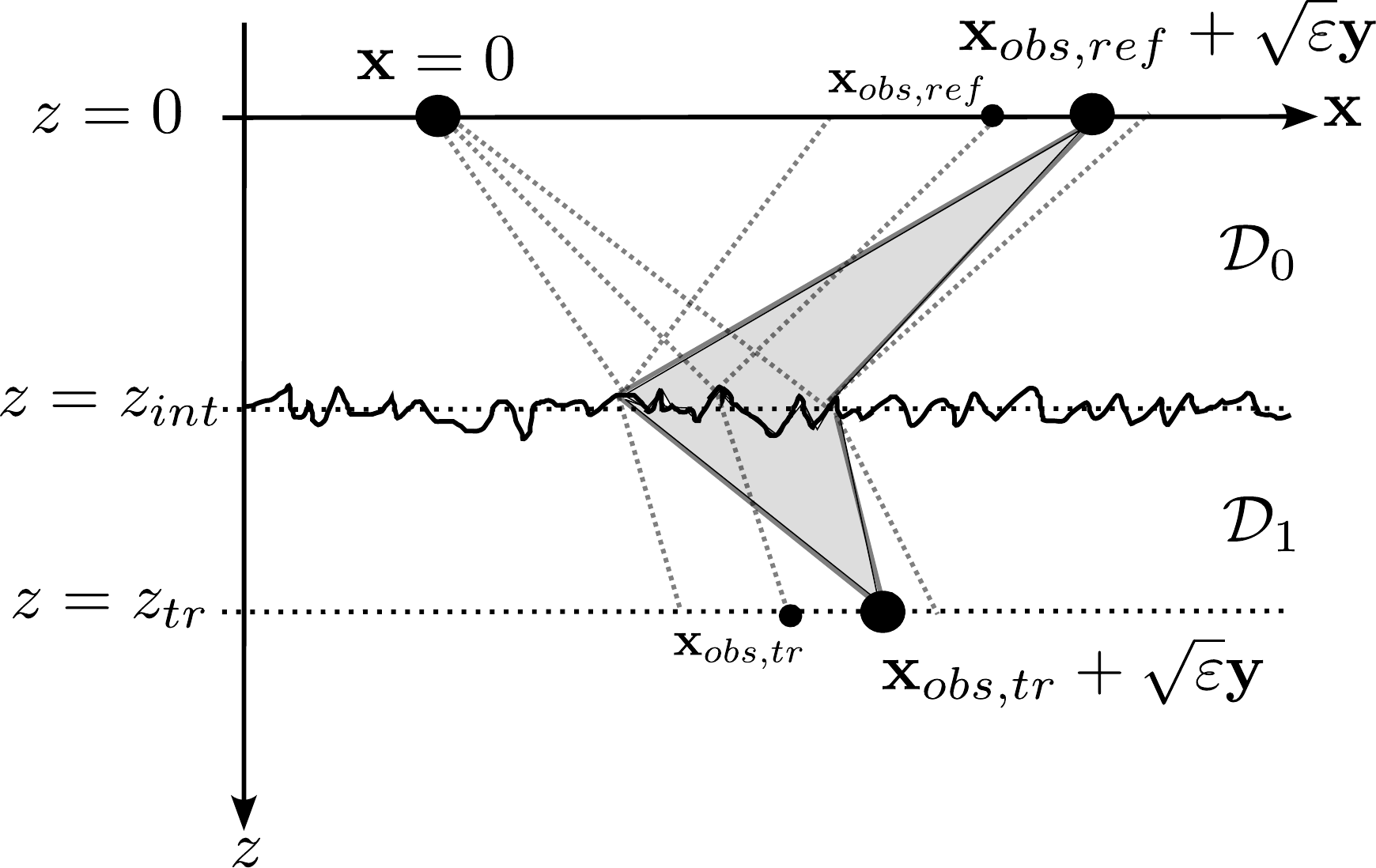} 
\end{center}
\caption{\label{fig:ref_wave_rand} Illustration of the contributions to the  reflected and transmitted waves observed at points $\bx_{obs,ref}+\se\by$ and $\bx_{obs,tr}+\se\by$ respectively.
}
\end{figure} 
Note that this reflected wave front still has a beam width of order $\se$, and the observation point remains $\bx_{obs,ref}$. In other words, $U^{ref}$ corresponds to the specular reflected component: there is no significant perturbation of the specular reflected cone and of the angle $\theta^0_{ref}$ defined in \eqref{def:angle_inc_ref}. Finally, at $z=0$, the random phase shift  $2\bcals_0 \CalV(\by')$ appears in $U^{ref}$ through contributions scattered at the relative position $\by'$. 

Moreover, $U^{ref}$ can be rewritten as
\begin{align*}
U^{ref}(s, \by) & = \frac{\CalR}{2(2\pi)^3}\iiiint  e^{-i\omega (s - \bq \cdot \by)} K(2\bcals_0,\omega,\bq',\bq) \\
& \hspace{3cm} \times \hat \CalU_0(\omega,\bq,z_{int})\hat \CalU_0(\omega,\bq',z_{int}) \hat \Psi(\omega, \bq') \omega^2 d\omega d\by' d\bq' d\bq,
\end{align*}
where the scattering operator 
\begin{equation}\label{def:scat_op_ref}
K(\tau,\omega,\bq',\bq):=
\frac{\omega^2}{(2\pi)^2}
\int e^{i\omega(\bq'-\bq)\cdot \by'}e^{i\omega \tau \CalV(\by')}d\by'
\end{equation}
has the same law as $K^\e$ defined in \eqref{def:scat_op}, by stationarity of $V$: for any $\varphi,\psi \in \CalS(\R^2)$ (where $\CalS(\R^2)$ denotes the Schwartz class), the law of the random variable 
\[
K^\e(\tau,\omega,\varphi,\psi) := \iint K^\e(\tau,\omega,\bq',\bq) d\bq d\bq'
\]
coincides with that of 
\[
K(\tau,\omega,\varphi,\psi) := \iint K(\tau,\omega,\bq',\bq) d\bq d\bq',
\]
for any $\e$.
\begin{proof}[of Proposition \ref{prop:random_travel_time}]
The proof is in three steps. First, we prove convergence in $\CalS'(\CalX)$ (the space of tempered distributions), and then in $L^2(\CalX)$ equipped with the weak topology, using a density argument. Finally, convergence in $L^2(\CalX)$ equipped with the strong topology follows from an energy conservation relation.

For the first step, $(U^{\e,ref})_\e$ is viewed as a family of tempered distributions. At leading order in $\e$, the wave front $U^{\e,ref}$ reads
\begin{align*}
U^{\e,ref}(s, \by) & \simeq \frac{\CalR}{2(2\pi)^5}\iiiint e^{-i\omega (s- \bq \cdot \by)} e^{i \omega (\bq'- \bq) \cdot \by'}  e^{2i\omega \bcals_0 V(\bx_{int}/\se + \by')} \\
& \hspace{3cm} \times \hat \CalU_0(\omega,\bq,z_{int})\hat \CalU_0(\omega,\bq',z_{int}) \hat \Psi(\omega, \bq') \omega^4 d\omega d\by' d\bq' d\bq,
\end{align*}
for which we have
\begin{equation}\label{eq:lim_conservation}
\lim_{\e \to 0}\|U^{\e,ref}\|_{L^2(\CalX)} = \|U^{ref}\|_{L^2(\CalX)} = \frac{\CalR}{2} \|\Psi\|_{L^2(\CalX)}. 
\end{equation}
Using \eqref{eq:lim_conservation} together with \cite{mitoma}, tightness in $\CalS'(\CalX)$ follows from
\begin{equation}\label{eq:limsup_Uref}
\limsup_{\e \to 0} \big| \big\langle U^{\e,ref},\varphi \rangle_{L^2(\CalX)} \big| \leq \frac{\CalR}{2} \|\Psi\|_{L^2(\CalX)}\|\varphi\|_{L^2(\CalX)},
\end{equation}
for any $\varphi \in \CalS(\CalX)$. The convergence of finite-dimensional distributions follows from stationarity of $V$, which allows us to remove the shift $\bx_{int}/\se$, yielding for any $n\geq 1$ and $\varphi_1,\dots,\varphi_n \in \CalS(\CalX)$:
\[
\lim_{\e\to 0}\E\Big[  \prod_{j=1}^n \big\langle U^{\e,ref},\varphi_j \rangle_{L^2(\CalX)} \Big]  = \E\Big[  \prod_{j=1}^n \big\langle U^{ref},\varphi_j \rangle_{L^2(\CalX)} \Big]. 
\]
Therefore, $(U^{\e,ref})_\e$ is tight in $\CalS'(\CalX)$ and its finite-dimensional distributions converge to the appropriate limit; by \cite{mitoma}, $(U^{\e,ref})_\e$ converges in law in $\CalS'(\CalX)$.

To obtain convergence in $L^2(\CalX)$, we use Skorohod’s representation theorem \cite[Theorem 6.7, p.~70]{billingsley}: there exists a probability space on which $(\widetilde U^{\e,ref})_\e$ and $\widetilde U^{ref}$ are defined, with the same laws as $(U^{\e,ref})_\e$ and $U^{ref}$, respectively, and such that $(\widetilde U^{\e,ref})_\e$ converges almost surely to $\widetilde U^{ref}$ in $\CalS'(\CalX)$. Since $\CalS(\CalX)$ is dense in $L^2(\CalX)$, let $\varphi \in L^2(\CalX)$ and let $(\varphi_n)_n$ be a sequence of functions in $\CalS(\CalX)$ converging to $\varphi$ in $L^2(\CalX)$. Following the same lines as in \eqref{eq:limsup_Uref}, we obtain
\begin{align*}
\limsup_{\e\to 0} \big| \big\langle \widetilde U^{\e,ref}-\widetilde U^{ref},\varphi \rangle_{L^2(\CalX)} \big| & \leq \CalR\|\Psi\|_{L^2(\CalX)}\|\varphi-\varphi_n\|_{L^2(\CalX)} \\
& + \underbrace{\limsup_{\e\to 0} \big| \big\langle \widetilde U^{\e,ref}-\widetilde U^{ref},\varphi_n \rangle_{L^2(\CalX)}\big|}_{=0\text{ with probability 1}}.
\end{align*}
Passing to the limit in $n$ shows that $ (\big\langle \widetilde U^{\e,ref},\varphi \rangle_{L^2(\CalX)})_\e$ converges to $\big\langle U^{ref},\varphi \rangle_{L^2(\CalX)}$ with probability one. This implies almost sure convergence of $(\widetilde U^{\e,ref})_\e$ to $\widetilde U^{ref}$ in the separable space $L^2(\CalX)$ equipped with the weak topology. Finally, strong convergence follows from \eqref{eq:lim_conservation} (which holds on the new probability space) together with \cite[Theorem 3.32, p.~78]{brezis}. Since almost sure convergence implies convergence in law, we conclude that $(U^{\e,ref})_\e$ converges in law to $U^{ref}$ in $L^2(\CalX)$ equipped with the strong topology. \hfill $\square$
\end{proof}

The transmitted wave, observed at the same lateral position and in the same moving time frame as the specularly transmitted wave for a flat interface, is defined by
\[
U^{\e,tr}(s, \by) := u^{\e,tr}\big(t^\e_{obs,tr}(\by) + \e s,\bx_{obs,tr} +\se \by, z=z_{tr}\big)\qquad (s, \by)\in \CalX,
\]
where $\bx_{obs,tr}$ and $t^\e_{obs,tr}(\by)$ are defined by \eqref{def:pos_tr} and \eqref{def:time_tr}, respectively, and $u^{\e,tr}$ is given by the leading-order expression \eqref{eq:ue_tr}. Its asymptotic behavior is characterized by the following result, whose proof follows the same lines as for the reflected wave and is therefore omitted.
\begin{proposition} 
The family $(U^{\e,tr})_\e$ converges in law in $L^2(\CalX)$ to 
\begin{align*}
U^{tr}(s, \by) & := \frac{\CalT}{2(2\pi)^5}\sqrt{\frac{\bcals_0}{\bcals_1}}\iiiint e^{-i\omega (s- (\bcals_0-\bcals_1) \CalV(\by') - \bq \cdot \by)} e^{i \omega (\bq'- \bq) \cdot \by'} \\
& \hspace{3cm} \times \hat \CalU_1(\omega,\bq,z_{tr}-z_{int})\hat \CalU_0(\omega,\bq',z_{int}) \hat \Psi(\omega, \bq') \omega^4 d\omega d\by' d\bq' d\bq,
\end{align*}
where $\CalV$ is a random field with the same law as $V$.
\end{proposition} 

The transmitted wave front $U^{tr}$ also consists of a superposition of diffracted contributions observed at the point $\by$. Each contribution corresponds to waves emitted with initial lateral direction $\bq'$, propagating according to $\hat \CalU_0$, scattered into direction $\bq$ at the relative point $\by'$ as they cross the interface, and then propagating according to $\hat \CalU_1$ until reaching $z=z_{tr}$ (see Figure \ref{fig:ref_wave_rand}). As in the reflected case, the beam width remains of order $\se$, and the observation point remains $\bx_{obs,tr}$, so that there is no significant perturbation of the specular transmitted cone and of the angle $\theta^0_{tr}$ defined by \eqref{def:angle_tr}. Moreover, the random phase shift $(\bcals_0-\bcals_1)\CalV(\by')$ appears in $U^{tr}$ through contributions scattered at the relative position $\by'$.

The transmitted wave can also be rewritten in term of the scattering operator \eqref{def:scat_op_ref}  as
\begin{align*}
U^{tr}(s, \by) & := \frac{\CalT}{2(2\pi)^3}\sqrt{\frac{\bcals_0}{\bcals_1}}\iiiint e^{-i\omega (s- \bq \cdot \by)} K(\bcals_0-\bcals_1,\omega,\bq',\bq) \\
& \hspace{3cm} \times \hat \CalU_1(\omega,\bq,z_{tr}-z_{int})\hat \CalU_0(\omega,\bq',z_{int}) \hat \Psi(\omega, \bq') \omega^2 d\omega d\by' d\bq' d\bq,
\end{align*}
where $K$ is given by \eqref{def:scat_op_ref}.

\section{The case $\gamma > 1/2$: homogenized specular reflected and transmitted components}\label{sec:randint1}

The case $\gamma>1/2$ corresponds to a correlation length of the interface fluctuations that is smaller than the beam width. The incident wave is therefore strongly scattered and homogenization phenomena occur.

In this section, the following reflected wave front is considered
\ba\label{eq:Uref2}
U^{\e,ref}(s, \by, \widetilde\by) := u^{\e,ref}\big(t^\e_{obs,ref}(\by) + \e^\gamma \bk_0 \cdot\widetilde\by +\e s,\bx_{obs,ref} + \se \by + \e^\gamma \widetilde\by, z=0\big) \qquad (s, \by,\widetilde \by)\in \R \times \R^2\times \R^2,
\ea
where $u^{\e,ref}$ is given at leading order by \eqref{eq:ue_ref}, and $t^\e_{obs,ref}(\by)$ and $\bx_{obs,ref}$ are defined respectively by \eqref{def:pos_time} and \eqref{def:pos_time_ref}. We are therefore looking at the specular reflected component. The additional variable $\widetilde \by$ accounts here for variations at the scale of the correlation length. The asymptotic behavior of $U^{\e,ref}$ is described by the following result. For notational simplicity, we denote
\[\mathbb{X}=\R\times\R^2\times \R^2\]
in the remaining of the paper.
\begin{proposition}\label{prop:homog_ref}
The family $(U^{\e,ref})_\e$ converges in probability in $\CalS'(\mathbb{X})$ (the set of tempered distributions) to the deterministic pulse profile
\begin{align*}
U^{ref}(s, \by) & = \frac{\CalR}{2(2\pi)^3 }\iint e^{-i\omega(s-\bq\cdot \by)}   \hat \CalU_0(\omega, \bq,  2z_{int}) \phi_V (2\omega\bcals_0) \hat \Psi(\omega, \bq) \omega^2 d\omega  d\bq\\
& = \CalR \, \widetilde \CalU_0(s,\by,2z_{int}).
\end{align*}
Here,
\begin{equation}\label{def:charac_V}
\phi_V(u) := \E\big[e^{i u V(0)} \big]
\end{equation}
is the characteristic function of the random variable  $V(0)$, and $\hat \CalU_0$ is defined by \eqref{def:hCalU0}. 
The function $\widetilde \CalU_0$ satisfies \eqref{eq:paraxial} with initial condition
\[
\widetilde \CalU_0(s,\by,z=0)=\frac{1}{2} \Phi \ast_s \Psi(s,\by),
\]
where $\ast_s$ denotes convolution in the $s$-variable, and
\begin{equation}\label{eq:Phi}
 \Phi(s) := \frac{1}{2\pi}\int e^{-i\omega s}\phi_V(2\omega \bcals_0 )d\omega 
 = \frac{1}{2 \bcals_0}f_V\Big( \frac{s}{2 \bcals_0}\Big), 
\end{equation}
is a scaled version of the probability density function of $V(0)$, which we denote by $f_V$.
\end{proposition}
We now comment on the homogenized specularly reflected wave front. Note first that the variable $\widetilde\by$  was introduced in \eqref{eq:Uref2} in order to capture variations
on the finest spatial scale of the problem  corresponding to the scale of the interface variations. In fact in the small $\e$ limit, the 
 asymptotic profile does not depend on $\widetilde\by$ and therefore does not vary at the scale of the interface fluctuations. This behavior reflects the fact that the interface fluctuates much faster than the beam width, so that the limit becomes deterministic. The leading-order effect of the random interface is then a smoothing in time, determined by the marginal distribution of the interface height.
 Indeed, through $\widetilde \CalU_0$, the effect of the interface fluctuations on the specular reflection can be interpreted as a reflection problem for a flat interface, where the effective scattering properties are captured by the convolution in the initial condition $\Phi \ast_s \Psi/2$.  Note that, 
 compared to the case $\gamma=1/2$, the scattering operator \eqref{def:scat_op_ref} is now homogenized:
\[
\E[K(2\bcals_0,\omega,\bq',\bq)] = \frac{\omega^2}{(2\pi)^2} 
\int e^{i\omega(\bq'-\bq)\cdot\by'}\E\Big[e^{2i\omega\bcals_0 \CalV(\by')}\Big]d\by' = 
\delta(\bq'-\bq)\phi_V(2\omega\bcals_0).
\]
This homogenized scattering operator acts as for a flat interface: it does not modify the incident direction $\bq'$. The net effect of the interface fluctuations is a frequency-dependent low-pass filtering through the characteristic function $\phi_V$. This characteristic function depends only on the one-point statistics of the interface elevation due to stationarity of $V$. In this regime, the specularly reflected component still has a beam width of order $\se$, and the observation point remains $\bx_{obs,ref}$. Therefore, there is still no significant perturbation of the specular reflected cone and of the angle $\theta^0_{ref}$ given by \eqref{def:angle_inc_ref}.
  
\begin{proof}[of Proposition \ref{prop:homog_ref}]
The proof consists of two steps. First, we compute the first-order moment of $U^{\e,ref}$ to identify the homogenized limit. Second, we compute the second-order moment to prove convergence in probability.

Before evaluating the moment, we make the change of variable in \eqref{eq:ue_ref}
\[\bx' \to \bx_{int} + \se \, \by' + \e^\gamma \widetilde \by,\]
yielding at the leading order
\begin{align*}
U^{\e,ref}(s, \by, \widetilde\by) & \simeq \frac{\CalR}{2(2\pi)^5}\iiiint e^{-i\omega (s - \bq \cdot \by)}   e^{2i\omega  \bcals_0 V(\bx_{int}/\e^\gamma + \by'/\e^{\gamma-1/2}+\widetilde \by)} \\
& \hspace{3cm} \times e^{i \omega (\bq'-\bq)\cdot \by'} e^{i \omega \bq' \cdot \widetilde \by \e^{\gamma-1/2}} \hat \CalU_0(\omega, \bq, z_{int}) \hat \CalU_0(\omega, \bq', z_{int})  \hat \Psi(\omega, \bq') \omega^4 d\omega d\by' d\bq' d\bq.\nonumber
\end{align*}

\paragraph{First order moment.}
Taking expectations of $U^{\e,ref}$ and using stationarity of the fluctuations yield
\begin{align*}
\lim_{\e\to 0}\E[U^{\e,ref}(s, \by, \widetilde\by)] & = \frac{\CalR}{2(2\pi)^5 }\iiiint e^{-i\omega(s-\bq\cdot \by)}  \E\big[e^{2i\omega\bcals_0 V(0)} \big]   e^{i \omega (\bq'-\bq)\cdot \by'}\\
& \hspace{3cm} \times \hat \CalU_0(\omega, \bq,  z_{int})\hat \CalU_0(\omega, \bq', z_{int}) \hat \Psi(\omega, \bq') \omega^4 d\omega d\by' d\bq' d\bq\\
& = \frac{\CalR}{2(2\pi)^3 }\iint e^{-i\omega(s-\bq\cdot \by)}  \E\big[e^{2i\omega\bcals_0 V(0)} \big] \\
& \hspace{3cm} \times \hat \CalU_0(\omega, \bq,  2z_{int}) \hat \Psi(\omega, \bq) \omega^2 d\omega  d\bq\\
& = \CalR \, \widetilde \CalU_0(s,\by,2z_{int}).
\end{align*}

\paragraph{Second order moment.}
At leading order, the second order moments  of $U^{\e,ref}$ reads
\begin{align*}
\E[U^{\e,ref}(s_1, \by_1, \widetilde\by_1)&U^{\e,ref}(s_2, \by_2, \widetilde\by_2)] \\
& = \frac{\CalR^2}{4(2\pi)^{10}}\int\dots\int e^{-i\omega_1(s_1 - \bq_1\cdot\by_1)} e^{-i\omega_2(s_2 - \bq_2\cdot \by_2 )} \\
& \hspace{3cm} \times \E\big[e^{2i\omega_1\bcals_0 V(\bx_{int}/\e^\gamma + \by'_1/\e^{\gamma-1/2} + \widetilde \by_1)}e^{2i\omega_2\bcals_0 V(\bx_{int}/\e^\gamma + \by'_2/\e^{\gamma-1/2} + \widetilde \by_2)}\big] \\
& \hspace{3cm} \times e^{i \omega_1 (\bq'_1 - \bq_1) \cdot \by'_1} e^{i \omega_2 (\bq'_2 - \bq_2) \cdot \by'_2}\\
& \hspace{3cm} \times \hat \CalU_0(\omega_1,\bq_1,z_{int}) \hat \CalU_0(\omega_1,\bq'_1,z_{int}) \hat \CalU_0(\omega_2,\bq_2,z_{int})\hat \CalU_0(\omega_2,\bq'_2,z_{int}) \\
& \hspace{3cm} \times \hat \Psi(\omega_1, \bq'_1) \overline{\hat \Psi(\omega_2, \bq'_2)} \omega^4_1 \omega^4_2 d\omega_1 d\omega_2 d\by'_1 d\by'_2 d\bq'_1 d\bq'_2 d\bq_1 d\bq_2.
\end{align*}
Using stationarity of the fluctuations together with Lemma \ref{lem:mixing} yields
\begin{align*}
\lim_{\e\to 0}\E\big[e^{2i\omega_1\bcals_0 V(\bx_{int}/\e^\gamma + \by'_1/\e^{\gamma-1/2} + \widetilde \by_1)} & e^{2i\omega_2\bcals_0 V(\bx_{int}/\e^\gamma + \by'_2/\e^{\gamma-1/2} + \widetilde \by_2)}\big]\\
& = \lim_{\e\to 0}\E\big[e^{2i\omega_1\bcals_0 V( \by'_1/\e^{\gamma-1/2} +\widetilde\by_1)}e^{2i\omega_2\bcals_0 V(\by'_2/\e^{\gamma-1/2}+\widetilde\by_2) }\big]\\
& = \lim_{\e\to 0}\E\big[e^{2i\omega_1\bcals_0 V( \by'_1/\e^{\gamma-1/2} +\widetilde\by_1)}\big] \E\big[e^{2i\omega_2\bcals_0 V(\by'_2/\e^{\gamma-1/2}+\widetilde\by_2) }\big]\\
& = \E\big[e^{2i\omega_1\bcals_0 V(0 )}\big] \E\big[e^{2i\omega_2\bcals_0 V(0) }\big],
\end{align*}
so that
\[
\lim_{\e\to 0}\E\big[\big \langle U^{\e,ref},\varphi\rangle_{\CalS',\CalS}^2] = \E\big[\big \langle U^{ref},\varphi\rangle_{\CalS',\CalS}]^2 = \big \langle U^{ref},\varphi\rangle^2_{\CalS',\CalS}.
\]
As a result, from the Markov inequality, for any $\eta>0$
\begin{align*}
\lim_{\e\to 0} \Pro\big(&|\big \langle U^{\e,ref},\varphi\rangle_{\CalS',\CalS}- \big \langle U^{ref},\varphi\rangle_{\CalS',\CalS}|>\eta\big) \\
& \leq \frac{1}{\eta^2} \lim_{\e\to 0} \E\big[|\big \langle U^{\e,ref},\varphi\rangle_{\CalS',\CalS}-\big \langle U^{ref},\varphi\rangle_{\CalS',\CalS}|^2\big] \\
& \leq \frac{1}{\eta^2} \Big(\lim_{\e\to 0}\E\big[\big \langle U^{\e,ref},\varphi\rangle^2_{\CalS',\CalS}] + \big \langle U^{ref},\varphi\rangle^2_{\CalS',\CalS}  - 2 \big \langle U^{ref},\varphi\rangle_{\CalS',\CalS}\lim_{\e\to 0}\E\big[\big \langle U^{\e,ref},\varphi\rangle_{\CalS',\CalS}]\Big) = 0,
\end{align*}
which concludes the proof. 
\hfill $\square$
\end{proof}

The strategy used above to characterize the specular reflected wave also applies to characterize the specular transmitted wave front:
\[
U^{\e,tr}(s, \by, \widetilde\by) := u^{\e,tr}\big(t^\e_{obs,tr}(\by) + \e s,\bx_{obs,tr} + \se \by + \e^\gamma \widetilde\by, z=0\big) \qquad (s, \by)\in \R \times \R^2\times \R^2,
\]
where $u^{\e,tr}$ is given at leading order by \eqref{eq:ue_tr}, and $\bx_{obs,tr}$ and $t^\e_{obs,tr}(\by)$ are defined in \eqref{def:pos_tr} and \eqref{def:time_tr}, respectively. The asymptotic behavior of $U^{\e,tr}$ is described by the following result. 

\begin{proposition}
The family $(U^{\e,tr})_\e$ converges in probability in $\CalS'(\mathbb{X})$ to the deterministic pulse profile

\begin{align*}
U^{tr}(s, \by) & = \frac{\CalT}{2(2\pi)^3}\sqrt{\frac{\bcals_0}{\bcals_1}}\iint e^{-i\omega(s - \bq \cdot \by)}   \hat \CalU_1(\omega,\bq,z_{tr}-z_{int}) \hat \CalU_0(\omega,\bq,z_{int}) \nonumber \\ 
& \hspace{3cm}\times  \phi_V(\omega(\bcals_0-\bcals_1))  \hat \Psi(\omega, \bq) \omega^2 d\omega d\bq,\nonumber
\end{align*}
where $\phi_V$ is defined by \eqref{def:charac_V}, and the functions $\hat \CalU_j$ are defined by \eqref{def:hCalU0} and \eqref{def:hCalU1}.
\end{proposition}
The proof of this result is omitted, as it follows the same lines as that of Proposition \ref{prop:homog_ref}. Note that the profile $U^{tr}$ is similar to the one obtained for a flat interface \eqref{def:Utr}, but with an  initial condition 
that is smoothed in time via the kernel $\phi_V$
reflecting the homogenization effects induced by the highly oscillatory rough interface.

As for the reflected wave, the effective scattering operator is given by
\[
\E[K(\bcals_0-\bcals_1,\omega,\bq',\bq)] = \frac{\omega^2}{(2\pi)^2} \int e^{i\omega(\bq'-\bq)\cdot\by'}\E\Big[e^{i\omega(\bcals_0-\bcals_1) \CalV(\by')}\Big]d\by' = \delta(\bq'-\bq)\phi_V((\bcals_0-\bcals_1)\omega)
\] 
as a result of the homogenization process. Thus, $\phi_V$ accounts for the homogenized diffraction effects in transmission by acting as a low-pass filter in frequency, while the Dirac mass indicates that the incident direction $\bq'$ is not affected. As for the reflected wave, the specular transmitted component observed at $\bx_{obs,tr}$ still has a beam width of order $\se$. Therefore, there is no significant perturbation of the specular transmitted cone or of the angle $\theta^0_{tr}$ given by \eqref{def:angle_tr}.

\section{The case $\gamma>1/2$: incoherent wave fluctuations for the nonspecular reflected contributions}

As a consequence of the homogenization phenomenon described in the previous section, an effective frequency-dependent attenuation is observed on the specular components through the factor $\phi_V$. The purpose of this section is to describe how the scattering phenomena generate incoherent wave fluctuations that account for the missing energy not carried by the specular components. 

In this section, we analyze the following reflected wavefield, referred to as the speckle profile:
\begin{equation}\label{def:Se_ref}
S^{\e,ref}(\bar s, \bar \by, \by, \widetilde s, \widetilde \by) := \e^{1-2\gamma} u^{\e,ref}(t^\e_{obs,ref}(\bar s,\bar \by, \by, \widetilde\by) + \e \widetilde s, \bx^\e_{obs, ref}(\bar \by, \by) + \e^\gamma \widetilde\by , z=0),
\end{equation}
where
\begin{equation}\label{def:x_obs_eps}
\bx^\e_{obs, ref}(\bar \by, \by) := \bx_{obs, ref} + \e^{1-\gamma} \bar \by + \se \, \by,
\end{equation}
and
\begin{equation}\label{def:t_obs_eps}
t^\e_{obs,ref}(\bar s, \bar \by, \by,\widetilde\by):= t_{obs,ref} + \e^{1-\gamma}\bk_0\cdot \bar \by + \se\, \bk_0\cdot \by + \e^\gamma \bk_0\cdot\widetilde \by + \e^{2(1-\gamma)} \bar s.
\end{equation}
The proof of Proposition \ref{prop:exp_corr} below is carried out in a way that justifies the scalings introduced in \eqref{def:Se_ref}–\eqref{def:t_obs_eps}. Tracking the relevant scales throughout this derivation allows us to provide a physical interpretation of these scalings and we do so below. 

Recalling \eqref{eq:scaling_param}, the wavefield $u^{\e,ref}$ is observed around the reference observation location $\bx_{obs,ref}$ (corresponding to the specular reflected component), but now over a neighborhood of size the relative roughness parameter
\begin{equation}\label{eq:rel_eps}
\e^{1-\gamma} = \frac{\lambda}{\ell_c} \sim \frac{r_0^2}{\ell_c L} 
\end{equation}
This neighborhood is parameterized by the variable $\bar{\by}$ and follows from the paraxial scaling relation $r_0^2 \sim \lambda L$. Recall that $r_0$ is the beam radius, $\ell_c$ the correlation length, and $\lambda$ the central wavelength. The radius of the region over which the reflected speckle can be observed is therefore of order $\lambda/\ell_c$. This region corresponds to the speckle cone and thus contains the specular cone as a smaller neighborhood around the reference observation point. Indeed, this neighborhood is larger than the beam width (of order $\se$) since $\gamma>1/2$, and can even be of order one when $\gamma=1$. We will see below that this radius is additionally proportional to the distance from the source plane to the interface, which is of order one in our scaling.

The variables $\by$ and $\widetilde \by$ account for variations at the scale of the beam width and at the scale of the correlation length respectively (see Figure \ref{fig:ref_wave_speckle} for an illustration).
 \begin{figure}
\begin{center}
\includegraphics*[scale=0.25]{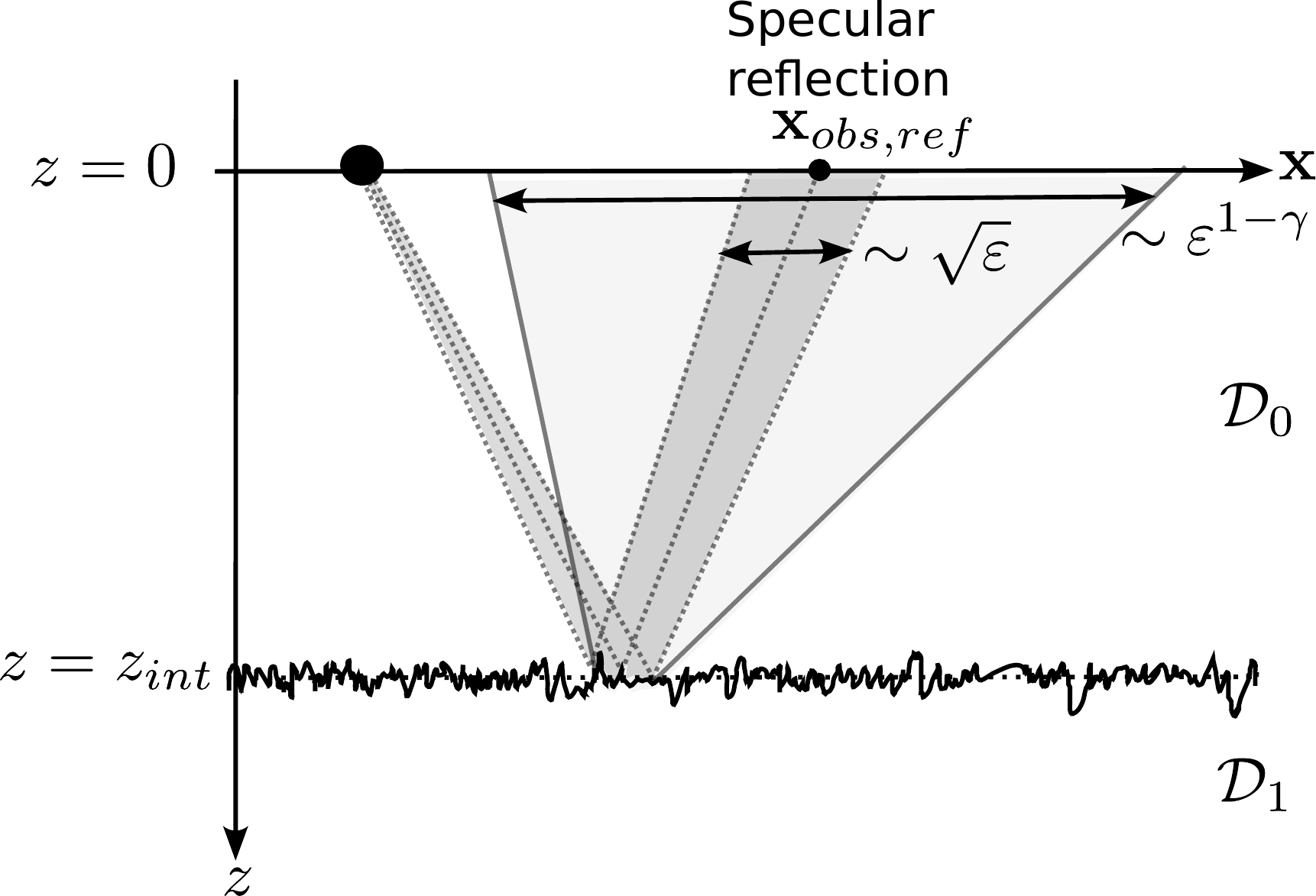} \hspace{1.5cm}
\includegraphics*[scale=0.25]{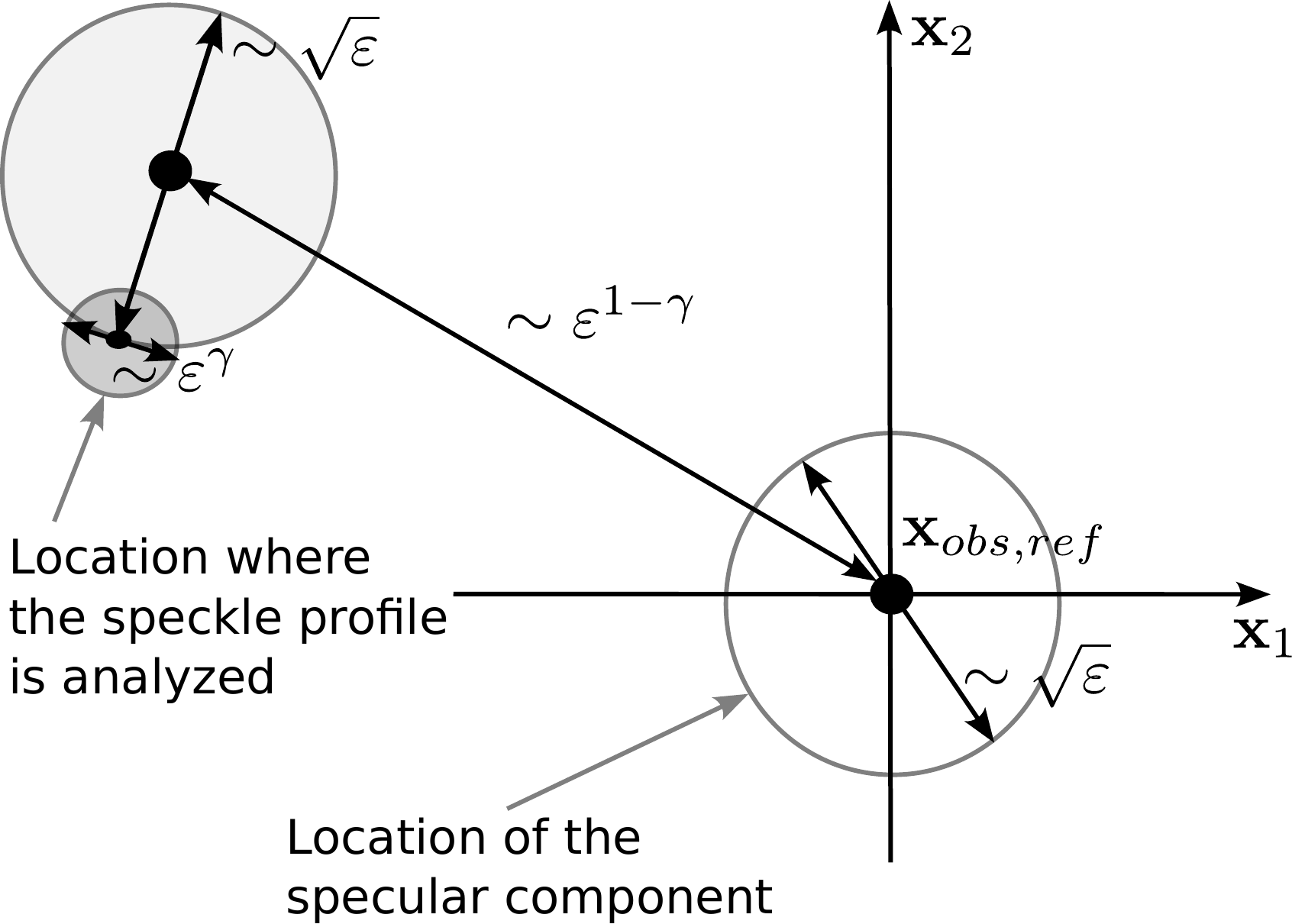} 
\end{center}
\caption{\label{fig:ref_wave_speckle} The left panel illustrates the region where the speckle profile can be observed (light gray) compared to the specular reflection component (gray). The right panel illustrates the spatial windows over which the speckle profile is analyzed.
}
\end{figure}
The relevant time scale is centered around the travel time $t_{obs,ref}$, with three corrections depending on the observation offsets $(\bar \by, \by, \widetilde \by)$ (see Figure \ref{fig:scaling2}), and lies within a time window larger than the pulse duration (of order $\e$).
\begin{figure}
\begin{center}
\includegraphics*[scale=0.18]{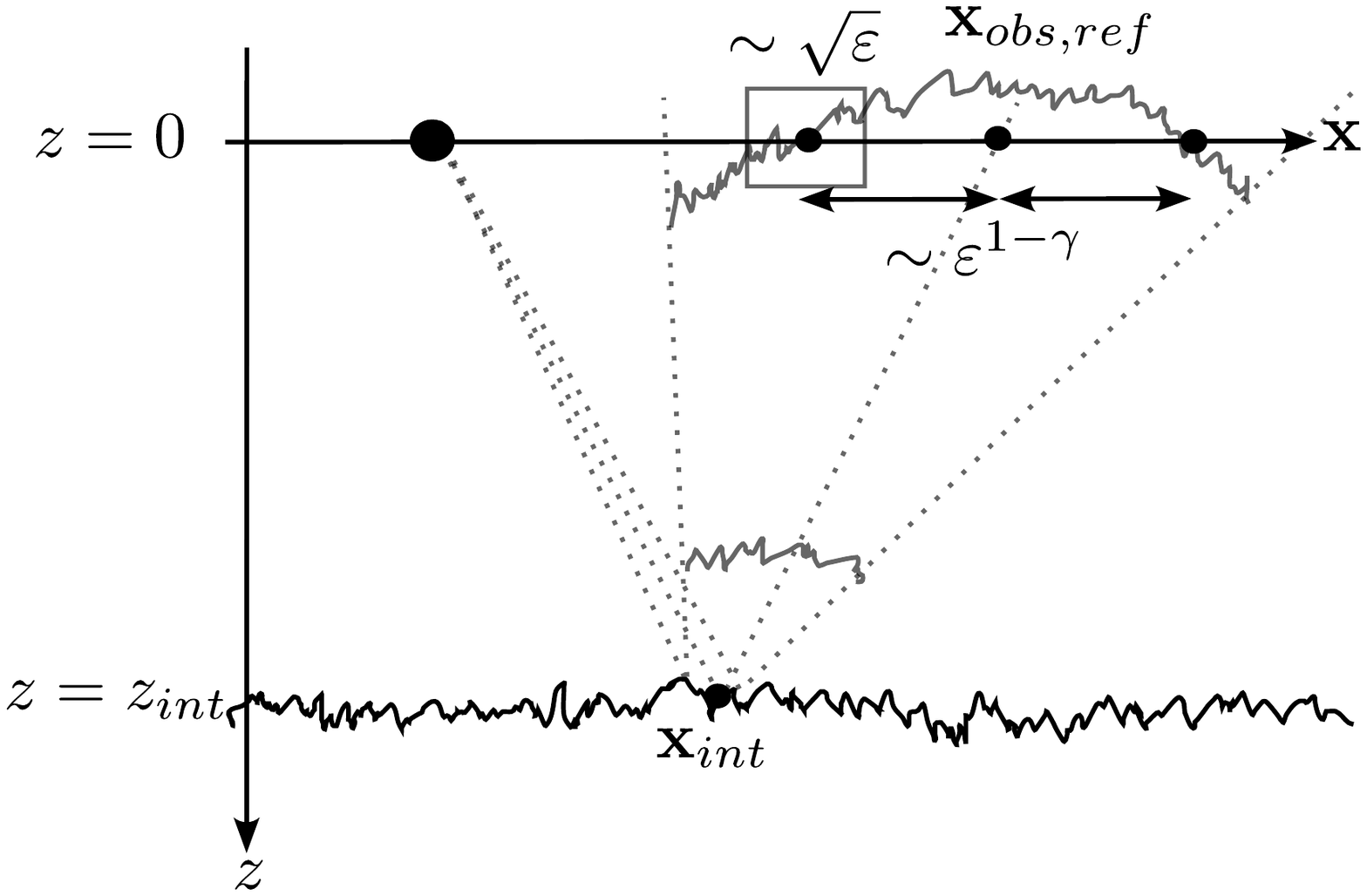} 
\includegraphics*[scale=0.18]{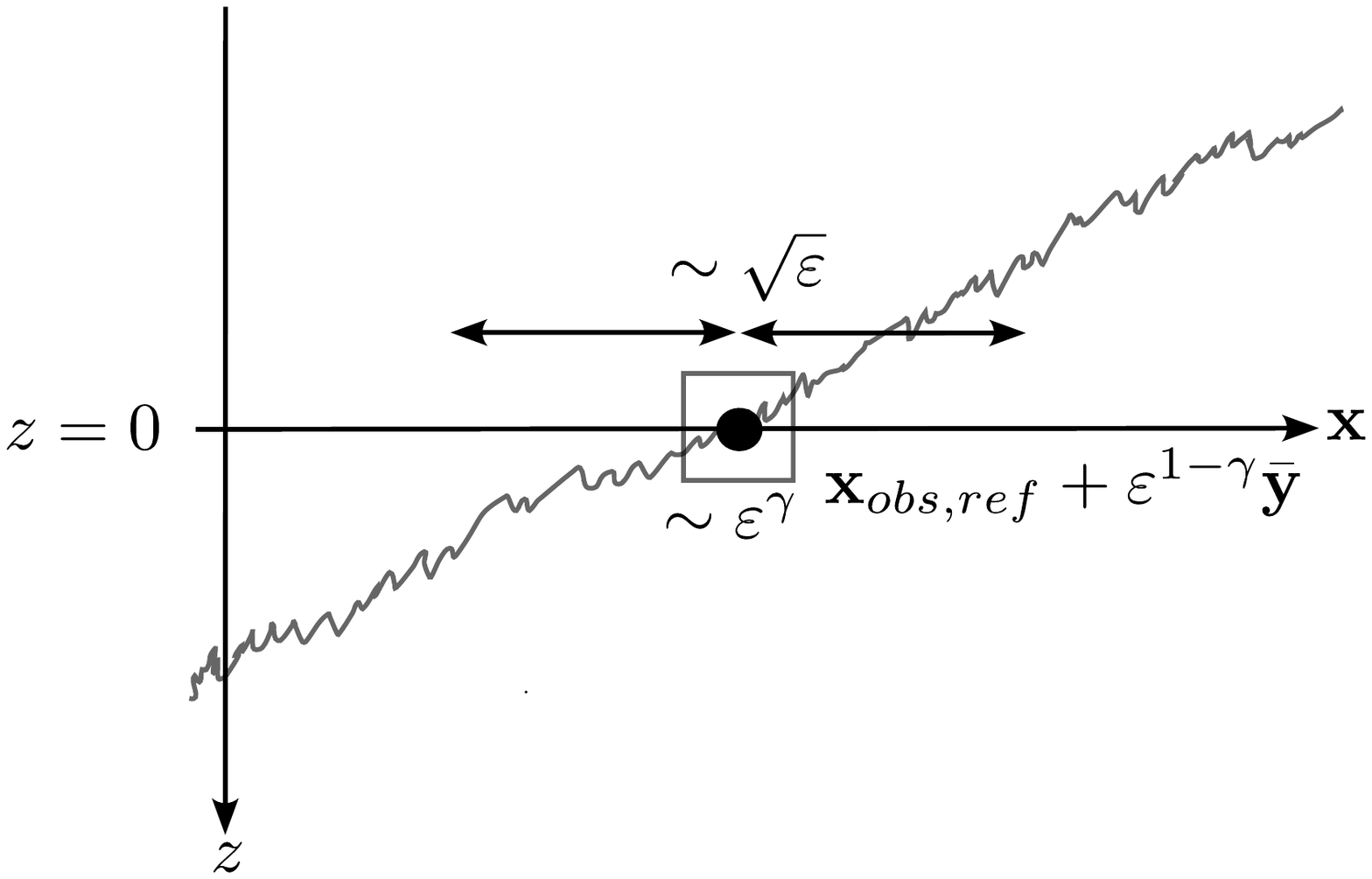} 
\includegraphics*[scale=0.18]{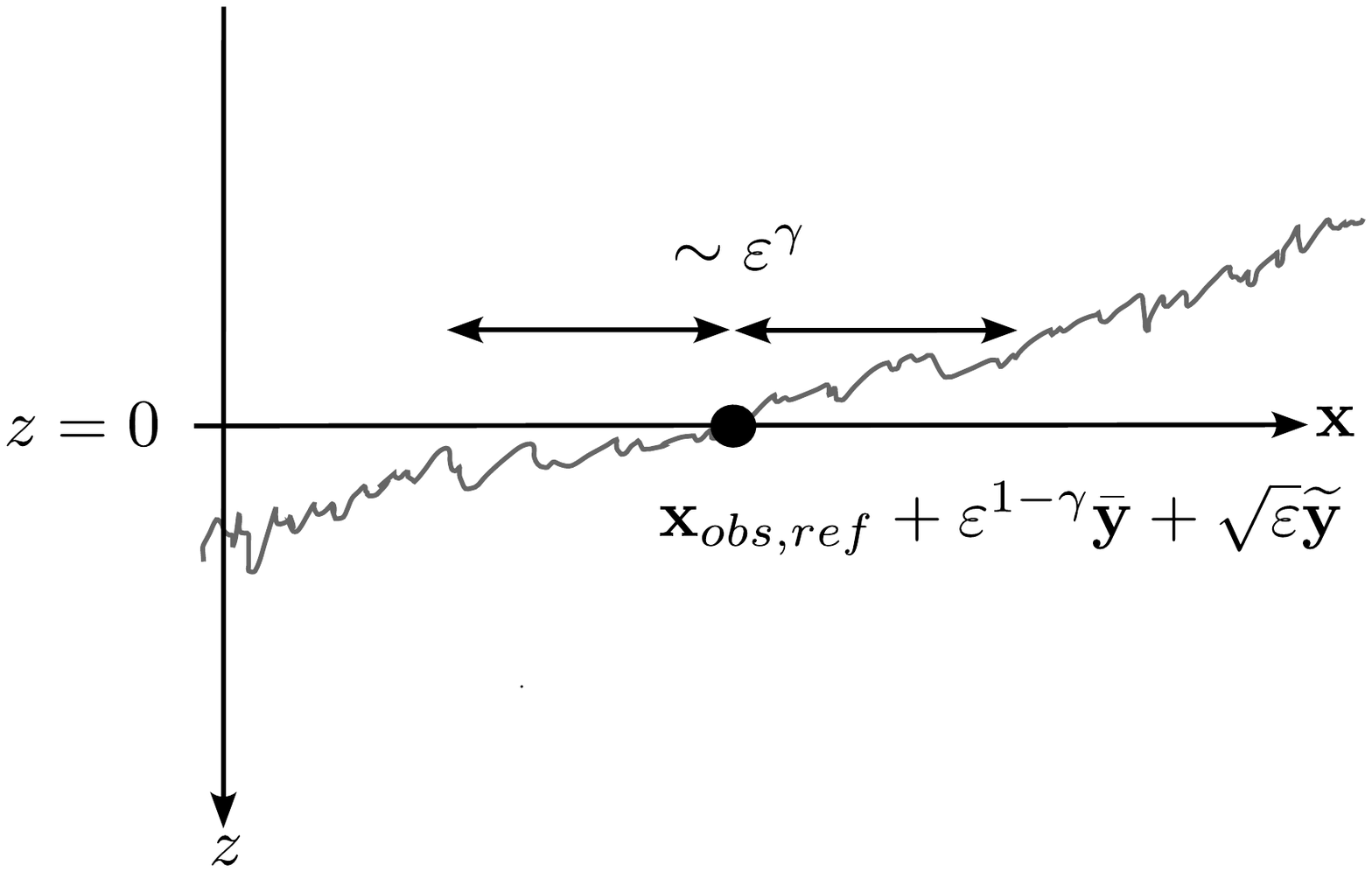} 
\end{center}
\caption{\label{fig:scaling2}
 The figure illustrates the terms in \eqref{def:t_obs_eps}. 
In the left panel we illustrate the speckle wavefront, which spreads at the scale of the relative roughness parameter $\e^{1-\gamma}$. 
The middle panel shows a zoom-in to the scale of the beam width $\se$, which corresponds to the scale of the specular cone, 
while the right panel provides a further zoom-in to the correlation length $\e^\gamma$, which is the scale of the variations 
of the speckle supported within the speckle cone. 
All three scales are larger than the pulse width $\e$, which is represented in the figure by the magnitude of the random
 fluctuations in the wavefront line. 
All these scalings must be accounted for in the observation time in order to accurately represent the incoherent wave fluctuations 
at different locations within the speckle cone.   
}
\end{figure} 
This time window, described by the variable $\bar s$, is of order
\[
\e^{2(1-\gamma)} = \frac{r_0^2 c_0 T_0}{\ell_c^2 L} \sim  \Big(\frac{\lambda}{\ell_c}\Big)^2,
\]
that is, the square of the relative roughness parameter, and corresponds to the order of magnitude of the speckle duration over the observation plane. Here, $T_0$ denotes the pulse duration. The variable $\widetilde s$ in \eqref{def:Se_ref} accounts for fluctuations at the scale of the pulse width. 

To summarize, for the specular component the temporal duration is $\CalO(\e)$ and the spatial width of the specular cone
is $\CalO(\se)$; these are also the coherence scales of the specular component. In contrast, the speckle component has a temporal duration of $\CalO(\e^{2(1-\gamma)})$ with a coherence time $\CalO(\e)$. 
Moreover, the width of the speckle front is $\CalO(\e^{1-\gamma})$ with a spatial coherence length $\CalO(\e^\gamma)$. Thus, the temporal support of the wavefield is enhanced by a factor $\e^{1-2\gamma}$, and the spatial support in each of the two lateral
directions is enhanced by a factor $\e^{(1-2\gamma)/2}$. From energy conservation considerations this corresponds to a wavefield amplitude scaling of
\[
\e^{-(1-2\gamma)}.
\]
This motivates the large amplitude prefactor $\e^{1-2\gamma}\gg 1$ in \eqref{def:Se_ref}. 
Note  that increasing $\gamma$,  so that the interface is rougher,  enlarges the spatial region around $\bx_{obs,ref}$ over which the speckle pattern can be observed, while reducing its magnitude. 

We next analyze the asymptotic behavior of the second-order statistics and the intensity of the speckle field $S^{\e,ref}$, as well as the self-averaging properties exhibited by these quantities. The expectation of $S^{\e,ref}$, which converges to zero, will be discussed in Section \ref{sec:stat_speckle}, where the asymptotic behavior of the speckle field itself is characterized as a mean-zero Gaussian random field.

\subsection{Correlation and intensity functions}

The above $\e$ scalings of the speckle profile $S^{\e,ref}$ can be understood  by a careful analysis of its correlation function. This correlation function is defined as the product of $S^{\e,ref}$ evaluated at two nearby points:
\begin{equation}\label{def:corr_func}
\CalC^{\e, ref}(\bar s,\bar \by, \by,\widetilde s_1,\widetilde \by_1, \widetilde s_2,\widetilde \by_2) : = S^{\e, ref}(\bar s,\bar \by, \by,\widetilde s_1,\widetilde \by_1)S^{\e, ref}(\bar s,\bar \by, \by,\widetilde s_2,\widetilde \by_2),
\end{equation}
with associated intensity
\[
\CalI^{\e, ref}(\bar s,\bar \by, \by,\widetilde s,\widetilde \by) : = \big\vert S^{\e, ref}(\bar s,\bar \by, \by,\widetilde s,\widetilde \by)\big\vert^2 = \CalC^{\e, ref}(\bar s,\bar \by, \by,\widetilde s,\widetilde \by, \widetilde s,\widetilde \by).
\]
The following result characterizes the asymptotic mean correlation function. Note that we use below the notation
\[
\CalX=\R\times\R^2\qquad\text{and}\qquad \mathbb{X} = \R\times\R^2\times \R^2.
\]

\begin{proposition}\label{prop:exp_corr}
For $\gamma \in (1/2,1]$ and $\bar \by \neq 0$, we have 
\begin{equation*}
\lim_{\e\to 0}\E\big[\CalC^{\e, ref}(\bar s,\bar \by, \by,\widetilde s_1,\widetilde \by_1, \widetilde s_2,\widetilde \by_2)\big] = 
\bC^{ref}(\bar s,\bar \by,\widetilde s_1-\widetilde s_2,\widetilde \by_1-\widetilde \by_2),
\end{equation*}
in $\CalS'(\mathbb{X}\times \CalX \times \CalX)$, where
\begin{align}
\label{def:Cor_func}
\bC^{ref}(\bar s,\bar \by, \widetilde s,\widetilde \by) := \frac{\CalR^2}{4(2\pi)^3}\iint & e^{-i\omega (\widetilde s - \bp\cdot\widetilde \by)}   \CalA(2\bcals_0,\omega,\bp) |\hat \Psi|^2_2(\omega) \\
& \times \delta(\bar s - z_{int} c_0 \bp^T A_0 \bp/2) \delta(\bar \by-z_{int}c_0 A_0 \bp) \omega^2 d\omega d\bp,\nonumber
\end{align}
with
\begin{equation}
\label{def:CalA} \CalA(v,\omega,d\bp) := \int \E\big[ e^{i v ( V(\by') - V(0)) } \big] e^{-i\omega \bp \cdot\by'} d\by',
\end{equation}
and
\begin{equation}
\label{def:hPsi2}
|\hat \Psi|^2_2(\omega) :=  \frac{\omega^2}{(2\pi)^2}\int | \hat \Psi(\omega , \bq)|^2  d\bq = \int \Big|\int e^{i\omega s} \Psi(s,\bx)ds \Big|^2 d\bx.
\end{equation}
\end{proposition}

In \eqref{def:Cor_func}, the spatial window in $\bar \by$ is larger than the beam width, and the contribution of the source is integrated over all supported lateral directions (see \eqref{def:hPsi2}). The term $\CalA(\cdot,\omega,\bp)$ corresponds to the distribution of scattered directions $\bp$ produced by the random interface at frequency $\omega$, and is related to variations of the interface at the scale of the correlation length through the integral with respect to $\by'$ in \eqref{def:CalA}. According to the Bochner--Khinchin theorem, $\CalA$ defines a finite positive measure on $\R^2$ that we call \emph{scattering distribution} in this paper. Therefore, properly normalized $\CalA$ can be recast as a probability measure. For notational simplicity, even though $\CalA$ is a measure in $\bp$, we still use the notation $\CalA(v,\omega,\bp)$ throughout the remaining of this paper. Moreover, this distribution couples the variables $\widetilde s$ and $\widetilde \by$, which represent small-scale fluctuations generated by the random interface, evolving respectively at the scale of the pulse width (wavelength scale) and the correlation length of the interface. The proof of Proposition \ref{prop:exp_corr} is postponed to Section \ref{sec:proof_prop_exp_cor}.

We now provide two complementary interpretations of \eqref{def:Cor_func}. First, for a fixed offset position $\bar \by$, using the Dirac distributions in $\bp$, we obtain
\begin{equation}\label{def:Cor_func_bis}
\bC^{ref}(\bar s,\bar \by, \widetilde s,\widetilde \by) = \frac{\CalR^2c_0^2\bcals^4_0}{4(2\pi)^3 z^2_{int}}\delta\Big(\bar s - \frac{\bar \by^T  A_0^{-1}\bar \by}{2z_{int}c_0}\Big) \int  e^{-i\omega (\widetilde s - \bar \by ^T A_0^{-1} \widetilde \by/(z_{int}c_0))}  |\hat \Psi|^2_2(\omega)  \CalA\Big(2\bcals_0,\omega,\frac{A_0^{-1}\bar \by}{z_{int}c_0}\Big)  \omega^2 d\omega,
\end{equation}
with
\begin{equation}\label{def:A0_inv}
A_0^{-1} = c_0 \bcals_0(\mathbf{I}_2 - c_0^2 \bk_0 \otimes \bk_0) ,
\end{equation}
which is positive definite as we assume $|\bk_0|<c_0^{-1}$. As a result, for a given offset position  $\bar \by$, the relative  time delay  parameter at which the speckle is observed is explicitly given by
\[
\bar s = \frac{\bar \by^T A_0^{-1}\bar \by}{2z_{int}c_0} \ge 0.
\]
Conversely, for a given time $\bar s$, the incoherent wave fluctuations are supported over an ellipse defined by $\bar s = \bar \by^T  A_0^{-1}\bar \by/(2z_{int}c_0)$, see Figure \ref{fig:speckle} for an illustration.  
\begin{figure}
\begin{center}
\includegraphics*[scale=0.25]{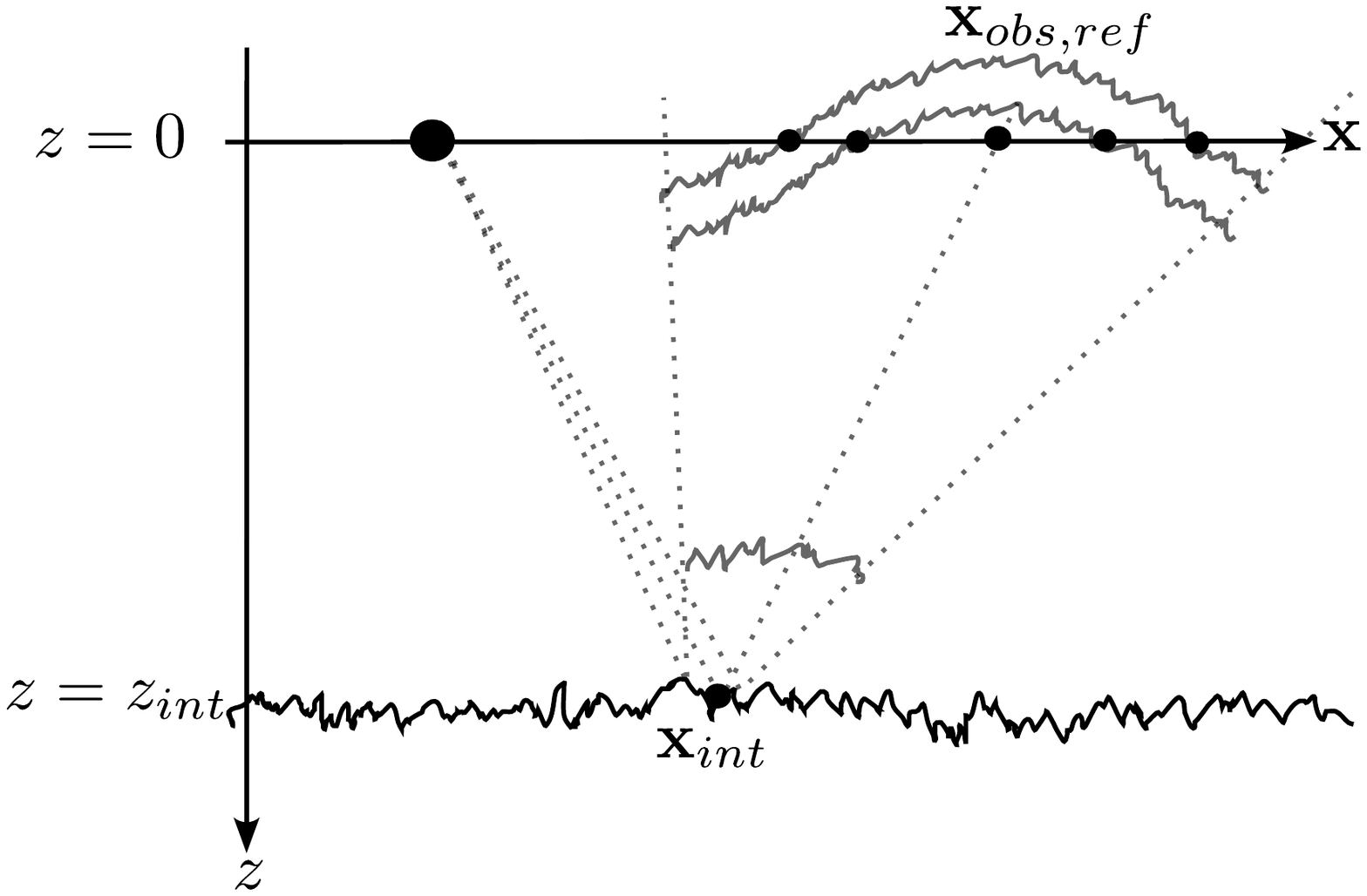} 
\includegraphics*[scale=0.35]{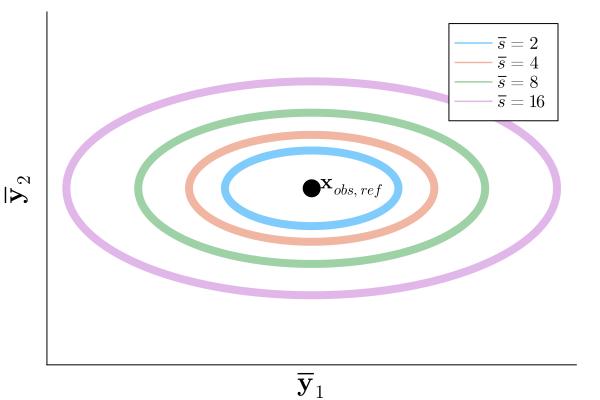} 
\end{center}
\caption{\label{fig:speckle} The left panel illustrates the reflection of a pulse by a rough interface. When the speckles cross a horizontal plane, they form ellipses that evolve in time, illustrated in the right panel for different values of $\bar s$, and centered around $\bx_{obs,ref}$. For the right panel we take $\bk_0=(0.9/c_0,0)^T$.
}
\end{figure} 

Since $\bC^{ref}$ does not depend on the variable $\by$ (associated with variations at the beam-width scale), the statistics of these ellipses are identical over regions of size comparable to the beam width for a given $\bar s$. However, as we will see in Corollary \ref{cor:vector_speckle}, when the speckle is observed at two distinct points separated at the beam-width scale (that is we consider two distinct points $\by_1$ and $\by_2$ instead of a single $\by$), their contributions are statistically independent. Moreover, the correlation function decays on a temporal scale corresponding to the pulse duration and on a spatial scale corresponding to the correlation length of the interface fluctuations, through the variables $\widetilde s$ and $\widetilde \by$, respectively. Thus, $\bC^{ref}(\bar s,\bar \by,\widetilde s_1-\widetilde s_2,\widetilde \by_1-\widetilde \by_2)$ represents the mean correlation function at two nearby points and exhibits asymptotic stationarity with respect to the small-scale fluctuation variables $\widetilde s$ and $\widetilde \by$. These two variables are coupled in \eqref{def:Cor_func_bis} through the complex exponential in the Fourier representation. The decay of correlations with respect to these variables is ensured by the frequency bandwidth of the source profile together with the integration over $\omega$, which reflects the frequency coupling induced by scattering at the interface.

Similarly, the mean intensity carried by the speckle profile is given by
\[
\lim_{\e\to 0}\E[\CalI^{\e, ref}(\bar s,\bar \by,\by,\widetilde s,\widetilde \by)] = \bC^{ref}(\bar s,\bar \by,0,0),
\] 
which does not depend on $\by$, nor on the small-scale fluctuation variables $\widetilde s$ and $\widetilde \by$. The intensity is therefore uniform over these scales and evolves only with respect to $\bar s$ and $\bar \by$.

Finally, we note that the specular component (corresponding to $\bar \by=0$) does not by itself contribute to $\bC^{ref}$. Indeed, the limit $\bC^{ref}$ is obtained in $\CalS'(\mathbb{X}\times \CalX \times \CalX)$, which involves integration over $\bar \by$ on regions much larger than the size of the specular cone. In this framework, the explicit contribution of the specular component can be shown to be negligible. However, in the direction of the specular cone  one will observe the specular  reflection as discussed in Sections 
\ref{sec:specint} and \ref{sec:randint1} and this component is indeed large  compared to the speckle component discussed here.

\paragraph{Generalized Snell's law of reflection.} 

The correlation function \eqref{def:Cor_func} allows for the derivation of a generalized Snell's law of reflection. From the Dirac mass in \eqref{def:Cor_func}, the resulting speckle can be observed, for a given direction $\bp$, at relative position
\ba\label{eq:off}
\bar \by = \by^{ref}_\bp := z_{int} c_0 A_0 \bp,
\ea
and at the corresponding time lag 
\[
\bar s= s^{ref}_\bp := \bp \cdot \by^{ref}_\bp /2 = z_{int} c_0 \bp^T A_0 \bp/2 \geq 0.
\]
These two quantities depend on the distance $z_{int}$ from the source to the interface. This distance therefore influences the size of the reflected speckle cone as well as the time window over which its evolution can be observed.

After some algebra, one can generalize the standard reflection relation between the incident and reflected angles \eqref{def:angle_inc_ref} as follows. For a given angular frequency $\omega$ and a nonzero slowness vector $\bk_0$ (so that $\theta_{inc}>0$), we have in view of the parameterization (\ref{def:x_obs_eps}) and  the expression  (\ref{eq:off})
 \begin{align}
\frac{\tan(\theta^\e_{ref}(\bp))}{\tan(\theta_{inc})} & = \frac{|\bx_{int} + \e^{1-\gamma}\by_\bp^{ref}|}{|\bx_{int}|} \nonumber \\
&  = \sqrt{\Big(1+ \xi_\e \frac{\bp\cdot \bk_0}{\cos^2(\theta_{inc})} \Big)^2+ \xi^2_\e (\bp\cdot \bk_0^{\perp})^2},
\label{eq:inc_ref_angle}
\end{align} 
where
\[
\xi_\e = \e^{1-\gamma}\frac{c_0^2}{\sin^2(\theta_{inc})}, \qquad \bk_0^\perp = (-\bk_{0,2}, \bk_{0,1})^T, 
\]
and for $\bp$ distributed according to the scattering distribution $\CalA(2\bcals_0,\omega,\cdot)$. 

Finally, we obtain
\begin{equation*}
\sin(\theta^\e_{ref}(\bp)) = \sin(\theta_{inc}) + \sin(\theta_{inc})  \left(\sqrt{\frac{\Xi^\e_{ref}(\bp)}{1+\sin^2(\theta_{inc})\Big( \Xi^\e_{ref}(\bp)-1\Big)}} -1\right),
\end{equation*} 
where 
\[
\Xi^\e_{ref}(\bp) = \Big(1+\xi_\e\frac{\bp\cdot\bk_0}{\cos^2(\theta_{inc})}\Big)^2 + \xi^2_\e(\bp\cdot \bk_0^\perp)^2 . 
\]
yielding \eqref{eq:ref_formula_intro} using \eqref{eq:rel_eps} and that $\sin(\theta_{inc})=c_0|\bk_0|$ (recalling \eqref{def:tauj} and  \eqref{def:angle_inc_ref}).
 
We refer to Figure \ref{fig:ref_angle} for an illustration of $\theta^\e_{ref}(\bp)$ when $\gamma=1$, i.e., in the rough-interface regime $\lambda \sim \ell_c$ for which the opening angle of the speckle cone is $\CalO(1)$.
\begin{figure}
\begin{center}
\includegraphics*[scale=0.35]{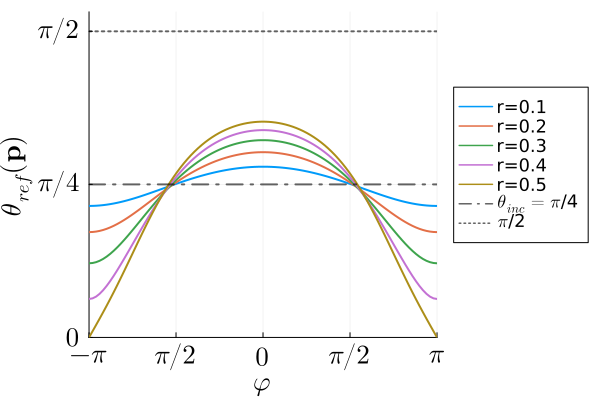} \hspace{1.5cm}
\includegraphics*[scale=0.35]{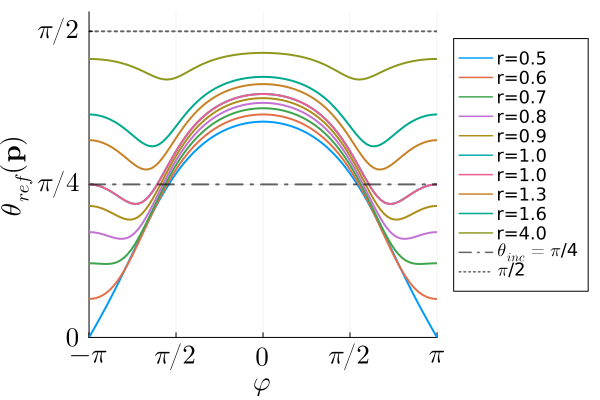} 
\end{center}
\caption{\label{fig:ref_angle} Illustration of the incident angle $\theta_{inc}$ (dash-dot line), describing the cones defined by the incoming beam and the reflected specular component, and the reflected azimuthal angle $\theta_{ref}(\bp)\in (0,\pi/2)$, for $\gamma=1$ and $\bp = \beta_1 \bk_0 + \beta_2 \bk_0^\perp$, with $(\beta_1,\beta_2)= r \, (\cos(\varphi),\sin(\varphi))$.
}
\end{figure}

The case $\gamma\in(1/2,1)$ leads to relatively small opening angles for the speckle cone and we have 
the expansion 
\begin{equation*}
\sin(\theta^\e_{ref}(\bp)) = \sin(\theta_{inc}) + \e^{1-\gamma}\frac{c_0 \bp\cdot \widehat \bk_0}{\cos^2(\theta_{inc})} + \CalO(\e^{2(1-\gamma)} \vert \bp \vert ^2),
\end{equation*} 
with $\widehat \bk_0 = \bk_0 /|\bk_0|$, which corresponds to \eqref{eq:correction} using again \eqref{eq:rel_eps}. Note that we also have by truncated expansion   
\begin{equation*}
\theta^\e_{ref}(\bp) = \theta_{inc} + \e^{1-\gamma}\frac{ c_0 \bp\cdot \widehat \bk_0}{\cos^3(\theta_{inc})} + \CalO(\e^{2(1-\gamma)} \vert \bp \vert ^2),
\end{equation*} 
corresponding to \eqref{eq:correction_angle_ref}. 

We remark that the lateral support  of the reflected speckle  component then is of order $\CalO(\e^{1-\gamma} L)$, where $L$ denotes the distance from the source to the interface.  This can be contrasted with the anomalous beam spreading in the case of the It\^o-Schr\"odinger equation with propagation through a random medium with bulk scattering for which the anomalous beam spreading corresponds to the width 
$\CalO(\e^{1/2} L^{3/2})$ \cite{51}.
 
For $\bk_0=0$ (so that $\theta_{inc}=0$), the reflection formula reduces to $\tan(\theta^\e_{ref}(\bp)) = \e^{1-\gamma} c_0 |\bp|$, that is,
\ba\label{eq:tan1}
\theta^\e_{ref}(\bp) = \arctan(\e^{1-\gamma} c_0 |\bp|).
\ea

The following result describes the statistical stabilization of $\CalC^{\e,ref}$ as $\e\to 0$. In other words, when correlating the speckle profile at two nearby points with respect to the variables $\widetilde s$ and $\widetilde \by$, its limit in the high-frequency regime $\e \to 0$ becomes deterministic and coincides with the mean correlation function given in Proposition~\ref{prop:exp_corr}.
\begin{proposition}\label{prop:cv_corr}
We have
\[
\lim_{\e\to 0}\CalC^{\e,ref}(\bar s,\bar \by, \by,\widetilde s_1,\widetilde \by_1, \widetilde s_2,\widetilde \by_2) = \bC^{ref}(\bar s,\bar \by,\widetilde s_1-\widetilde s_2,\widetilde \by_1-\widetilde \by_2)
\]
in probability in $\CalS'(\mathbb{X}\times\CalX\times\CalX)$, where the deterministic limit $\bC^{ref}$ is defined by \eqref{def:Cor_func}.
\end{proposition} 
The same result holds for the intensity function, which also converges in probability in $\CalS'(\mathbb{X}\times\CalX\times\CalX)$:
\[
\lim_{\e\to 0}\CalI^{\e,ref}(\bar s,\bar \by, \by,\widetilde s,\widetilde \by) = \bC^{ref}(\bar s,\bar \by,0,0).
\]
The proof of Proposition \ref{prop:cv_corr} is provided in Section \ref{sec:proof_prop_cv_corr}.

\subsection{Statistics of the incoherent wave fluctuations}
\label{sec:stat_speckle}

A precise analysis of the statistics of the reflected speckle profile is carried out first in the Fourier domain, at the scales of the pulse width and the correlation length, through
\[
\hat S^{\e,ref}(\bar s,\bar \by, \by, \omega, \bp) := \e^{1-2\gamma} \iint e^{i\omega(\widetilde s - \bp \cdot \widetilde \by)}  u^{\e,ref}\big( t^\e_{obs, ref}(\bar s,\bar \by,\by,\widetilde \by) +\e\widetilde s, \bx^\e_{obs,ref}(\bar \by, \by) + \e^\gamma \widetilde \by, z=0 \big) d\widetilde s d\widetilde \by,
\]
where $\bx^\e_{obs,ref}$ and $t^\e_{obs,ref}$ are given in \eqref{def:x_obs_eps} and \eqref{def:t_obs_eps}, respectively. As described by \eqref{def:Cor_func}, the position and time at which we observe incoherent wave fluctuations, away from the specular component, can be expressed in terms of the scattered direction $\bp$:
\begin{equation}\label{eq:y_s_speckle}
\by^{ref}_{\bp} := z_{int}c_0 A_0 \bp \qquad \text{and}\qquad s^{ref}_{\bp}:=z_{int} c_0 \bp^T A_0 \bp/2.
\end{equation}

Due to the singular nature of the correlation function \eqref{def:Cor_func}, which involves a Dirac mass at $\by^{ref}_{\bp}$ and $s^{ref}_{\bp}$, we do not study $\hat S^{\e,ref}$ directly but rather a tempered (mollified) version:
\begin{equation}\label{def:speckle_func}
\hat \CalS^{\e,ref}_{\by}(\bar s,\bar \by,\omega, \bp) := \hat S^{\e,ref}(\bar s, \bar \by, \by, \omega, \bp)\frac{1}{\e^{3(\gamma-1/2)}}\varphi^{1/2}\Big(2\frac{\bar s- s^{ref}_{\bp}}{\e^{2(\gamma-1/2)}},2\frac{\bar \by-\by^{ref}_{\bp}}{\e^{2(\gamma-1/2)}}\Big),
\end{equation}
where we multiply by the square root of a symmetric mollifier in order to smooth the correlation function around $s^{ref}_{\bp}$ and $\by^{ref}_{\bp}$. This singularity originates from the critical scaling considered in this paper, in which the wavelength and the amplitude of the interface fluctuations are of the same order. This leads to technical difficulties and hence to the need to introduce the smoothed field $\hat \CalS^{\e,ref}$.

This smoothing can be interpreted as a windowing of the signal. Writing
\ba\label{eq:cov}
\bar s = s^{ref}_\bp+\e^{2(\gamma-1/2)} \widetilde s \qquad\text{and}\qquad \bar \by=\by^{ref}_\bp+\e^{2(\gamma-1/2)} \widetilde \by,
\ea
where $\widetilde s$ and $\widetilde \by$ range over the window support
in (\ref{def:speckle_func}),
the corresponding observation position and time are
\[
\bx^\e_{obs,ref}(\bp) := \bx_{obs,ref} + \e^{1-\gamma}\by^{ref}_\bp + \se\by + \e^\gamma \widetilde \by,
\]
and
\[
t^\e_{obs,ref}(\bp) := t_{obs,ref} + \e^{1-\gamma} \bk_0 \cdot \by^{ref}_\bp + \se \bk_0\cdot \by + \e^{\gamma} \bk_0 \cdot \widetilde \by + \e^{2(1-\gamma)}s^{ref}_\bp + \e\widetilde s.
\]
As a result, the windowing is at the scale of the pulse width in time and at the scale of the correlation length in space.

Since the variable $\by$, associated with the beam-width scale $\se$, plays no particular role for the speckle correlation function, we first fix $\by$ for simplicity and discuss its role later.

In this section, we view the speckle pattern as a random field indexed by the scattered direction $\bp$, and we aim to describe its statistical behavior at the spatial scale corresponding to the correlation length of the interface fluctuations and at the temporal scale corresponding to the pulse width. At leading order, $\hat \CalS^{\e,ref}_{\by}$ reads
\begin{align}\label{eq:S_eps}
\hat \CalS^{\e,ref}_{\by}(\bar s,\bar \by, \omega, \bp) \simeq & \frac{\CalR}{2(2\pi)^{2} \e^{7(\gamma-1/2)+1} }\iint e^{-i\omega \bar s /\e^{2\gamma-1}} e^{i\omega \bp \cdot \bar \by /\e^{2\gamma-1}} e^{i\omega \bp \cdot \by /\e^{\gamma-1/2}} e^{i \omega ( \bq' - \bp/\e^{\gamma-1/2}) \cdot (\bx' - \bx_{int})/\se} \\
& \hspace{3cm} \times e^{2i\omega  \bcals_0 V(\bx'/\e^\gamma)} \hat \CalU_0(\omega, \bp/\e^{\gamma-1/2}, z_{int}) \hat \CalU_0(\omega, \bq', z_{int}) \hat \Psi(\omega, \bq')\nonumber\\
& \hspace{3cm} \times \omega^2 d\bx'  d\bq'\varphi^{1/2}\Big(2\frac{\bar s-s_{\bp}^{ref} }{\e^{2(\gamma-1/2)}},2\frac{\bar \by-\by_{\bp}^{ref}}{\e^{2(\gamma-1/2)}}\Big).\nonumber
\end{align}
Note that the expectation of $\hat \CalS^{\e,ref}_{\by}$ is, at leading order,
\begin{align*}
\E\big[\hat \CalS^{\e,ref}_{\by}(\bar s,\bar \by, \omega, \bp)\big] \simeq  \frac{\CalR}{2 \e^{7(\gamma-1/2)} } & e^{-i\omega \bar s  / \e^{2\gamma-1}} e^{i\omega \bp \cdot \bar \by /\e^{2\gamma-1}} e^{i \omega \bp \cdot \by /\e^{\gamma-1/2}}  \E\big[e^{2i\omega  \bcals_0 V(0)}\big] \\
&\times \hat \CalU_0(\omega, \bp/\e^{\gamma-1/2}, 2z_{int}) \hat \Psi(\omega, \bp/\e^{\gamma-1/2})\varphi^{1/2}\Big(2\frac{\bar s-s_{\bp}^{ref}}{\e^{2(\gamma-1/2)}},2\frac{\bar \by-\by_{\bp}^{ref}}{\e^{2(\gamma-1/2)}}\Big),
\end{align*}
so that, for any test function 
$\phi\in\CalS(\CalX\times\CalX, \mathbb{C})$,
we obtain, after the changes of variables \eqref{eq:cov} 
and also $\bp\to \e^{\gamma-1/2}\bq$,
\begin{align}\label{eq:mean_Se_ref}
\E\big[ \big< \hat \CalS^{\e,ref}_{\by}, \phi \big>_{\CalS',\CalS} \big] \simeq \frac{\e^{\gamma-1/2}\CalR}{2} \iiiint & e^{-i\omega s_\bq} e^{-i\omega \widetilde s}e^{i\omega \bq\cdot\by} \E\big[e^{2i\omega  \bcals_0 V(0)}\big] \\
&\times \hat \Psi(\omega, \bq)\varphi^{1/2}(2\widetilde s,2\widetilde\by) \overline{\phi(0,0,\omega,0)} \omega^2 d\widetilde s d\widetilde \by d\omega d\bq \nonumber .
\end{align}
In particular,
\[
\lim_{\e\to 0} \E\big[ \hat \CalS^{\e,ref}_{\by}(\bar s,\bar \by,\omega,\bp) \big] = 0
\]
in $\CalS'(\CalX \times \CalX, \mathbb{C})$, and the asymptotic random field is then mean zero. Here we emphasize the complex-valued nature of the Fourier transform of the speckle by specifying $\mathbb{C}$ in the space of tempered distributions $\CalS'(\CalX \times \CalX, \mathbb{C})$.

\begin{theorem}\label{th:speckle} For any fixed $ \by \in \R^2$, the family $(\hat \CalS^{\e,ref}_{\by})_\e$ converges in law in $\CalS'(\CalX \times\CalX,\mathbb{C})$ to a complex mean-zero  Gaussian random field $\hat \CalS^{ref}$, which does not depend on $\by$, and with a covariance functions given by
\begin{align}\label{eq:cov_1}
\E[\hat \CalS^{ref}(\phi_1)\hat \CalS^{ref}(\phi_2)]& = \int\dots\int \hat \CalK_{ref}(\bar s_1, \bar s_2,\bar \by_1,\bar \by_2,\omega_1,-\omega_2,d\bp_1,d\bp_2) \nonumber \\
& \hspace{3cm} \times \overline{\phi_1(\bar s_1,\bar \by_1,\omega_1,\bp_1)\phi_2(\bar s_2,\bar \by_2,\omega_2,\bp_2)}d\bar s_1 d\bar s_2 d\bar \by_1 d\bar \by_2 d\omega_1 d\omega_2 ,\\
\label{eq:cov_2}\E[\hat \CalS^{ref}(\phi_1)\overline{\hat \CalS^{ref}(\phi_2)}]& =  \int\dots\int  \hat \CalK_{ref}(\bar s_1,\bar s_2,\bar \by_1,\bar \by_2,\omega_1,\omega_2,d\bp_1,d\bp_2) \nonumber \\
& \hspace{3cm} \times \overline{\phi_1(\bar s_1,\bar \by_1,\omega_1,\bp_1)}\phi_2(\bar s_2,\bar \by_2,\omega_2,\bp_2)d\bar s_1 d\bar s_2 d\bar \by_1 d\bar \by_2 d\omega_1 d\omega_2 ,
\end{align}
for $\phi_1,\phi_2 \in \CalS(\CalX \times\CalX,\mathbb{C})$. The kernel $ \hat \CalK_{ref}$ is given by
\begin{align*}
\hat \CalK_{ref}(\bar s_1,\bar s_2,\bar \by_1,\bar \by_2,\omega_1,\omega_2,d\bp_1,d\bp_2) & = \frac{(2\pi)^3\CalR^2}{4}  \CalA(2\bcals_0, \omega_1,\bp_1) \vert\hat\Psi\vert^2_2(\omega_1)\hat \varphi(\omega_1,\bp_1)\delta(\bar s_1-s^{ref}_{\bp_1})\delta(\bar \by_1-\by^{ref}_{\bp_1}) \\
&\hspace{2cm}  \times \delta(\omega_1-\omega_2)\delta(\bp_1-\bp_2) \delta(\bar s_1-\bar s_2) \delta(\bar \by_1 - \bar \by_2)d\bp_1 d\bp_2.
\end{align*}
\end{theorem}
Here, due to the presence of Dirac masses involving  $\bar s_j$ and $\bar \by_j$, $j=1,2$, the speckle behaves as a Gaussian white-noise with respect to these variables; observed at two distinct points $\bar \by_1$ and $\bar \by_2$ (or at two distinct times $\bar s_1$ and $\bar s_2$) the speckles are independent. 
Indeed, if  $\phi_1$ and $\phi_2$ have disjoint support in either variable the integrals in  \eqref{eq:cov_1} and \eqref{eq:cov_2} will be zero. The proof of Theorem \ref{th:speckle} is provided in Section \ref{sec:proof_th_speckle}.

The statement of Theorem \ref{th:speckle} holds for a fixed $\by$. However, for two different values of $\by$, one can show that
\[
\lim_{\e\to 0} \E\big[\big< \hat \CalS^{\e,ref}_{\by_1}, \phi_1 \big>_{\CalS',\CalS}\big< \hat \CalS^{\e,ref}_{\by_2}, \phi_2 \big>_{\CalS',\CalS}\big] = 0,
\]
for any $\by_1,\by_2\in\R^2$ and any test functions $\phi_1,\phi_2\in\CalS(\CalX \times\CalX,\mathbb{C})$, meaning that the two random fields are asymptotically uncorrelated. This point is discussed at the end of Section \ref{sec:proof_prop_exp_cor}. More precisely, we have the following result.
\begin{corollary}\label{cor:vector_speckle}
For $n\geq 1$ and any fixed strictly different $ \by_1,\dots,\by_n, \in \R^2$, the family $(\hat \CalS^{\e,ref}_{\by_1},\dots,\hat \CalS^{\e,ref}_{\by_n})_\e$ converges in law in $\CalS'_n$ (which is $n$ times the Cartesian product of $\CalS'(\CalX\times \CalX,\mathbb{C})$) to a limit $(\hat \CalS^{ref}_1,\dots,\hat \CalS^{ref}_n)$ made of $n$ independent copies of the Gaussian random field $\hat \CalS^{ref}$ defined in Theorem \ref{th:speckle}.
\end{corollary}
In other words, for a given location $\bar \by$, the speckles observed at the two points
\[
\bx_{obs, ref} + \e^{1-\gamma}\bar \by +\se \by_1 \qquad\text{and}\qquad \bx_{obs, ref} + \e^{1-\gamma}\bar \by +\se \by_2,
\]
whose distance $|\by_1 - \by_2|$   is of order the beam width are independent.   

Now, the real-valued tempered random field
\[
\CalS^{\e,ref}_{\by}(\bar s,\bar \by, \widetilde s, \widetilde \by) := \frac{1}{(2\pi)^3}\iint e^{-i\omega(\widetilde s - \bp\cdot\widetilde \by)} \hat \CalS^{\e,ref}_{\by}(\bar s,\bar \by, \omega, \bp) \omega^2 d\omega d \bp,
\]
accounts for the time and space variations at the scales of the pulse width and the correlation length. For this random field we have the following result.

\begin{corollary}\label{corr:7.2}
For $n \geq 1$ and any fixed strictly different  $\by_1,\dots,\by_n \in \R^2$, the family $(\CalS^{\e,ref}_{\by_1},\dots,\CalS^{\e,ref}_{ \by_n})_\e$ converges in law in $\CalS'_n$ to $(\CalS^{ref}_1,\dots,\CalS^{ref}_n)$ made of $n$ independent copies of a real-valued mean-zero stationary (in $(\widetilde s, \widetilde \by)$) Gaussian random field $\CalS^{ref}$ with covariance function given by 
\begin{align}\label{eq:cov_3}
\E[\CalS^{ref}(\phi_1)\CalS^{ref}(\phi_2)]& = \int\dots\int \CalK_{ref}(\bar s_1,\bar s_2,\bar \by_1,\bar \by_2,\widetilde s_1- \widetilde s_2,\widetilde \by_1 - \widetilde \by_2)\nonumber  \\
& \hspace{3cm} \times \phi_1(\bar s_1,\bar \by_1,\widetilde s_1,\widetilde \by_1)\phi_2(\bar s_2,\bar \by_2,\widetilde s_2,\widetilde \by_2)d\bar s_1 d\bar s_2 d\bar \by_1 d\bar \by_2 d\widetilde s_1 d\widetilde s_2 d\widetilde \by_1 d\widetilde \by_2
\end{align}
for $\phi_1,\phi_2 \in \CalS(\CalX \times\CalX)$. The kernel $\CalK_{ref}$ is given by
\begin{align*}
\CalK_{ref}(\bar s_1,\bar s_2,\bar \by_1,\bar \by_2,\widetilde s,\widetilde \by) & =  \frac{\CalR^2}{4(2\pi)^3} \iint e^{-i\omega(\widetilde s - \bp\cdot\widetilde \by)} \CalA(2\bcals_0,\omega,\bp) \vert\hat\Psi\vert^2_2(\omega)\hat \varphi(\omega,\bp)\\
&\hspace{2cm}  \times \delta(\bar s_1-s^{ref}_{\bp})\delta(\bar \by_1-\by^{ref}_{\bp})  \delta(\bar s_1-\bar s_2)\delta(\bar \by_1-\bar \by_2)\omega^2 d\omega d\bp\\
& = \frac{\CalR^2 c_0^2 \bcals_0^4}{4(2\pi)^3 z^2_{int}}  \int e^{-i\omega(\widetilde s - \bar \by^T_1 A_0^{-1}\widetilde \by/(z_{int} c_0))} \vert\hat\Psi\vert^2_2(\omega) \CalA\Big(2\bcals_0,\omega,\frac{A_0^{-1} \bar \by_1}{z_{int}c_0}\Big) \hat \varphi\Big(\omega,\frac{A_0^{-1}\bar \by_1}{z_{int}c_0}\Big)\omega^2 d\omega \\
&\hspace{2cm}  \times \delta\Big(\bar s_1-\frac{\bar \by^T_1 A^{-1}_0\bar \by_1}{2z_{int}c_0}\Big)\delta(\bar s_1-\bar s_2)\delta(\bar \by_1-\bar \by_2).
\end{align*}
Here, $A^{-1}_0$ is given by \eqref{def:A0_inv}
\end{corollary}
The covariance kernel $\CalK_{ref}$ provides a direct characterization of the Gaussian speckle pattern as it emerges in the observation plane $z=0$ and evolves in time. In particular, for a given time $\bar s$, the speckles form ellipses satisfying $\bar s=\bar \by^T A^{-1}_0\bar \by/(2z_{int}c_0)$.
  
We end this section with a comment on the difference between the correlation function \eqref{def:corr_func}, which exhibits a self-averaging property, and the speckle field \eqref{def:speckle_func} in the Fourier domain, which remains random.
For the correlation function, the speckle profile is correlated around given values of $\bar s$ and $\bar \by$.
Hence, studying the second-order moment of the correlation function in $\CalS'$ corresponds to looking at the fourth-order moment of the speckle, but with two factors localized around $(\bar s_1,\bar \by_1)$ and two others around $(\bar s_2,\bar \by_2)$.
When these two observation points are far apart, the correlated speckles around each point involve rapidly oscillating phase terms that cancel when tested against a smooth test function, so that the two contributions become statistically independent:
\[
\E[\CalC^{\e,ref}(\bar s_1, \bar \by_1)\CalC^{\e,ref}(\bar s_2, \bar \by_2)]\simeq \E[\CalC^{\e,ref}(\bar s_1, \bar \by_1)]\E[\CalC^{\e,ref}(\bar s_2, \bar \by_2)],
\] 
yielding convergence in probability in $\CalS'$. 

In this section, we instead study the speckle itself as a random field indexed by the scattered direction $\bp$, with, for each $\bp$, the corresponding observation time $\bar s = s^{ref}_\bp$ and position $\bar \by=\by^{ref}_\bp$.
In this formulation, there is no statistical stabilization (self-averaging) at the level of the field itself, and Theorem \ref{th:speckle} describes the statistical coupling between the diffracted directions $\bp$ through a Gaussian random field.
A similar distinction between the wavefield and its correlation function has been observed in wave propagation problems with random bulk fluctuations, see, e.g., \cite{bal10, bal11, bal23, gu21, papanicolaou07}.

\section{The case $\gamma>1/2$: incoherent wave fluctuations for the nonspecular transmitted contributions}\label{sec:randint3}

\begin{figure}
\begin{center}
\includegraphics*[scale=0.25]{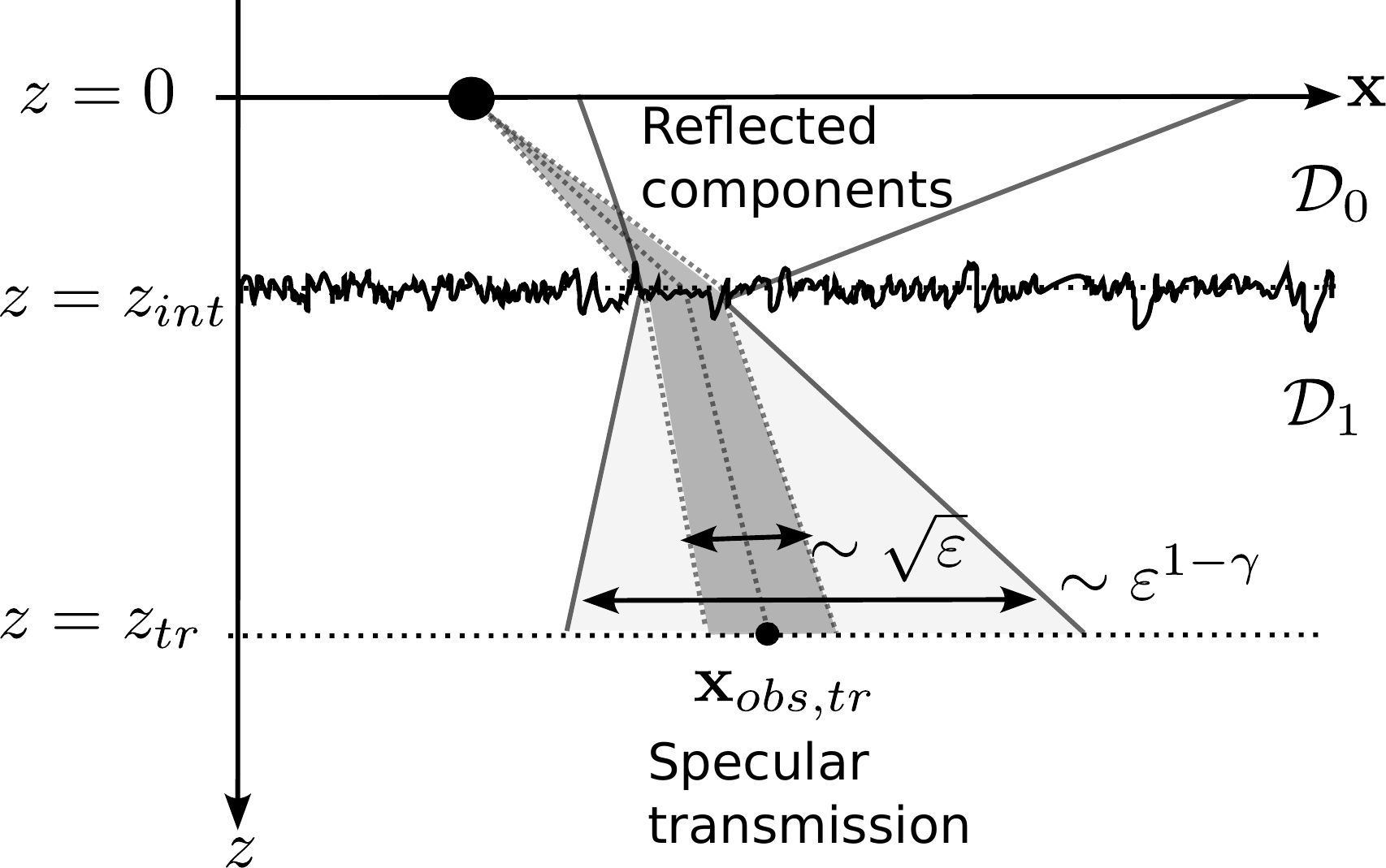}  
\end{center}
\caption{
Illustration of the zone where the transmitted speckle profile can be observed (in light gray), compared to the specular transmitted component (in gray).
}
\end{figure}

Following the same strategy as for the reflected component, the transmitted speckle profile is given by
\begin{equation*}
S^{\e,tr}(\bar s, \bar \by, \by, \widetilde s, \widetilde \by) := \e^{1 -2\gamma} u^{\e,tr}(t^\e_{obs,tr}(\bar s,\bar \by, \by, \widetilde\by) + \e \widetilde s, \bx^\e_{obs, tr}(\bar \by, \by) + \e^\gamma \widetilde\by , z=z_{tr}),
\end{equation*}
where $u^{\e,tr}$ is given by \eqref{eq:ue_tr} at the leading order,
\[
\bx^\e_{obs, tr}(\by,\bar \by) := \bx_{obs, tr} + \e^{1-\gamma} \bar \by + \se \, \by,
\]
and
\[
t^\e_{obs,tr}(\bar s,\bar \by,\by,\widetilde\by):= t_{obs,tr} + \e^{1-\gamma}\bk_0\cdot\bar \by + \se \, \bk_0\cdot \by + \e^\gamma \bk_0\cdot\widetilde \by + \e^{2(1-\gamma)}\bar s.
\]
As for the reflected wavefield, the transmitted wavefield $u^{\e,tr}$ is observed in a neighborhood of order then relative roughness parameter $\lambda/\ell_c=\e^{1-\gamma}$, which is larger than the beam width (since $\gamma>1/2$), around the location $\bx_{obs,tr}$ at which the specular transmitted component is observed. We also consider variations at the beam-width scale and fluctuations of order the correlation length through $\by$ and $\widetilde \by$, respectively (see the right panel of Figure \ref{fig:ref_wave_speckle} for an illustration). The speckle profile $S^{\e,tr}$ is observed on a time scale around the travel time $t_{obs,tr}$, with three corrections depending on the observation offsets ($\bar \by$, $\by$, and $\widetilde \by$). The overall observation window, of order $\e^{2(1-\gamma)}$ in time, is larger than the pulse width, and we also resolve fluctuations at the pulse-width scale through the variable $\widetilde s$. Finally, the prefactor $\e^{1-2\gamma}\gg 1$ characterizes the magnitude of the speckle component. Recall that as $\gamma$ increases (i.e., as the interface becomes rougher), the region in which the speckle pattern can be observed around $\bx_{obs,tr}$ increases, as does the speckle duration (captured by the term $\e^{2(1-\gamma)} \bar s$), while the speckle magnitude decreases.

As for the reflected speckle, the first moment of $S^{\e,tr}$ converges to zero in the high-frequency limit, and we focus on its two-point statistics and on the intensity.

\subsection{Correlation function and intensity.}

The correlation function of $S^{\e,tr}$ at two nearby points is defined by
\[
\CalC^{\e, tr}(\bar s,\bar \by,\by,\widetilde s_1,\widetilde \by_1, \widetilde s_2,\widetilde \by_2) : = S^{\e, tr}(\bar s,\bar \by,\by,\widetilde s_1,\widetilde \by_1)S^{\e, tr}(\bar s,\bar \by,\by,\widetilde s_2,\widetilde \by_2),
\]
and the associated intensity by
\[
\CalI^{\e, tr}(\bar s,\bar \by, \by,\widetilde s,\widetilde \by) : = \big\vert S^{\e, tr}(\bar s, \bar \by, \by,\widetilde s,\widetilde \by)\big\vert^2 = \CalC^{\e, tr}(\bar s, \bar \by, \by, \widetilde s,\widetilde \by, \widetilde s, \widetilde \by).
\]
The following result characterizes the asymptotic correlation function.

\begin{proposition}
For $\gamma \in (1/2,2]$ and $\bar \by \neq 0$, we have 
\begin{align*}
\lim_{\e\to 0}\CalC^{\e, tr}(\bar s, \bar \by, \by, \widetilde s_1, \widetilde \by_1, \widetilde s_2, \widetilde \by_2) = \bC^{tr}(\bar s, \bar \by, \widetilde s_1-\widetilde s_2, \widetilde \by_1-\widetilde \by_2),
\end{align*}
in probability in $\CalS'(\mathbb{X}\times \CalX\times \CalX)$, where
\begin{equation}
\label{def:Cor_func_tr}
\bC^{tr}(\bar s, \bar \by, \widetilde s, \widetilde \by) := \frac{\CalT^2 \bcals_0}{4\bcals_1 (2\pi)^3}\iint  e^{-i\omega (\widetilde s - \bp\cdot\widetilde \by)}   \CalA(\bcals_0-\bcals_1, \omega, \bp) |\hat \Psi|^2_2(\omega) \delta(\bar s - s^{tr}_\bp) \delta(\bar \by- \by^{tr}_\bp) \omega^2 d\omega  d\bp,
\end{equation}
with $\CalA$ defined by \eqref{def:CalA}, $|\hat \Psi|^2_2$ by \eqref{def:hPsi2}, and
\[s^{tr}_\bp= (z_{tr}-z_{int})c_1 \bp^T A_1 \bp/2 \geq 0,\qquad\text{and}\qquad \by^{tr}_\bp= (z_{tr}-z_{int})c_1 A_1 \bp.\]
\end{proposition}
The proof is omitted since it follows closely the proofs of Propositions \ref{prop:exp_corr} and \ref{prop:cv_corr}.

As for the reflected wavefield, the correlation function can be rewritten as 
\begin{align*}
\bC^{tr}(\bar s,\bar \by, \widetilde s,\widetilde \by) = \frac{\CalT^2 c_1^2 \bcals_0\bcals_1^3 }{4(2\pi)^3(z_{tr}-z_{int})^2}\delta\Big(\bar s - \frac{\bar \by^T  A_1^{-1}\bar \by}{2(z_{tr}-z_{int})c_1}\Big) \int & e^{-i\omega (\widetilde s - \bar \by ^T A_1^{-1} \widetilde \by/((z_{tr}-z_{int})c_1))}  |\hat \Psi|^2_2(\omega)  \\
&\times\CalA\Big(\bcals_0-\bcals_1,\omega,\frac{A_1^{-1}\bar \by}{(z_{tr}-z_{int})c_1}\Big)  \omega^2 d\omega,
\end{align*}
with 
\begin{equation}\label{def:A1_inv}
A_1^{-1} = c_1 \bcals_1(\mathbf{I}_2 - c_1^2 \bk_0 \otimes \bk_0).
\end{equation}
Thus, for a given location $\bar \by$, the observation time at which the speckle is observed is
\[
\bar s = \bar \by^T  A_1^{-1}\bar \by/(2(z_{tr}-z_{int})c_1)\geq 0,
\] 
and the temporal correlation length is of order the initial pulse duration $\e$, through the variable $\widetilde s$. Likewise, for a given time $\bar s$, the incoherent wave fluctuations are supported on an ellipse defined by $\bar s = \bar \by^T  A_1^{-1}\bar \by/(2(z_{tr}-z_{int})c_1)$. The correlation function decays on a temporal scale corresponding to the pulse duration and on a spatial scale corresponding to the correlation length of the interface fluctuations, through $\widetilde s$ and $\widetilde \by$, respectively. The stationarity property is again observed with respect to $(\widetilde s,\widetilde \by)$, and $\bC^{tr}$ does not depend on $\by$, which corresponds to variations at the beam-width scale. Similarly, the mean intensity carried by the speckle profile satisfies
\[
\lim_{\e\to 0}\CalI^{\e, tr}(\bar s,\bar \by,\by,\widetilde s,\widetilde \by) = \bC^{tr}(\bar s,\bar \by,0,0),
\] 
in probability. This limit does not depend on $\by$ nor on the small-scale variables $\widetilde s$ and $\widetilde \by$. The intensity is therefore uniform over these scales and varies only with $\bar s$ and $\bar \by$.

\paragraph{Generalized Snell's law of transmission.} As already mentioned, in \eqref{def:Cor_func_tr}, the spatial window of the speckle cone is larger than the beam width. For a given direction $\bp$, the speckle observed at position $\bar \by = \by^{tr}_\bp$ occurs at the corresponding time $\bar s= s^{tr}_\bp \geq 0$, which is nonnegative by \eqref{def:Aj} and the assumption $|\bk_0|<1/c_0<1/c_1$. Note that for $\gamma=1$, corresponding to a correlation length of order the wavelength, the speckle can be observed at distances of order one from $\bx_{obs,tr}$.

In terms of the transmission angle, for a given frequency $\omega$ and nonzero slowness vector $\bk_0$ (i.e.\ $\theta_{inc}>0$), we have, after some algebra, the following relation between the transmission angle and the specular transmission angle $\theta^0_{tr}$ at leading order in $\e$:
\begin{align*}
\frac{\tan(\theta^\e_{tr}(\bp))}{\tan(\theta^0_{tr})} & = \frac{|\bx_{obs,tr}-\bx_{int} + \e^{1-\gamma}\by^{tr}_\bp|}{|\bx_{obs,tr}-\bx_{int}|}\\
& = \sqrt{\Big(1+\xi_\e\frac{\bp\cdot\bk_0}{\cos^2(\theta^0_{tr})}\Big)^2 + \xi^2_\e(\bp\cdot \bk_0^\perp)^2},\qquad \xi_\e = \e^{1-\gamma}\frac{c_0^2}{\sin^2(\theta_{inc})},
\end{align*}
for $\bp$ distributed according to the scattering distribution $\CalA(\bcals_0-\bcals_1,\omega, \cdot)$. From this relation, the standard Snell's law for transmission can be generalized as follows:
\[
\frac{\sin(\theta^\e_{tr}(\bp))}{c_1} = \frac{\sin(\theta_{inc})}{c_0}+|\bk_0|\left(\sqrt{\frac{\Xi^\e_{tr}(\bp)}{1+\sin^2(\theta^0_{tr})\Big( \Xi^\e_{tr}(\bp)-1\Big)}}-1\right), \qquad \Xi^\e_{tr}(\bp) = \Big(1+\xi_\e\frac{\bp\cdot\bk_0}{\cos^2(\theta^0_{tr})}\Big)^2 + \xi^2_\e(\bp\cdot \bk_0^\perp)^2,
\] 
which is equivalent to \eqref{eq:tr_formula_intro}. For $\gamma \in (1/2,1)$, this yields the perturbative relation
\[
\frac{\sin(\theta^\e_{tr}(\bp))}{c_1} = \frac{\sin(\theta_{inc})}{c_0}+\e^{1-\gamma} \frac{\bp\cdot \widehat \bk_0}{\cos^2(\theta^0_{tr})} +\CalO(\e^{2(1-\gamma)}|\bp|^2),
\] 
with $\widehat \bk_0 = \bk_0 / |\bk_0|$. 
This provides the approximation
\[
\theta^\e_{tr}(\bp) = \theta^0_{tr} + \e^{1-\gamma}\frac{c_1\bp\cdot \widehat \bk_0}{\cos^3(\theta^0_{tr})} + \CalO(\e^{2(1-\gamma)}|\bp|^2),
\]
which yields \eqref{eq:correction2} recalling    \eqref{eq:rel_eps} .
 
Finally, for a null slowness vector $\bk_0=0$ (i.e.\ $\theta_{inc}=0$), the transmission angle satisfies $\tan(\theta^\e_{tr}(\bp)) = \e^{1-\gamma}c_1 |\bp|$, so that
\ba\label{eq:tant}
\theta^\e_{tr}(\bp) = \arctan(\e^{1-\gamma}c_1 |\bp|).
\ea

\begin{figure}
\begin{center}
\includegraphics*[scale=0.35]{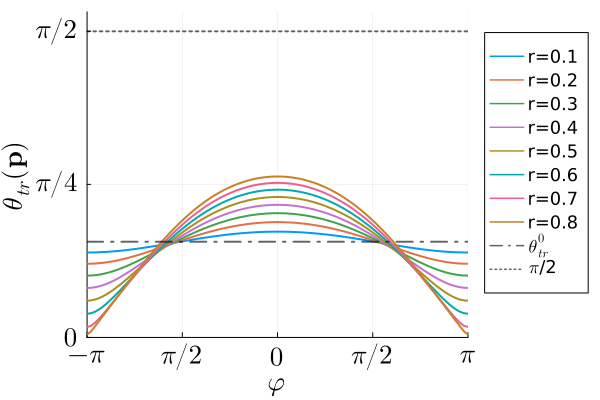} \hspace{1.5cm}
\includegraphics*[scale=0.35]{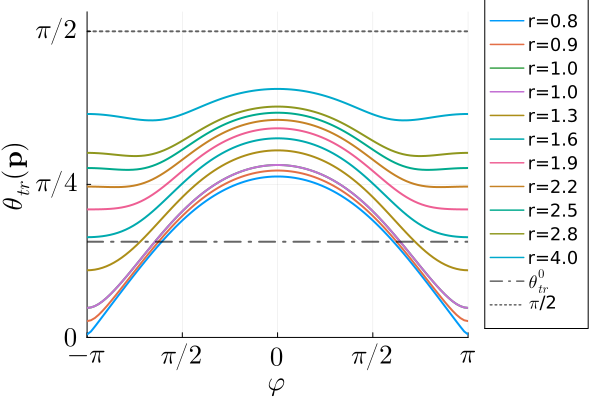} 
\end{center}
\caption{
Illustration of the specular transmission angle $\theta^0_{tr}$ (dash-dot line) and the transmission angle $\theta_{tr}(\bp)\in (0,\pi/2)$, for an incident angle $\theta_{inc}=\pi/4$, $\gamma=1$, $c_0=1.5$, $c_1=1$, and $\bp = \beta_1 \bk_0 + \beta_2 \bk_0^\perp$ with $(\beta_1,\beta_2)=r (\cos(\varphi),\sin(\varphi))$.
}
\end{figure}

\subsection{Statistics of the incoherent wave fluctuations}

The transmitted speckle in the Fourier domain is given by
\[
\hat S^{\e,tr}(\bar s,\bar \by, \by, \omega, \bp) := \e^{1-2\gamma} \iint e^{i\omega(\widetilde s - \bp \cdot \widetilde \by)}  u^{\e,tr}\big( t^\e_{obs, tr}(\bar s,\bar \by, \by,\widetilde \by) +\e\widetilde s, \bx^\e_{obs,tr}(\bar \by, \by) + \e^\gamma \widetilde \by \big) d\widetilde s d\widetilde \by.
\]
As described by \eqref{def:Cor_func_tr}, the position and time at which we observe the incoherent wave fluctuations away from the specular component can be expressed in terms of the scattered direction $\bp$:
\begin{equation}\label{eq:y_s_speckle_tr}
\by^{tr}_{\bp} := (z_{tr}-z_{int})c_1 A_1 \bp \qquad \text{and}\qquad s^{tr}_{\bp}:=\bp\cdot\by^{tr}_\bp/2 = (z_{tr}-z_{int}) c_1 \bp^T A_1 \bp/2.
\end{equation}
As for the reflected speckle profile, due to the singular nature of the correlation function, which involves Dirac masses at the observation time and position, we instead study a tempered version of $\hat S^{\e,tr}$:
\[
\hat \CalS^{\e,tr}_{\by}(\bar s,\bar \by,\omega, \bp) := \hat S^{\e,tr}(\bar s, \bar \by, \by, \omega, \bp)\frac{1}{\e^{3(\gamma-1/2)}}\varphi^{1/2}\Big(2\frac{\bar s- s^{tr}_{\bp}}{\e^{2(\gamma-1/2)}},2\frac{\bar \by-\by^{tr}_{\bp}}{\e^{2(\gamma-1/2)}}\Big),
\]
for a fixed $\by$, where we multiply by the square root of a symmetric mollifier in order to smooth the correlation function around $\bar s=s^{tr}_{\bp}$ and $\bar \by=\by^{tr}_{\bp}$. The results characterizing the transmitted speckle are as follows.

\begin{theorem}\label{th:speckle_tr} 
For $n \geq 1$ and any fixed strictly different $\by_1,\dots,\by_n \in \R^2$, the family  $(\hat \CalS^{\e,tr}_{\by_1},\dots, \hat \CalS^{\e,tr}_{\by_n})_\e$ converges in law in $\CalS'_n$ to $(\hat \CalS^{tr}_{\by_1},\dots, \hat \CalS^{tr}_{\by_n})$ made of $n$ independent copies of a complex mean-zero  Gaussian random field $\hat \CalS^{tr}$ with covariance functions similar to \eqref{eq:cov_1} and \eqref{eq:cov_2} but with kernel
\begin{align*}
\hat \CalK_{tr}(\bar s_1,\bar s_2,\bar \by_1,\bar \by_2,\omega_1,\omega_2,d\bp_1,d\bp_2) & = \frac{(2\pi)^3\CalT^2 \bcals_0}{4\bcals_1}  \CalA(\bcals_0-\bcals_1, \omega_1, \bp_1) \vert\hat\Psi\vert^2_2(\omega_1)\hat \varphi(\omega_1,\bp_1)\\
&\hspace{0.5cm}  \times \delta(\omega_1-\omega_2)\delta(\bp_1-\bp_2)   \delta(\bar s_1-\bar s_2)\delta(\bar \by_1-\bar \by_2)\delta(\bar s_1-s^{tr}_{\bp_1})  \delta(\bar \by_1-\by^{tr}_{\bp_1})d\bp_1 d\bp_2.
\end{align*} 
\end{theorem}
The proof of Theorem \ref{th:speckle_tr} is omitted since it follows the same steps as the proof of Theorem \ref{th:speckle} (see Section \ref{sec:proof_th_speckle}) together with Corollary \ref{cor:vector_speckle}. In particular, speckles observed at distances of order at least the beam width are independent. 

For the real-valued tempered random field
\[
\CalS^{\e,tr}_{\by}(\bar s,\bar \by, \widetilde s, \widetilde \by) := \frac{1}{(2\pi)^3}\iint e^{-i\omega(\widetilde s - \bp\cdot\widetilde \by)} \hat \CalS^{\e,tr}_{\by}(\bar s,\bar \by, \omega, \bp) \omega^2 d\omega d \bp,
\]
which represents the speckle signal around times $s^{tr}_{\bp}$ and positions $\by^{tr}_{\bp}$, we have the following result.
\begin{corollary}
For $n \geq 1$ and any fixed $\by_1,\dots,\by_n \in \R^2$, the family  $(\CalS^{\e,tr}_{\by_1},\dots, \CalS^{\e,tr}_{\by_n})_\e$ converges in law in $\CalS'_n$ to $( \CalS^{tr}_1,\dots, \CalS^{tr}_n)$ made of $n$ independent copies of a real valued mean-zero  Gaussian random field $\CalS^{tr}$ with a covariance function similar to \eqref{eq:cov_3}, but with kernel
\begin{align*}
\CalK_{tr}(\bar s_1, \bar s_2, \bar \by_1, \bar \by_2,\widetilde s,\widetilde \by) & =  \frac{\CalT^2 \bcals_0}{4\bcals_1(2\pi)^3} \iint e^{-i\omega(\widetilde s - \bp\cdot\widetilde \by)} \CalA(\bcals_0-\bcals_1, \omega, \bp) \vert\hat\Psi\vert^2_2(\omega)\hat \varphi(\omega,\bp)\\
&\hspace{2cm}  \times \delta(\bar s_1-s^{tr}_{\bp})  \delta(\bar \by_1-\by^{tr}_{\bp}) \delta(\bar s_1-\bar s_2)\delta(\bar \by_1-\bar \by_2)\omega^2 d\omega d\bp\\
&= \frac{\CalT^2 c_1^2 \bcals_0\bcals_1^3 }{4(2\pi)^3(z_{tr}-z_{int})^2} \int e^{-i\omega (\widetilde s - \bar \by ^T A_1^{-1} \widetilde \by/((z_{tr}-z_{int})c_1))}  |\hat \Psi|^2_2(\omega) \\
& \hspace{3cm}\times \CalA\Big(\bcals_0-\bcals_1,\omega,\frac{A_1^{-1} \bar \by}{(z_{tr}-z_{int})c_1}\Big) \hat \varphi\Big(\omega,\frac{A_1^{-1}\bar \by}{(z_{tr}-z_{int})c_1}\Big)  \omega^2 d\omega\\
& \hspace{3cm}\times \delta\Big(\bar s_1 - \frac{\bar \by^T_1  A_1^{-1}\bar \by_1}{2(z_{tr}-z_{int})c_1}\Big)\delta(\bar s_1-\bar s_2)\delta(\bar \by_1-\bar \by_2).
\end{align*}
Here, $A^{-1}_1$ given by \eqref{def:A1_inv}.
\end{corollary}
The covariance $\CalK_{tr}$ also provides a direct characterization of the Gaussian speckle as it reaches the observation plane $z=z_{tr}$ as a function of time, and it is supported on the ellipse
$\bar s=\bar \by^T A^{-1}_1\bar \by/(2(z_{tr}-z_{int})c_1)$.

\section{Proof of Propositions \ref{prop:exp_corr} and \ref{prop:cv_corr}}\label{sec:proof1}

This section is devoted to the proof of the convergence of the expected correlation function, as well as to its statistical stability. The proof is presented in a way that naturally justifies the scalings introduced in \eqref{def:Se_ref}--\eqref{def:t_obs_eps}.

\subsection{Proof of Proposition \ref{prop:exp_corr}}
\label{sec:proof_prop_exp_cor}

Recalling \eqref{eq:ue_ref}, the expected correlation function for the reflected wavefield reads, at the leading order,
\begin{align*}
C_\e & := \E[u^{\e,ref}(t_1,\bx_1, z=0)u^{\e,ref}(t_2,\bx_2, z=0)] \\
& \simeq \frac{\CalR^2}{4(2\pi)^{10} \e^2 }\int\dots\int e^{-i\omega_1(t_1 - 2\bcals_0z_{int}-\bx_1\cdot \bk_0 ) /\e} e^{i\omega_2(t_2 - 2\bcals_0z_{int}-\bx_2\cdot \bk_0 ) /\e} \\
& \hspace{3cm} \times \E\big[ e^{2i\omega_1 \bcals_0 V(\bx'_1/\e^\gamma)}e^{-2i\omega_2 \bcals_0 V(\bx'_2/\e^\gamma)} \big] \\
& \hspace{3cm} \times e^{i \omega_1 \bq_1 \cdot (\bx_1 - \bx'_1 - \bx_{int})/\se} e^{i \omega_1 \bq'_1 \cdot (\bx'_1 - \bx_{int})/\se}\\
& \hspace{3cm} \times e^{-i \omega_2 \bq_2 \cdot (\bx_2 - \bx'_2 - \bx_{int})/\se} e^{-i \omega_2 \bq'_2 \cdot (\bx'_2 - \bx_{int})/\se}\\
& \hspace{3cm} \times \hat \CalU_0(\omega_1,\bq_1, z_{int}) \hat \CalU_0(\omega_1,\bq'_1, z_{int})\overline{\hat \CalU_0(\omega_2,\bq_2, z_{int})\hat \CalU_0(\omega_2,\bq'_2, z_{int})} \\
& \hspace{3cm} \times \hat \Psi(\omega_1, \bq'_1) \overline{\hat \Psi(\omega_2, \bq'_2)} \omega^4_1 \omega^4_2 d\omega_1 d\omega_2 d\bx'_1 d\bx'_2 d\bq'_1 d\bq'_2 d\bq_1 d\bq_2.
\end{align*}
For notational convenience, we denote by $L_j$, $j=1,\dots,6$, the factors corresponding to each line of the above expression, so that
\ba\label{eq:Ls}
C_\e \simeq \frac{\CalR^2}{4(2\pi)^{10} \e^2 }\int\dots\int L_1 \times L_2 \times L_3 \times L_4 \times L_5 \times L_6 \,\omega^4_1 \omega^4_2 d\omega_1 d\omega_2 d\bx'_1 d\bx'_2 d\bq'_1 d\bq'_2 d\bq_1 d\bq_2.
\ea
To obtain a non-trivial limit for the factor $L_2$, we introduce the change of variables
\begin{equation*}
\bx'_1 \to \bx_{int} + \br' + \e^\gamma \by'/2 \qquad\text{and}\qquad \bx'_2 \to \bx_{int} + \br' - \e^\gamma \by'/2,
\end{equation*}
so that by stationarity 
\begin{equation}\label{eq:T2}
L_2 = \e^{2\gamma}\E\Big[ e^{2i\omega_1 \bcals_0 V((\bx_{int}+\br')/\e^\gamma +\by'/2)}e^{-2i\omega_2 \bcals_0 V((\bx_{int}+\br')/\e^\gamma -\by'/2)} \Big] = \e^{2\gamma}\E\Big[ e^{2i\omega_1 \bcals_0 V(\by'/2)}e^{-2i\omega_2 \bcals_0 V(-\by'/2)} \Big],
\end{equation}
where the term $\e^{2\gamma}$ comes from the changes of variables. With these two changes of variables, the product $L_3 \times L_4$ becomes
\begin{align*}
L_3 \times L_4 & = e^{i \br' \cdot (\omega_1 (\bq'_1-\bq_1)-\omega_2(\bq'_2-\bq_2))/\se} \\
& \hspace{0.5cm}\times  e^{i \e^{\gamma-1/2} \by' \cdot (\omega_1 (\bq'_1-\bq_1) + \omega_2(\bq'_2-\bq_2))/2} \\
&  \hspace{0.5cm} \times e^{i \omega_1 \bq_1 \cdot (\bx_1 - \bx_{obs,ref})/\se}  e^{-i \omega_2 \bq_2 \cdot (\bx_2 - \bx_{obs,ref})/\se}.
\end{align*}

Since $\br'$ appears only in the first exponential factor $e^{i \br' \cdot (\omega_1 (\bq'_1-\bq_1)-\omega_2(\bq'_2-\bq_2))/\se}$, integration with respect to $\br'$ yields the Dirac distribution
\[
(2\pi)^2 \e \delta(\omega_1 (\bq'_1-\bq_1)-\omega_2(\bq'_2-\bq_2)).
\]
In order to retain the $\by'$-dependence appearing in $L_2$, we further introduce the change of variables
\begin{equation*}
\bq_j \to \bq'_j - \bp_j/\e^{\gamma-1/2} \qquad j=1,2,
\end{equation*}
which yields
\begin{align*}
\int L_3 \times L_4 \, d\br' & = (2\pi)^2\e^{-2\gamma+2} \delta(\omega_1 \bp_1 -\omega_2 \bp_2 ) \, e^{i \by' \cdot (\omega_1 \bp_1 + \omega_2 \bp_2)/2}e^{- i (\omega_1 \bp_1 \cdot \bx_1 - \omega_2 \bp_2 \cdot \bx_2 )/\e^\gamma} \\
& \hspace{0.5cm} \times e^{i \omega_1 \bq'_1 \cdot (\bx_1 - \bx_{obs,ref})/\se} e^{-i \omega_2 \bq'_2 \cdot (\bx_2 - \bx_{obs,ref})/\se},
\end{align*}
using the presence of the Dirac mass to write
\[
e^{- i \omega_1 \bp_1 \cdot (\bx_1 - \bx_{obs,ref})/\e^\gamma}e^{ i \omega_2 \bp_2 \cdot (\bx_2 - \bx_{obs,ref})/\e^\gamma} = e^{- i \omega_1 \bp_1 \cdot \bx_1 /\e^\gamma}e^{ i \omega_2 \bp_2 \cdot \bx_2 /\e^\gamma}.  
\]
We then perform the change of variables $\omega_j \bp_j \mapsto \bp_j$, so that 
$\omega_1^4 \omega_2^4$ in \eqref{eq:Ls} is replaced by  $\omega_1^2 \omega_2^2$, and obtain
\begin{align*}
\int L_3 \times L_4 \, d\br' & = (2\pi)^2\e^{-2\gamma+2}  \delta(\bp_1 -\bp_2 ) \, e^{i \by' \cdot  \bp_1 }e^{- i \bp_1 \cdot (\bx_1 - \bx_2 )/\e^\gamma} \\
& \hspace{0.5cm} \times e^{i \omega_1 \bq'_1 \cdot (\bx_1 - \bx_{obs,ref})/\se} e^{-i \omega_2 \bq'_2 \cdot (\bx_2 - \bx_{obs,ref})/\se}.
\end{align*}
The rapid phase $e^{- i \bp_1 \cdot (\bx_1 - \bx_2 )/\e^\gamma}$ in the above equation and
the corresponding phase term involving $t_j$ in $L_1$ motivates the choice
\begin{equation}\label{eq:scaling1}
\bx_j = \bx + \e^\gamma \widetilde \by_j,\qquad\text{and}\qquad t_j = t + \e^\gamma\bk_0\cdot \widetilde \by_j + \e \widetilde s_j\qquad j=1,2,
\end{equation}
where we will specify  $\bx$  and $t$  below. With this choice, we obtain
\begin{align*}
\int L_3 \times L_4 \, d\br' & = (2\pi)^2 \e^{-2\gamma+2}  \delta(\bp_1 -\bp_2 ) \, e^{i \bp_1 \cdot   \by' }e^{- i \bp_1 \cdot (\widetilde \by_1 - \widetilde \by_2 )} \\
& \hspace{0.5cm} \times e^{i  (\bx - \bx_{obs,ref})\cdot (\omega_1 \bq'_1-\omega_2 \bq'_2 )/\se}  e^{i \e^{\gamma-1/2} ( \omega_1 \bq'_1 \cdot \widetilde \by_1 - \omega_2 \bq'_2 \cdot \widetilde \by_2)},
\end{align*}
while the factor $L_1$ reduces to
\[
L_1 = e^{-i(\omega_1-\omega_2)(t - 2\bcals_0z_{int} -\bx\cdot \bk_0 ) /\e}  e^{-i(\omega_1 \widetilde s_1 -\omega_2 \widetilde s_2)}.
\]
Regarding the term $L_5$, after all the changes of variables introduced above, we obtain
\begin{align*}
L_5 & = \hat \CalU_0\Big(\omega_1,\bq'_1-\frac{\bp_1}{\omega_1\e^{\gamma-1/2}}, z_{int} \Big) \hat \CalU_0(\omega_1,\bq'_1, z_{int})\\
&\hspace{0.5cm}\times \overline{\hat \CalU_0\Big(\omega_2,\bq'_2-\frac{\bp_1}{\omega_2\e^{\gamma-1/2}}, z_{int}\Big) \hat \CalU_0(\omega_2,\bq'_2, z_{int})}\\
& = e^{-i z_{int}c_0 \bp_1^T A_0 \bp_1 (1/\omega_1 - 1/\omega_2)/\e^{2\gamma-1} } e^{2iz_{int}c_0 (\bq'_1-\bq'_2)^T A_0 \bp_1/\e^{\gamma-1/2}} \, \hat \CalU_0(\omega_1,\bq'_1, 2z_{int})\overline{\hat \CalU_0(\omega_2,\bq'_2, 2z_{int})}.
\end{align*}
The presence of the two rapidly oscillating phase factors suggests introducing the changes of variables
\begin{equation}\label{eq:chg_var3}
\omega_1 = \omega + \e^{2\gamma-1}h/2,\qquad \omega_2 = \omega - \e^{2\gamma-1}h/2,
\end{equation}
and
\begin{equation*}
\bq'_1 = \bq + \e^{\gamma-1/2}\br/2,\qquad \bq'_2 = \bq - \e^{\gamma-1/2}\br/2.
\end{equation*}
At the leading order in $\e$, this yields
\begin{align} \label{def:term34}
L_5\times L_6 & \simeq \e^{4\gamma-2} e^{i h z_{int}c_0 \bp_1^T A_0 \bp_1 /\omega^2} e^{2iz_{int}c_0 \br^T A_0 \bp_1} |\hat \Psi(\omega,\bq)|^2 \omega^4,\nonumber\\
L_2 & \simeq \e^{2\gamma }\E\big[ e^{2i\omega \bcals_0 ( V(\by') - V(0)) } \big], \nonumber \\
\int L_3 \times L_4 \, d\br' & \simeq (2\pi)^2\e^{-2\gamma+2}  \delta(\bp_1 -\bp_2 ) \, e^{i  \bp_1 \cdot   \by'}e^{- i \bp_1 \cdot (\widetilde \by_1 - \widetilde \by_2 )}  e^{i \omega (\bx - \bx_{obs,ref})\cdot \br /\e^{1-\gamma}}  e^{i h \e^{2\gamma-3/2} (\bx - \bx_{obs,ref})\cdot \bq},\\
L_1 & \simeq  e^{-ih(t - 2\bcals_0z_{int} -\bx\cdot \bk_0 ) /\e^{2(1-\gamma)}}e^{-i\omega (\widetilde s_1 -\widetilde s_2)}.\nonumber
\end{align}
Here, we have used the stationarity of $V$ in \eqref{eq:T2} to simplify the expression of $L_2$. The last two relations naturally motivate the choice of observation variables
\begin{equation*}\label{eq:scaling2}
\bx = \bx_{obs,ref} + \e^{1-\gamma} \bar \by + \se\, \by \qquad\text{and}\qquad t = t_{obs,ref} + \e^{1-\gamma} \bk_0 \cdot \bar \by+ \se\, \bk_0 \cdot \by + \e^{2(1-\gamma)} \bar s.
\end{equation*}
Gathering all terms, and after the change of variable $\bp \mapsto -\omega \bp$, we obtain
\begin{align*}
C_\e & \simeq \frac{ \e^{4(\gamma-1/2)} \CalR^2  }{4(2\pi)^{8} }\int\dots\int  e^{-i\omega (\widetilde s_1 - \widetilde s_2 - \bp\cdot(\widetilde \by_1 - \widetilde \by_2))} e^{i\omega \br \cdot\bar \by}e^{-i h \bar s}  e^{- i\omega \bp \cdot \by'} \E\big[ e^{2i\omega \bcals_0 ( V(\by') - V(0)) } \big]\\
& \hspace{4cm} \times e^{-2 i \omega z_{int}c_0 \br^T A_0\bp} e^{i h z_{int}c_0 \bp^T A_0 \bp } |\hat \Psi(\omega, \bq )|^2 \,\omega^6 dh d\omega d \bq d\by'  d\br d\bp.
\end{align*} 
after the change of variable $\bp \to -\omega \bp$. As a consequence, the limit does not depend on the variable $\by$. Therefore, for any test functions $\varphi$, $\psi$, and $\phi$, we obtain
\begin{align*}
\lim_{\e\to 0}\iiint \E\big[ & \big \langle S^{\e,ref}(\bar s,\bar\by,\by),\varphi\big\rangle_{\CalS',\CalS} \big \langle S^{\e,ref}(\bar s,\bar\by,\by),\psi\big\rangle_{\CalS',\CalS} \big] \phi(\bar s,\bar\by,\by) d\bar s d\bar \by d\by \\
& = \int\dots\int \bC^{ref}(\bar s,\bar \by,\widetilde s_1-\widetilde s_2,\widetilde \by_1-\widetilde \by_2) \varphi(\widetilde s_1,\widetilde \by_1)\psi(\widetilde s_2,\widetilde \by_2)\phi(\bar s,\bar \by,\by) d\bar s d\bar \by d\by d\widetilde s_1 d\widetilde s_2 d\widetilde \by_1 d\widetilde \by_2,
\end{align*}
where $\bC^{ref}$ is defined by \eqref{def:Cor_func}. This concludes the proof of Proposition~\ref{prop:exp_corr}.

Let us finally remark that if we consider two different beam-scale variables $\by_1$ and $\by_2$ 
so that in \eqref{eq:scaling1} and \eqref{eq:scaling1} $\bx_1-\bx_2= \CalO(\se)$,    then  an additional oscillatory factor
\[
e^{i\bp_1\cdot (\by_1-\by_2)/\e^{\gamma-1/2}}
\]
appears in \eqref{def:term34}. In this  case, we obtain asymptotic decorrelation at the scale of the beam width:
\begin{align*}
\lim_{\e\to 0}\int\dots\int & \E\big[ S^{\e,ref}(\bar s,\bar\by,\by_1,\widetilde s_1,\widetilde \by_1) S^{\e,ref}(\bar s,\bar\by,\by_2,\widetilde s_2,\widetilde \by_2)\big] \\
& \times \varphi(\by_1,\widetilde s_1,\widetilde \by_1)\psi(\by_2,\widetilde s_2,\widetilde \by_2)\phi(\bar s,\bar \by,\by) d\bar s d\bar \by d\by_1 d\by_2 d\widetilde s_1 d\widetilde s_2 d\widetilde \by_1 d\widetilde \by_2 = 0.
\end{align*}
This absence of correlation at the scale of the beam width underlies the independence properties stated in Corollaries \ref{cor:vector_speckle} and \ref{corr:7.2}, as well as the corresponding results for the transmitted speckle.

\subsection{Proof of Proposition \ref{prop:cv_corr}}
\label{sec:proof_prop_cv_corr}

In this section we prove the statistical stabilization of the correlation function, namely
\[
\lim_{\e\to 0} \E\big[\CalC^{\e,ref}(\phi,\varphi,\psi)^2\big] = \CalC^{ref}(\phi,\varphi,\psi)^2,
\]
where
\[
\CalC^{ref}(\bar s,\bar \by,\by,\widetilde s_1,\widetilde \by_1,\widetilde s_2,\widetilde \by_2)= \bC^{ref}(\bar s,\bar \by,\widetilde s_1-\widetilde s_2,\widetilde \by_1-\widetilde \by_2).
\] 
As a consequence of the Markov inequality, for any $\eta>0$,
\[
\lim_{\e\to 0} \Pro\big( \big| \CalC^{\e,ref}(\phi,\varphi,\psi) - \CalC^{ref}(\phi,\varphi,\psi) \big| > \eta \big) = 0, 
\]
which yields convergence in probability. The proof follows closely the one of Proposition \ref{prop:exp_corr}, so we only highlight the key additional ingredients. 

At leading order in $\e$, the reflected speckle profile can be written as
\begin{align*}
S^{\e,ref}(\bar s,\bar \by, \by,\widetilde s, \widetilde \by) \simeq & \frac{\CalR}{2(2\pi)^{5} \e^{2\gamma} }\iiiint e^{-i\omega \bar s /\e^{2\gamma-1}} e^{-i\omega \widetilde s } e^{i \omega \bq \cdot \bar \by /\e^{\gamma-1/2}} e^{i \omega \bq \cdot \by } e^{i \e^{\gamma-1/2} \omega \bq \cdot  \widetilde \by } \\
& \hspace{3cm} \times  e^{i \omega (\bq' - \bq) \cdot (\bx' - \bx_{int})/\se}e^{2i\omega  \bcals_0 V(\bx'/\e^\gamma)} \nonumber\\
& \hspace{3cm} \times \hat \CalU_0(\omega, \bq, z_{int}) \hat \CalU_0(\omega, \bq', z_{int}) \hat \Psi(\omega, \bq')  \omega^4 d\omega  d\bx'  d\bq'  d\bq.\nonumber
\end{align*}
Therefore, for test functions $\phi\in\CalS(\mathbb{X})$ and $\varphi,\psi\in\CalS(\CalX)$,
\begin{align*}
\CalC^{\e,ref}(\phi,\varphi, \psi) := & \frac{\CalR^2}{4(2\pi)^{10} \e^{4\gamma} }\int\cdots\int e^{-i(\omega_1-\omega_2) \bar s /\e^{2\gamma-1}} e^{-i(\omega_1  \widetilde s_1-\omega_2 \widetilde s_2) } \\
& \hspace{3cm} \times  e^{i (\omega_1 \bq_1-\omega_2 \bq_2) \cdot \bar \by /\e^{\gamma-1/2}}  e^{i (\omega_1 \bq_1-\omega_2 \bq_2) \cdot \by} e^{i \e^{\gamma-1/2} (\omega_1 \bq_1 \cdot  \widetilde \by_1 -  \omega_2 \bq_2 \cdot \widetilde \by_2) } \\
& \hspace{3cm} \times  e^{i \omega_1 (\bq'_1 - \bq_1) \cdot (\bx'_1 - \bx_{int})/\se} e^{-i \omega_2 (\bq'_2 - \bq_2) \cdot (\bx'_2 - \bx_{int})/\se}\\
& \hspace{3cm} \times  e^{2i  \bcals_0 (\omega_1 V(\bx'_1/\e^\gamma)-\omega_2 V(\bx'_2/\e^\gamma))} \\
& \hspace{3cm} \times \hat \CalU_0(\omega_1, \bq_1, z_{int}) \hat\CalU_0(\omega_1, \bq'_1, z_{int})  \overline{\hat \CalU_0(\omega_2, \bq_2, z_{int}) \hat \CalU_0(\omega_2, \bq'_2, z_{int})  } \\
& \hspace{3cm} \times \hat \Psi(\omega_1, \bq'_1)\overline{\hat \Psi(\omega_2, \bq'_2)} \phi(\bar s,\bar \by,\by)\varphi(\widetilde s_1,\widetilde \by_1) \psi(\widetilde s_2, \widetilde \by_2)\\
& \hspace{3cm} \times \omega^4_1 \omega^4_2 d\omega_1 d\omega_2 d\bar s d\widetilde s_1 d\widetilde s_2  d\bx'_1  d\bq'_1  d\bq_1 d\bx'_2  d\bq'_2  d\bq_2 d\bar \by d\by d\widetilde \by_1 d\widetilde \by_2.
\end{align*}
Taking the expectation of $\CalC^{\e,ref}(\phi,\varphi,\psi)^2$ yields 
\[
E_\e := \E\Big[ e^{2i  \bcals_0 \sum_{j=1}^2 \omega_{1,j} V(\bx'_{1,j}/\e^\gamma)-\omega_{2,j} V(\bx'_{2,j}/\e^\gamma)}\Big]. 
\]
Introducing the changes of variables, for $j=1,2$,
\[
\bx'_{1,j} \to \bx_{int} + \br'_j + \e^\gamma \by'_j/2 \qquad\text{and}\qquad \bx'_{2,j} \to \bx_{int} + \br_j' - \e^\gamma \by'_j/2,\qquad j=1,2,
\]
and using stationarity of $V$, we can rewrite $E_\e$ as
\begin{align*}
E_\e & =  \E\Big[ e^{2i  \bcals_0 \sum_{j=1}^2 \omega_{1,j} V((\bx_{int} + \br'_j)/\e^\gamma+\by'_j/2)-\omega_{2,j} V((\bx_{int} + \br'_j)/\e^\gamma-\by'_j/2)}\Big]\\
& = \E\Big[ e^{2i  \bcals_0 ( \omega_{1,1} V( \br'_1/\e^\gamma+\by'_1/2)-\omega_{2,1} V(\br'_1/\e^\gamma-\by'_1/2))}e^{2i  \bcals_0 ( \omega_{1,2} V( \br'_2/\e^\gamma+\by'_2/2)-\omega_{2,2} V(\br'_2/\e^\gamma-\by'_2/2))}\Big].
\end{align*}
Since $\gamma>1/2$, the points $\br'_1/\e^\gamma$ and $\br'_2/\e^\gamma$ become asymptotically far apart (for $\br'_1\neq \br'_2$), and by Lemma \ref{lem:mixing} we obtain the factorization
\[
E_\e\underset{\e\to 0}{\simeq} \E\Big[ e^{2i  \bcals_0 ( \omega_{1,1} V( \br'_1/\e^\gamma+\by'_1/2)-\omega_{2,1} V(\br'_1/\e^\gamma-\by'_1/2))}\Big]\E\Big[e^{2i  \bcals_0 ( \omega_{1,2} V( \br'_2/\e^\gamma+\by'_2/2)-\omega_{2,2} V(\br'_2/\e^\gamma-\by'_2/2))}\Big].
\]
At this stage, one can repeat for each factor the same sequence of changes of variables and oscillatory-integral arguments as in the proof of Proposition~\ref{prop:exp_corr}. This yields, for each $j=1,2$, the convergence of the corresponding contribution toward $\CalC^{ref}(\phi,\varphi,\psi)$, and therefore
\[
\lim_{\e\to 0} \E\big[\CalC^{\e,ref}(\phi,\varphi,\psi)^2\big] = \CalC^{ref}(\phi,\varphi,\psi)^2.
\]
This concludes the proof of Proposition \ref{prop:cv_corr}.

\section{Proof of Theorem \ref{th:speckle}}\label{sec:proof_th_speckle}\label{sec:proof2}

The proof follows the strategy developed in \cite[Section~9.3.4]{fouque}. We characterize the limiting law of the speckle by studying the convergence of all finite-dimensional moments. Tightness then follows from uniform bounds on the second-order moment.

Let $\phi \in \CalS(\CalX \times \CalX,\mathbb{C})$ be a test function. Applying $\phi$ to \eqref{eq:S_eps} yields

\begin{align*}
 J^{\e}(\phi) & := \big< \hat \CalS^{\e,ref}_{\by}, \phi \big>_{\CalS',\CalS} \\
& \simeq \frac{\CalR}{2(2\pi)^{2} \e^{7(\gamma-1/2)+1} }\int\dots\int e^{-i \omega \bar s/\e^{2\gamma-1}} e^{i \omega \bp \cdot \bar \by /\e^{2\gamma-1}} e^{i \omega \bp \cdot \by /\e^{\gamma-1/2}}e^{i \omega( \bq' - \bp/\e^{\gamma-1/2}) \cdot (\bx' - \bx_{int})/\se} \\
& \hspace{3cm} \times e^{2i\omega  \bcals_0 V(\bx'/\e^\gamma)}\hat \CalU_0(\omega, \bp/\e^{\gamma-1/2}, z_{int}) \hat \CalU_0(\omega, \bq', z_{int}) \hat \Psi(\omega, \bq')  \\
& \hspace{3cm}  \times \varphi^{1/2}\Big(2\frac{\bar s-s^{ref}_{\bp}}{\e^{2(\gamma-1/2)}},2\frac{\bar \by-\by^{ref}_{p}}{\e^{2(\gamma-1/2)}}\Big) \overline{\phi(\bar s,\bar \by,\omega,\bp)} \omega^2 d\omega d\bar s d\bx'  d\bq'  d\bar \by   d\bp.
\end{align*}
Using the stationarity of the random interface $V$, we may assume without loss of generality that
\[
\bx_{int}=0.
\]
The proof relies on computing the limits of all moments of $J^{\e}(\phi)$. The key mechanism is the pairing of the spatial integration variables $\bx'$ so as to cancel the rapidly oscillating phases, combined with the mixing property stated in Lemma~\ref{lem:mixing}. To organize these pairings, we order the variables according to the associated frequencies $\omega$.

We decompose $J^{\e}(\phi)$ into its positive and negative frequency contributions by writing
\[
J^{\e}_{\pm} := \int_0^\infty \CalJ^{\e}(\pm \omega) d\omega,
\]
with
\begin{align*}
\CalJ^{\e}(\omega)  := & \frac{\CalR}{2(2\pi)^{2} \e^{7(\gamma-1/2)+1}}\int\cdots\int  e^{-i\omega \bar s / \e^{2\gamma-1}} e^{i \omega \bp \cdot \bar \by /\e^{2\gamma-1}} e^{i \omega \bp \cdot \by /\e^{\gamma-1/2}}e^{ i\omega (  \bq' - \bp/\e^{\gamma-1/2}) \cdot \bx'/\se} e^{2i\omega  \bcals_0 V(\bx'/\e^\gamma)} \nonumber\\
& \hspace{4cm} \times \hat \CalU_0(\omega, \bp/\e^{\gamma-1/2}, z_{int}) \hat \CalU_0(\omega, \bq', z_{int}) \hat \Psi(\omega, \bq') \varphi^{1/2}\Big(2\frac{\bar s-s^{ref}_{\bp}}{\e^{2(\gamma-1/2)}},2\frac{\bar \by-\by^{ref}_{p}}{\e^{2(\gamma-1/2)}}\Big)  \nonumber\\
& \hspace{4cm}  \times   \overline{\phi(\bar s,\bar \by,\omega,\bp)} \omega^2 d\bar s d\bx' d\bq' d\bp d\bar \by.
\end{align*}

\paragraph{Even moments.}

We begin by considering the even moments of $J^{\e}(\phi)$. Decomposing the integral with respect to $\omega$ into positive and negative frequencies yields
\[
\E\big[J^{\e}(\phi)^{2n}\big] = \sum_{l=0}^{2n} \binom{2n}{l} M_\e(l,2n-l),
\]
where
\begin{equation}\label{def:Me}
M_{\e}(l,2n-l) := \E\big[ (J^{\e}_{+})^l (J^{\e}_{-})^{2n-l}\big].
\end{equation}
Here we have used the change of variable $\omega\mapsto -\omega$ to rewrite the contribution of negative frequencies in terms of $\CalJ^{\e}(-\omega)$.

The analysis of the asymptotic behavior of $M_{\e}(l,2n-l)$ relies on pairing the spatial variables associated with equal frequencies and exploiting the decorrelation induced by the mixing property of the random interface. As will be shown below, only fully frequency paired configurations contribute in the limit $\e\to 0$, leading to Gaussian statistics.

\paragraph{The case $l=n$.} We now analyze the moment $M_\e(n,n)$, which will provide the only nontrivial contribution as $\e\to 0$.
By definition, we have
\[
M_{\e}(n,n)  =  \int_0^\infty \cdots \int_0^\infty  \E\Big[\prod_{j=1}^n\CalJ^{\e}(\omega_{1,j}) \CalJ^{\e}(-\omega_{2,j})\Big]d \omega_{1,j}d \omega_{2,j}.
\]
Using symmetry with respect to the $\omega$-variables, we may restrict to ordered frequencies:
\[
M_{\e}(n,n)  = n!^2 \int_{\{0<\omega_{1,1}<\dots<\omega_{1,n}\}}\int_{\{0<\omega_{2,1}<\dots<\omega_{2,n}\}}  \E\Big[\prod_{j=1}^n \CalJ^{\e}(\omega_{1,j}) \CalJ^{\e}(-\omega_{2,j})\Big]d\omega_{1,j}d\omega_{2,j}.
\]
We then pair the positive and negative frequencies through the change of variables
\[
\omega_{1,j} \to \omega_{j} + \e^{2\gamma-1} h_j/2 \qquad\text{and}\qquad  \omega_{2,j} \to \omega_{j} - \e^{2\gamma-1} h_j/2,
\]
and introduce $H_n^\e$ the image of the ordering constraints in the $h$-variables:
\[
H^\e_n = \Big\{ (h_1,\dots,h_n):\quad h_j \in \Big(2\frac{\omega_{j} - \omega_{j+1}}{\e^{2\gamma-1}} + h_{j+1}, 2\frac{\omega_{j} - \omega_{j-1}}{\e^{2\gamma-1}} + h_{j-1} \Big),\qquad j\in \{1,\dots,n\} \Big\},
\]
where the convention $\omega_{n+1}=\infty$, $h_{n+1}=\omega_{0}=h_0=0$ is used. In particular, the domain $H^\e_n$ becomes $\R^n$ in the limit $\e\to 0$. Therefore, we obtain
\begin{align*}
M_{\e}(n,n) & = \e^{n(2\gamma-1)} n!^2 \int_{\{0<\omega_{1}<\dots<\omega_{n}\}}\int_{H^\e_n}  \E\Big[\prod_{j=1}^n \CalJ^{\e}(\omega_{j}+\e^{2\gamma-1}h_j/2) \CalJ^{\e}(-\omega_{j}+\e^{2\gamma-1}h_j/2)\Big]d\omega_{j}dh_{j} \nonumber\\
& = \frac{n!^2 \CalR^{2n}}{4^n(2\pi)^{4n} \e^{n(12(\gamma-1/2)+2)}}\nonumber\\
&\hspace{2cm}\times \int_{\{0<\omega_{1}<\dots<\omega_{n}\}}\int_{H^\e_n} \int\cdots\int \E\Big[\prod_{j=1}^n M_\e^{j} \Big]\, \prod_{j=1}^n \prod_{l=1,2} d\omega_{j} dh_j d\bar s_{l,j} d\bx'_{l,j} d\bq'_{l,j}  d\bp_{l,j} d\bar \by_{l,j},
\end{align*}
where
\begin{align*}
M_\e^{j} & :=  e^{-i\omega_j(\bar s_{1,j}-\bar s_{2,j})/\e^{2\gamma-1}} e^{-ih_j (\bar s_{1,j}+\bar s_{2,j})/2} \\
&\hspace{0.5cm}\times  e^{i ((\omega_j+\e^{2\gamma-1}h_j/2) \bp_{1,j} \cdot (\bar \by_{1,j} + \e^{\gamma-1/2} \by ) -  (\omega_j-\e^{2\gamma-1}h_j/2) \bp_{2,j}\cdot (\bar \by_{2,j} + \e^{\gamma-1/2} \by ) /\e^{2\gamma-1}}   \\
&\hspace{0.5cm}\times  e^{i ((\omega_j+\e^{2\gamma-1}h_j/2) \bq'_{1,j} \cdot (\bar \by_{1,j} + \e^{\gamma-1/2} \by ) -  (\omega_j-\e^{2\gamma-1}h_j/2) \bq'_{2,j}\cdot (\bar \by_{2,j} + \e^{\gamma-1/2} \by ) /\e^{\gamma-1/2}}  \\
&\hspace{0.5cm}\times e^{ - i  ((\omega_j+\e^{2\gamma-1}h_j/2) \bp_{1,j}\cdot \bx'_{1,j} -  (\omega_j-\e^{2\gamma-1}h_j/2) \bp_{2,j}\cdot \bx'_{2,j})  /\e^{\gamma}}     \\
&\hspace{0.5cm}\times e^{ 2i\bcals_0((\omega_j+\e^{2\gamma-1}h_j/2)V(\bx'_{1,j}/\e^\gamma) - (\omega_{j}-\e^{2\gamma-1}h_j/2)V(\bx'_{2,j}/\e^\gamma))}\\
&\hspace{0.5cm}\times  \hat \CalU_0(\omega_j+\e^{2\gamma-1}h_j/2,\bp_{1,j}/\e^{\gamma-1/2}+\bq'_{1,j},z_{int}) \overline{\hat \CalU_0(\omega_{j}-\e^{2\gamma-1}h_j/2,\bp_{2,j}/\e^{\gamma-1/2}+\bq'_{2,j},z_{int})} \\
& \hspace{0.5cm}\times \hat \CalU_0(\omega_j+\e^{2\gamma-1}h_j/2,\bq'_{1,j},z_{int}) \overline{\hat \CalU_0(\omega_{j}-\e^{2\gamma-1}h_j/2,\bq'_{2,j},z_{int})}  \\
& \hspace{0.5cm}\times \hat \Psi(\omega_j+\e^{2\gamma-1}h_j/2, \bq'_{1,j}) \overline{ \hat \Psi( \omega_j-\e^{2\gamma-1}h_j/2, \bq'_{2,j})} \\
& \hspace{0.5cm}\times \varphi^{1/2}\Big(2\frac{\bar s_{1,j}-s^{ref}_{\bp_{1,j}} }{\e^{2(\gamma-1/2)}},2\frac{\bar \by_{1,j}-\by^{ref}_{\bp_{1,j}}}{\e^{2(\gamma-1/2)}}\Big) \varphi^{1/2}\Big(2\frac{\bar s_{2,j}-s^{ref}_{\bp_{2,j}} }{\e^{2(\gamma-1/2)}},2\frac{\bar \by_{2,j}-\by^{ref}_{\bp_{2,j}}}{\e^{2(\gamma-1/2)}}\Big) \\
& \hspace{0.5cm}\times \overline{\phi(\bar s_{1,j}, \bar \by_{1,j}, \by,\omega_j+\e^{2\gamma-1}h_j/2, \bp_{1,j}+\e^{\gamma-1/2}\bq'_{1,j})}\\
&\hspace{0.5cm}\times \overline{\phi(\bar s_{2,j},\bar \by_{2,j},\by,-\omega_{j}+\e^{2\gamma-1} h_j/2, \bp_{2,j}+\e^{\gamma-1/2}\bq'_{2,j})}\\
& \hspace{0.5cm}\times (\omega_j+\e^{2\gamma-1}h_j/2)^2(\omega_j-\e^{2\gamma-1}h_j/2)^2 
\end{align*}
after the change of variables
\[
\bp_{1,j} \to \e^{\gamma-1/2}\bq'_{1,j} + \bp_{1,j} \qquad\text{and}\qquad \bp_{2,j} \to \e^{\gamma-1/2}\bq'_{2,j} + \bp_{2,j}.
\]

Next, we perform the change of variables in space:
\[
\bx'_{1,j}\to \br'_j + \e^\gamma\by'_j/2 \qquad\text{and}\qquad \bx'_{2,j}\to \br'_j - \e^\gamma\by'_j/2,
\]
Taking expectation in the random phase term in $M_\e^{j}$ yields
\begin{align*}
\E\Big[\prod_{j=1}^n  e^{ 2i\bcals_0  ( (\omega_{j}+\e^{2\gamma-1} h_j/2)V(\br_{j}/\e^\gamma + \by'_j/2)  -(\omega_{j}-\e^{2\gamma-1} h_j/2)  V(\br_{j}/\e^\gamma - \by'_j/2))}\Big] & \simeq  \prod_{j=1}^n\E\big[ e^{ 2i\omega_j \bcals_0  (V(\br_{j}/\e^\gamma+\by'_j/2)  -  V(\br_{j}/\e^\gamma - \by'_j/2))}\big]\\
& \simeq  \prod_{j=1}^n\E\big[ e^{ 2i\omega_j\bcals_0 ( V(\by'_j/2) - V( - \by'_j/2))}\big]\\
& \simeq  \prod_{j=1}^n\E\big[ e^{ 2i\omega_j\bcals_0 (V(\by'_j) -  V(0))}\big]
\end{align*} 
as $\e\to 0$. Here we used Lemma\ref{lem:mixing} for the factorization, and stationarity of $V$ in the last two steps. Since the resulting expression no longer depends on $\br'_j$, the integration over the $\br'_j$ variables produces Dirac masses. Indeed, the only remaining dependence in $\br'_j$ appears through
\[
e^{i((\omega_j+\e^{2\gamma-1}h_j/2)\bp_{1,j}-(\omega_j-\e^{2\gamma-1}h_j/2)\bp_{2,j})\cdot\br'_j/e^\gamma}
\]
hence integration in $\br'_j$ yields 
\[
(2\pi)^{2n}\e^{2n\gamma} \prod_{j=1}^n \delta((\omega_j+\e^{2\gamma-1}h_j/2)\bp_{1,j} - (\omega_j-\e^{2\gamma-1}h_j/2)\bp_{2,j}).
\]
We then use these constraints to reduce to a single scattered-direction variable $\bp_j$ for each $j$. At leading order in $\e$, this gives
\begin{align*}
\E\Big[\prod_{j=1}^n M_\e^{j} \Big] & \simeq \prod_{j=1}^n e^{-i\omega_j(\bar s_{1,j}-\bar s_{2,j})/\e^{2\gamma-1}} e^{-ih_j((\bar  s_{1,j}+ \bar s_{2,j})/2 - s^{ref}_{\bp_j})}  e^{ih_j \bp_j\cdot (\bar \by_{1,j}-\bar \by_{2,j})/2} e^{i \omega_j\bp_{j} \cdot (  \bar \by_{1,j}   - \bar \by_{2,j}  ) /\e^{2\gamma-1}}\\
&\hspace{0.5cm}\times  e^{i \omega_j ( \bq'_{1,j} \cdot (\bar \by_{1,j}- \by^{ref}_{\bp_j})  -   \bq'_{2,j}\cdot (\bar \by_{2,j}-\by^{ref}_{\bp_j}) ) /\e^{\gamma-1/2}} e^{i \omega_j( \bq'_{1,j}   -  \bq'_{2,j})\cdot \by}  \\
&\hspace{0.5cm}\times \E\big[ e^{ 2i\omega_j\bcals_0 (V(\by'_j) -  V(0))}\big]e^{- i \omega_j \by'_j \cdot \bp_{j} }\\
& \hspace{0.5cm}\times \hat \CalU_0(\omega_j,\bq'_{1,j},2z_{int}) \overline{\hat \CalU_0(\omega_{j},\bq'_{2,j},2z_{int})}   \hat \Psi(\omega_j, \bq'_{1,j}) \overline{ \hat \Psi( \omega_j, \bq'_{2,j})} \\
& \hspace{0.5cm}\times \varphi^{1/2}\Big(2\frac{\bar s_{1,j}-s^{ref}_{\bp_{j}}}{\e^{2\gamma-1}},2\frac{\bar \by_{1,j}-\by^{ref}_{\bp_{j}}}{\e^{2(\gamma-1/2)}}\Big) \varphi^{1/2}\Big(2\frac{\bar s_{2,j}-s^{ref}_{\bp_{j}}}{\e^{2\gamma-1}},2\frac{\bar \by_{2,j}-\by^{ref}_{\bp_{j}}}{\e^{2(\gamma-1/2)}}\Big) \\
& \hspace{0.5cm}\times \overline{\phi(\bar s_{1,j}, \bar \by_{1,j},\omega_j, \bp_{j})}\overline{\phi(\bar s_{2,j},\bar \by_{2,j},-\omega_{j}, \bp_{j})} \omega_j^2.
\end{align*}
Now, making the changes of variables
\[
\bar s_{1,j} \to s^{ref}_{\bp_j} + \bar s_j + \e^{2\gamma-1} \widetilde s_j/2, \qquad\text{and}\qquad \bar s_{2,j} \to s^{ref}_{\bp_j}  +  \bar s_j - \e^{2\gamma-1} \widetilde s_j/2,
\]
the resulting term  $e^{-ih_j(\bar s_j - s^{ref}_{\bp_j})}$ yields upon integration in $h_j$,
\[
2\pi \delta(\bar s_j).
\]
Similarly, set
\[
\bar \by_{1,j} \to \by^{ref}_{\bp_j} + \bar \by_j + \e^{2(\gamma-1/2)}\widetilde \by_j/2,  \qquad\text{and}\qquad \bar \by_{2,j} \to \by^{ref}_{\bp_j} + \bar \by_j - \e^{2(\gamma-1/2)}\widetilde \by_j/2
\]
together with
\[
\bq'_{1,j}\to \bq_j + \e^{\gamma-1/2}\br_j/2, \qquad\text{and}\qquad \bq'_{2,j}\to \bq_j - \e^{\gamma-1/2}\br_j/2,
\]
Then the oscillatory factor in $\bq'_{1,j}-\bq'_{2,j}$ produces
\[
e^{i\omega_j(\bq'_{1,j}-\bq'_{2,j})\cdot\bar\by_j/\e^{\gamma-1/2}}
=
e^{i\omega_j \br_j\cdot\bar\by_j},
\]
so that integration in $\br_j$ yields
\[(2\pi)^2\delta(\bar \by_j)/\omega_j^2.\]
Therefore, the integrals in $(h_j,\bar s_j,\br_j,\bar\by_j)$ collapse onto $\bar s_j=0$ and $\bar\by_j=0$.

Gathering all these reductions and passing to the limit $\e\to 0$ yields
\begin{align*}
\lim_{\e\to 0}M_{\e}(n,n)  =  \frac{(2\pi)^{n} n!^2 \CalR^{2n}}{4^n}  \int_{\{0<\omega_{1}<\dots<\omega_{n}\}}  \int\cdots\int \prod_{j=1}^n &  e^{-i\omega_j \widetilde s_j}e^{i\omega_j \bp_j \cdot \widetilde\by_j}\\
& \times \E\big[ e^{ 2i\omega_j\bcals_0 ( V( \by'_j)- V(0))}\big] e^{- i \omega_j \by'_j \cdot \bp_{j} } \\
&\times  \vert \hat \Psi( \omega_j, \bq_{j}) \vert^2 \varphi^{1/2}(\widetilde s_j,\widetilde \by_j)\varphi^{1/2}(-\widetilde s_j,-\widetilde \by_j)\\
& \times \overline{\phi(s^{ref}_{\bp_j}, \by^{ref}_{\bp_j}, \omega_j,\bp_{j})} \overline{\phi(s^{ref}_{\bp_j},\by^{ref}_{\bp_j}, -\omega_j, \bp_{j})}\\
&\times  d\omega_j d\widetilde s_jd\widetilde \by_j d\bq_{j} d\by'_j d\bp_{j},
\end{align*}
and by symmetry of the $\omega_j$'s
\[
\lim_{\e\to 0} M_{\e}(n,n) =\frac{n!}{2^{n}} \Big(\frac{\CalR^{2}}{2}\int_0^\infty \CalJ^0(\omega) d\omega\Big)^n.
\]
Finally, using stationarity of $V$ and the symmetry of the mollifier $\varphi$ (so that
$\varphi^{1/2}(x)\varphi^{1/2}(-x)=\varphi(x)$)
\begin{align}\label{def:sigma2_ref}
\sigma^2_{ref}:=\frac{\CalR^{2}}{2}\int_0^\infty\CalJ^0(\omega) d\omega=\frac{\CalR^{2}}{4}\int\CalJ^0(\omega) d\omega
& =  \frac{(2\pi)^3\CalR^{2}}{4} \int\hat \varphi(\omega,\bp) \CalA(2\bcals_0, \omega,\bp) \vert \hat \Psi \vert^2_2(\omega)  \nonumber\\
&\hspace{2.5cm}\times \overline{\phi(s^{ref}_{\bp}, \by^{ref}_{\bp}, \omega,\bp)} \overline{\phi(s^{ref}_{\bp},\by^{ref}_{\bp}, -\omega, \bp)}  d\omega d\bp,
\end{align}
where $\vert \hat \Psi \vert^2_2$ is defined in \eqref{def:hPsi2} and $\CalA$ in \eqref{def:CalA}.

\paragraph{The case $l\neq n$.} In this situation, we   have 
\[
\lim_{\e \to 0} M_\e(l,2n-l) = 0,
\]
for any $l\neq n$, where $M_\e(l,2n-l)$ is defined by \eqref{def:Me} and
\[
M_\e(l,2n-l)=\int_0^\infty \cdots \int_0^\infty  \E\Big[\prod_{j=1}^l\CalJ^{\e}(\omega_{1,j})d \omega_{1,j} \prod_{j'=1}^{2n-l} \CalJ^{\e}(-\omega_{2,j'})d \omega_{2,j'}\Big].
\]

The key observation is that, when $l\neq n$, the number of positive-frequency factors $J^\e_+$ and negative-frequency factors $J^\e_-$ is unbalanced. Consequently, after ordering the frequencies and attempting to pair them as in the case
$l=n$, there necessarily remain at least $|l-n|$ unpaired factors. Each unpaired factor produces a rapidly oscillating phase of the form $e^{-i\omega_j s_{\bp_j}^{ref}/\e^{2\gamma-1}}$ which cannot be compensated by any conjugate phase. Since all frequencies $\omega$ are strictly positive, these oscillations persist after all possible changes of variables and do not cancel. As a result, the integrand contains highly oscillatory factors with frequency of order
$\e^{-(2\gamma-1)}$, and standard non-stationary phase arguments imply that
\[
\lim_{\e \to 0} M_\e(l,2n-l) = 0.
\]
Combining this result with the computation of the dominant contribution $M_\e(n,n)$ yields
\begin{equation}\label{eq:lim_moment}
\lim_{\e\to 0} \E\big[J^{\e}(\phi)^{2n}\big] = \frac{n!}{2^n} \sigma^{2n}_{ref} = \E[\big< \hat \CalS^{ref}_{\by}, \phi \big>^{2n}_{\CalS',\CalS}].
\end{equation}

\paragraph{Odd moments.} For any odd moment, there is necessarily an unbalanced number of $J^\e_+$ and $J^\e_-$ in $M_\e(l,2n+1-l)$, for any $l\in\{0,\dots,2n+1\}$. As in the case of even moments with $l\neq n$, not all the rapid phase factors can be compensated. As a consequence, the corresponding oscillatory integrals vanish in the limit $\e\to 0$, and we obtain
\[
\lim_{\e\to 0} \E\big[J^{\e}(\phi)^{2n+1}\big] = 0.
\]
The case of the expectation (moment of order $1$) is given by \eqref{eq:mean_Se_ref}. 

\paragraph{Tightness.} Tightness follows from the tightness of the real and imaginary parts of
$\hat \CalS^{\e,ref}_{\by}$ when applied to test functions
$\phi\in\CalS(\CalX\times\CalX)$ taking real values.
These tightness properties are obtained from the convergence of the second-order moment
of $J^{\e}$ applied to the test functions
\[
\phi_r(s,\by,\omega,\bp) = \frac{1}{2}(\phi(s,\by,\omega,\bp)+\phi(s,\by,-\omega,\bp))\qquad\text{and}\qquad \phi_i(s,\by,\omega,\bp)=\frac{1}{2i}(\phi(s,\by,\omega,\bp)-\phi(s,\by,-\omega,\bp)).
\] 
In fact, for these two test functions we have
\[
J^{\e}(\phi_r) = \big< Re(\hat \CalS^{ref}_{\by}), \phi \big>_{\CalS',\CalS}\qquad\text{and}\qquad
J^{\e}(\phi_i) =\big< Im(\hat \CalS^{ref}_{\by}), \phi \big>_{\CalS',\CalS}.
\]

\paragraph{Covariance formulas \eqref{eq:cov_1} and \eqref{eq:cov_2}.}
These formulas follow from \eqref{def:sigma2_ref} and \eqref{eq:lim_moment} for $n=1$, together with the standard polarization identities
\[
J^\e(\phi_1)J^\e(\phi_2) = \frac{1}{4}\Big((J^\e(\phi_1+\phi_2))^2-(J^\e(\phi_1-\phi_2))^2\Big),
\]
and
\[
J^\e(\phi_1)\overline{J^\e(\phi_2)} = \frac{1}{4i}\Big((J^\e(\phi_1+i\overline{\phi_2^-}))^2-(J^\e(\phi_1-i\overline{\phi_2^-}))^2\Big),
\]
where $\phi_2^-(\omega) = \phi_2(-\omega)$. This concludes the proof of Theorem \ref{th:speckle}.

\appendix

\section{Proof of the jump conditions and continuity relations} \label{sec:proof_continuity}

This section is devoted to the justification of the jump conditions across the plane $z=0$ induced by the source term, as well as to the continuity relations satisfied by the wavefield across the randomly perturbed interface.

To derive these relations, the solution $u$ of the wave equation
\eqref{eq:wave_equation} is decomposed as
\begin{equation}\label{eq:dec_u}
u(t,\bx,z) = u^-(t,\bx,z)\1_{(-\infty,0)}(z)+u^+(t,\bx,z)\1_{(0,\infty)}(z),
\end{equation} 
where $u^+$ satisfies 
\[
\Delta u^+ - \frac{1}{c^2(\bx,z)}\partial^2_{tt} u^+ = 0 \qquad (t,\bx,z)\in \R\times \R ^2 \times \R,
\]
and $u^-$ satisfies 
\[
\Delta u^- - \frac{1}{c^2_0}\partial^2_{tt} u^- = 0 \qquad (t,\bx,z)\in \R\times \R ^2 \times \R.
\]

\subsection{Jump conditions across the plan $z=0$ of the source location}

Substituting the decomposition \eqref{eq:dec_u} into the wave equation
\eqref{eq:wave_equation} yields, in the sense of distributions,
\begin{align*}
\Psi\Big(\frac{t-\bk_0\cdot \bx}{T_0},\frac{\bx}{r_0}\Big) \delta'(z)   = \Delta u - \frac{1}{c^2(\bx,z)}\partial^2_{tt} u  = \big(&\partial_z u^+(t,\bx,z=0)-\partial_z u^-(t,\bx,z=0)\big)\delta(z)\\
& + \big( u^+(t,\bx,z=0)- u^-(t,\bx,z=0)\big)\delta'(z).
\end{align*}
Identifying the coefficients of $\delta(z)$ and $\delta'(z)$, we obtain the jump conditions
\begin{align}\label{eq:jump_cond}
u(t,\bx,z=0^+)- u(t,\bx,z=0^-) & = \Psi\Big(\frac{t-\bk_0\cdot \bx}{T_0},\frac{\bx}{r_0}\Big),\\
\partial_z u(t,\bx,z=0^+)- \partial_z u(t,\bx,z=0^-) & = 0.\nonumber
\end{align}

\subsection{Continuity relations across the randomly perturbed interface} 

The continuity relations at the randomly perturbed interface are derived for $u^+$, and then extended to $u$, following the same approach as in the previous subsection. The term $u^+$ is decomposed as 
\begin{equation}\label{eq:dec_u+}
u^+(t,\bx,z) = u^+_0(t,\bx,z)\1_{(-\infty,0)}(z-z_{int}(\bx))+u^+_1(t,\bx,z)\1_{(0,\infty)}(z-z_{int}(\bx)),
\end{equation} 
where
\[z_{int}(\bx) := z_{int}+\sigma V(\bx/\ell_c).\]
The function $u^+_0$ satisfies
\[
\Delta u^+_0 - \frac{1}{c^2_0}\partial^2_{tt} u^+_0 = 0 \qquad (t,\bx)\in \R\times \R ^2\qquad\text{for}\qquad z<z_{int}(\bx),
\]
while $u^+_1$ satisfies
\[
\Delta u^+_1 - \frac{1}{c^2_1}\partial^2_{tt} u^+_1 = 0 \qquad (t,\bx)\in \R\times \R ^2\qquad\text{for}\qquad z>z_{int}(\bx).
\]
To flatten the interface, we introduce the shifted field
\[
U(t,\bx,Z) = u^+(t,\bx, Z + z_{int}(\bx))
\]
which satisfies
\begin{equation}\label{eq:wave_shift}
\Delta_\perp U + \Big(1+\frac{\sigma^2}{\ell_c^2} |\nabla_\perp V(\bx/\ell_c)|^2\Big)\partial^2_{ZZ} U - \frac{1}{c^2(\bx, Z+z_{int}(\bx))}\partial^2_{tt} U -\frac{\sigma}{\ell_c^2}(\Delta_\perp V(\bx/\ell_c))\partial_Z U - \frac{\sigma}{\ell_c} \nabla_\perp V(\bx/\ell_c) \cdot \nabla_\perp \partial_Z U  = 0.
\end{equation}
In the new variables, the decomposition \eqref{eq:dec_u+} becomes
\begin{equation}\label{eq:dec_u_shift}
U(t,\bx,Z) = U_0(t,\bx,Z)\1_{(-\infty,0)}(Z)+U_1(t,\bx,Z)\1_{(0,\infty)}(Z),
\end{equation}
where $U_0$ and $U_1$ satisfy
\[
 \Delta_\perp U_j +\Big(1+\frac{\sigma^2}{\ell_c^2} |\nabla_\perp V(\bx/\ell_c)|^2\Big)\partial^2_{ZZ} U_j - \frac{1}{c^2_j}\partial^2_{tt} U_j -\frac{\sigma}{\ell_c^2}(\Delta_\perp V(\bx/\ell_c))\partial_Z U_j - \frac{\sigma}{\ell_c} \nabla_\perp V(\bx/\ell_c) \cdot \nabla_\perp \partial_Z U_j  = 0,
\]
for $j=0,1$. Substituting the decomposition \eqref{eq:dec_u_shift} into
\eqref{eq:wave_shift} yields
\begin{align*}
0 & =\delta(Z) \Big(\Big(1+\frac{\sigma^2}{\ell_c^2} |\nabla_\perp V(\bx/\ell_c)|^2\Big)\Big(\partial_Z U_1(Z=0)-\partial_Z U_0(Z=0)\Big) \\
& \hspace{2cm}-\frac{\sigma}{\ell_c^2}(\Delta_\perp V(\bx/\ell_c)) (U_1(Z=0)-U_0(Z=0)) \\
& \hspace{2cm} -\frac{\sigma}{\ell_c}\nabla_\perp V(\bx/\ell_c)\cdot (\nabla_\perp U_1(Z=0)-\nabla_\perp U_0(Z=0))\Big)\\
& + \delta'(Z)\Big(1+\frac{\sigma^2}{\ell_c^2} |\nabla_\perp V(\bx/\ell_c)|^2\Big)\Big(U_1(Z=0)-U_0(Z=0)\Big)
\end{align*}
The term in $\delta'$ provides $U_1(Z=0)=U_0(Z=0)$, and then $\nabla_\perp U_1(Z=0)=\nabla_\perp U_0(Z=0)$, so that the term in $\delta$ provides only $\partial_Z U_1(Z=0) = \partial_Z U_0(Z=0)$. Going back to the original variables, we obtain the relations
\[
u^+_0(z=z_{int}(\bx)) = u^+_1(z=z_{int}(\bx))\qquad\text{and}\qquad \partial_z u^+_0(z=z_{int}(\bx)) = \partial_z u^+_1(z=z_{int}(\bx)),
\]
leading to the continuity relations for $u$
\[
u(z=z_{int}(\bx)^+) = u(z=z_{int}(\bx)^-)\qquad\text{and}\qquad \partial_z u(z=z_{int}(\bx)^+) = \partial_z u(z=z_{int}(\bx)^-).
\]

\section{Proof of Lemma \ref{lem:mixing}}\label{proof:lem_mix}

The proof is carried out by induction on $n$. We focus on the second assertion of the lemma, as the first one follows from the same arguments by considering functions $g_j$ of a single argument and setting $\by_j=0$. The relation is obvious for $n=1$. Assume that the statement holds for $n$ terms, and we prove it for $n+1$ terms. To this end, we write
\begin{align*}
\E[X^\eta_1 X^\eta_2] & = \E[X^\eta_1 X^\eta_2] - \E[X^\eta_1]\E[X^\eta_2] + \E[X^\eta_1]\E[X^\eta_2] \\
& = Corr(X^\eta_1, X^\eta_2) \sqrt{Var(X^\eta_1)Var(X^\eta_2)} + \E[X^\eta_1]\E[X^\eta_2], 
\end{align*}
with
\[X^\eta_1:= g_{n+1}\Big(V\Big(\frac{\bx_{n+1}}{\eta}+\frac{\by_{n+1}}{2}\Big),V\Big(\frac{\bx_{n+1}}{\eta}-\frac{\by_{n+1}}{2}\Big)\Big)\qquad\text{and}\qquad X^\eta_2:= \prod_{j=1}^n g_j\Big(V\Big(\frac{\bx_j}{\eta}+\frac{\by_j}{2}\Big),V\Big(\frac{\bx_j}{\eta}-\frac{\by_j}{2}\Big)\Big),\]
By the induction hypothesis, the desired factorization property holds for $X^\eta_2$. Therefore, we only have to prove that
\[
\lim_{\eta\to 0} Corr(X^\eta_1, X^\eta_2) = 0,
\]
since 
\[
Var(X^\eta_1)Var(X^\eta_2) \leq \sup_{j,v_1,v_2} |g_j(v_1,v_2)|^{2(n+1)}<\infty.
\]
Now, using \eqref{def:rho}, we obtain
\[
Corr(X^\eta_1, X^\eta_2)  \leq \rho\Big(\min_{\substack{j=1,\dots, n \\ \nu_1,\nu_2=\pm}}\Big\vert\frac{\bx_{n+1}-\bx_j}{\eta} +  \nu_1\frac{\by_{n+1}}{2} +\nu_2 \frac{\by_{j}}{2}\Big\vert\Big),
\]
with 
\[
\min_{\substack{j=1,\dots, n \\ \nu_1,\nu_2=\pm}}\Big\vert\frac{\bx_{n+1}-\bx_j}{\eta} +  \nu_1\frac{\by_{n+1}}{2} +\nu_2 \frac{\by_{j}}{2}\Big\vert \geq \min_{j=1,\dots, n}\frac{|\bx_{n+1}-\bx_j|}{\eta} -  \frac{\vert\by_{n+1}\vert}{2} - \max_{j=1,\dots,n} \frac{\vert\by_{j}\vert}{2} \underset{\eta\to 0}{\longrightarrow} \infty,
\]
since all the $\bx_j$'s pairwise distinct. Using the decay assumption \eqref{def:alpha_mix} together with the equivalence result \eqref{prop:equiv_mix}, we conclude that
\[
\lim_{\eta\to 0} Corr(X^\eta_1, X^\eta_2) = 0.
\]
Finally, we obtain
\begin{align*}
\lim_{\eta \to 0}\E\Big[\prod_{j=1}^{n+1} g_j\Big(V\Big(\frac{\bx_j}{\eta}+\frac{\by_j}{2}\Big),V\Big(\frac{\bx_j}{\eta}-\frac{\by_j}{2}\Big)\Big) \Big] & = \lim_{\eta \to 0}\prod_{j=1}^{n+1} \E\Big[ g_j\Big(V\Big(\frac{\bx_j}{\eta}+\frac{\by_j}{2}\Big),V\Big(\frac{\bx_j}{\eta}-\frac{\by_j}{2}\Big)\Big) \Big]\\
&=\prod_{j=1}^{n+1} \E\Big[ g_j\Big(V\Big(\frac{\by_j}{2}\Big),V\Big(-\frac{\by_j}{2}\Big)\Big) \Big],
\end{align*}
where the last equality follows from the stationarity of the random field $V$. This completes the proof of Lemma \ref{lem:mixing}.

\section{Proof of Lemma \ref{lem:tech}}\label{sec:proof_lem_tech}

The proof of this lemma relies on a sequence of transformations that make explicit the role of the interface fluctuations in the reflection and transmission processes. These transformations consist of:
(\emph{i}) taking the Fourier transform with respect to time and lateral variables, (\emph{ii}) rescaling the dual variables according to the source scaling, and (\emph{iii}) performing a leading-order expansion of the vertical wavenumbers.

We first take the Fourier transform of the continuity relations \eqref{eq:continuity} with respect to time and the transverse variables. We then introduce the change of variables $\bk \to \bk_0/\se + \bq$ corresponding to the scaling of the source profile \eqref{def:a0}. This change is applied to all Fourier variables appearing in \eqref{eq:continuity}. In addition, we use the expansion of the vertical wavenumbers given in \eqref{eq:exp_lambda}, namely
\begin{equation}\label{eq:approx_lambda}
\bcals^\e_j\Big(\bq+\frac{\bk_0}{\se}\Big) = \bcals_j +\CalO(\se) \qquad\text{for }j=0,1. 
\end{equation}
Finally, reverting to the original dual variable via $\bq \to \bk-\bk_0/\se$, we obtain, at leading order in $\e$, the following continuity relations at the random interface:
\begin{align}\label{eq:continuity_leading}
\frac{1}{\sqrt{\bcals_0}}\check u^{\e,inc}_\omega & (\bx,z_{int})e^{i\omega \bcals_0V(\bx/\e^\gamma)} + \frac{1}{\sqrt{\bcals_0 }}\check u^{\e,ref}_\omega(\bx, z_{int})e^{-i\omega \bcals_0V(\bx/\e^\gamma)} = \frac{1}{\sqrt{\bcals_1 }} \check u^{\e,tr}_\omega(\bx, z_{int})e^{i\omega \bcals_1 V(\bx/\e^\gamma)},\\
\sqrt{\bcals_0 } \check u^{\e,inc}_\omega & (\bx, z_{int}) e^{i\omega \bcals_0 V(\bx/\e^\gamma)} -\sqrt{ \bcals_0 } \check u^{\e,ref}_\omega(\bx, z_{int})e^{-i\omega \bcals_0 V(\bx/\e^\gamma)} = \sqrt{\bcals_1 } \check u^{\e,tr}_\omega(\bx, z_{int})e^{i\omega \bcals_1 V(\bx/\e^\gamma)},\nonumber
\end{align}
where the incident, reflected, and transmitted Fourier components are defined by
\begin{align*}
\check u^{\e,inc}_\omega(\bx,z) & := \frac{\omega^2}{(2\pi)^2 \e } \int e^{i\omega \bx\cdot\bk/\se} \frac{ \hau(\bk)}{\sqrt{\omega}} e^{i\omega\lwu(\bk)(z-z_{int})/\e} d\bk, \\
\check u^{\e,ref}_\omega(\bx,z) & := \frac{\omega^2}{(2\pi)^2 \e }  \int e^{i\omega \bx\cdot\bk/\se} \frac{\hbr(\bk)}{\sqrt{\omega}} e^{-i\omega\lwu(\bk)(z-z_{int})/\e} d\bk,\\
\check u^{\e,tr}_\omega(\bx,z) & := \frac{\omega^2}{(2\pi)^2 \e } \int e^{i\omega \bx\cdot\bk/\se} \frac{\hatr(\bk)}{\sqrt{\omega}} e^{i\omega\lwd(\bk)(z-z_{int})/\e} d\bk.
\end{align*}
The expansion \eqref{eq:approx_lambda} is crucial for identifying the leading-order expressions of the mode amplitudes $\hbr(\bk)$ and $\hatr(\bk)$. Solving the system \eqref{eq:continuity_leading} yields
\begin{align*}
\check u^{\e,ref}_\omega(\bx, z_{int}) & = \CalR \, \check u^{\e,inc}_\omega(\bx,z_{int})e^{2i\omega \bcals_0 V(\bx/\e^\gamma)},\\
\check u^{\e,tr}_\omega(\bx,z_{int}) & =  \CalT \, \check u^{\e,inc}_\omega(\bx,z_{int})e^{i\omega (\bcals_0 - \bcals_1)V(\bx/\e^\gamma)},
\end{align*}
giving  for the reflection and transmission amplitudes:
\begin{align*}
\hbr(\bk) &= \frac{\e \CalR \omega^2}{2(2\pi)^2} \, \iint e^{-i\omega \bx' \cdot (\bk-\bk') /\se} e^{i\omega\lwu(\bk')z_{int}/\e} e^{2i\omega \bcals_0 V(\bx'/\e^\gamma)} \\
&\hspace{3cm}\times\sqrt{\omega \lwu(\bk')} \hat \Psi\Big(\omega, \bk' - \frac{\bk_0}{\se}\Big) d\bx' d\bk',\\
\hatr(\bk)&= \frac{\e \CalT \omega^2}{2(2\pi)^2} \, \iint e^{-i\omega \bx' \cdot (\bk-\bk') /\se} e^{i\omega\lwu(\bk')z_{int}/\e} e^{i\omega (\bcals_0-\bcals_1) V(\bx'/\e^\gamma)} \\
&\hspace{3cm}\times\sqrt{\omega \lwu(\bk')} \hat \Psi\Big(\omega, \bk' - \frac{\bk_0}{\se}\Big) d\bx' d\bk'.
\end{align*}

To remove   the source profile scaling, we introduce the change of variables
\[  \bk \to \frac{\bk_0}{\se} + \bq \qquad\text{and}\qquad  \bk' \to \frac{\bk_0}{\se} + \bq'. \]
Using again \eqref{eq:approx_lambda}, we obtain at the leading order in $\e$ the reflected and transmitted wavefields:
\begin{align*}
u^{\e,ref}(t,\bx, z=0) & \simeq \frac{\CalR}{2(2\pi)^5 \e }\iiiint e^{-i\omega(t -2\bcals_0 z_{int}-\bx\cdot \bk_0 ) /\e}   e^{2i\omega  \bcals_0 V(\bx'/\e^\gamma)} \nonumber \\
& \hspace{3cm} \times e^{i \omega \bq \cdot (\bx - \bx' - \bx_{int})/\se} e^{i \omega \bq' \cdot (\bx' - \bx_{int})/\se} \\
& \hspace{3cm} \times \hat \CalU_0(\omega,\bq,z_{int}) \hat \CalU_0(\omega,\bq',z_{int}) \hat \Psi(\omega, \bq') \omega^4 d\omega d\bx' d\bq' d\bq,\nonumber
\end{align*}
and
\begin{align*}
u^{\e,tr}(t,\bx, z=z_{tr}) & \simeq \frac{\CalT}{2(2\pi)^5 \e } \sqrt{\frac{\bcals_0}{\bcals_1}} \iiiint e^{-i\omega(t -\bcals_0 z_{int}-\bcals_1(z_{tr}-z_{int}))-\bx\cdot \bk_0 ) /\e}   e^{i\omega  (\bcals_0-\bcals_1) V(\bx'/\e^\gamma)} \nonumber \\
& \hspace{4cm} \times e^{i \omega \bq \cdot (\bx - \bx' - \bx_{tr})/\se} e^{i \omega \bq' \cdot (\bx' - \bx_{int})/\se}  \\
& \hspace{4cm} \times \hat \CalU_1(\omega,\bq,z_{tr}-z_{int}) \hat \CalU_0(\omega,\bq',z_{int})\hat \Psi(\omega, \bq') \omega^4 d\omega d\bx' d\bq' d\bq.\nonumber
\end{align*}
These expressions can be rewritten in terms of the scattering operator $K^\e$ defined in \eqref{def:scat_op}, which completes the proof of the lemma.

\end{document}